%% file: main.tex
\RequirePackage[l2tabu,orthodox]{nag}
\documentclass
[11pt,letterpaper]
{article} 

\providecommand{\lambdamin}{\lambda_{\mathrm{min}}}
\newcommand\MYcurrentlabel{xxx}

\usepackage{etoolbox} %

\newcommand{\MYstore}[2]{%
  \global\expandafter \protected@xdef \csname MYMEMORY #1 \endcsname{#2}%
}
\newcommand{\MYload}[1]{%
  \csname MYMEMORY #1 \endcsname%
}
\newcommand{\MYnewlabel}[1]{%
  \renewcommand\MYcurrentlabel{#1}%
  \MYoldlabel{#1}%
}
\newcommand{\MYdummylabel}[1]{}
\newcommand{\torestate}[1]{%
  \let\MYoldlabel\label%
  \let\label\MYnewlabel%
  #1%
  \MYstore{\MYcurrentlabel}{#1}%
  \let\label\MYoldlabel%
}
\newcommand{\restatetheorem}[1]{%
  \let\MYoldlabel\label
  \let\label\MYdummylabel
  \begin{theorem*}[Restatement of \cref{#1}]
    \MYload{#1}
  \end{theorem*}
  \let\label\MYoldlabel
}
\newcommand{\restatecorollary}[1]{%
  \let\MYoldlabel\label
  \let\label\MYdummylabel
  \begin{corollary*}[Restatement of \cref{#1}]
    \MYload{#1}
  \end{corollary*}
  \let\label\MYoldlabel
}
\newcommand{\restatelemma}[1]{%
  \let\MYoldlabel\label
  \let\label\MYdummylabel
  \begin{lemma*}[Restatement of \cref{#1}]
    \MYload{#1}
  \end{lemma*}
  \let\label\MYoldlabel
}
\newcommand{\restateprop}[1]{%
  \let\MYoldlabel\label
  \let\label\MYdummylabel
  \begin{proposition*}[Restatement of \cref{#1}]
    \MYload{#1}
  \end{proposition*}
  \let\label\MYoldlabel
}
\newcommand{\restatefact}[1]{%
  \let\MYoldlabel\label
  \let\label\MYdummylabel
  \begin{fact*}[Restatement of \cref{#1}]
    \MYload{#1}
  \end{fact*}
  \let\label\MYoldlabel
}
\newcommand{\restatedefinition}[1]{%
  \let\MYoldlabel\label
  \let\label\MYdummylabel
  \begin{definition*}[Restatement of \cref{#1}]
    \MYload{#1}
  \end{definition*}
  \let\label\MYoldlabel
}
\newcommand{\restate}[1]{%
  \let\MYoldlabel\label
  \let\label\MYdummylabel
  \MYload{#1}
  \let\label\MYoldlabel
}

\usepackage[notes=true,later=false,camera=false]{dtrt}
\usepackage[utf8]{inputenc}
\usepackage{etex}
\usepackage{ stmaryrd }
\usepackage{xspace,enumerate}
\usepackage[T1]{fontenc}
\usepackage[full]{textcomp}
\usepackage[american]{babel}
\usepackage{mathtools}

\usepackage{amsthm}
\usepackage{empheq}

   \usepackage{hyperref}
\usepackage[capitalise,nameinlink]{cleveref}
\crefname{algocf}{Algorithm}{Algorithms}
\Crefname{algocf}{Algorithm}{Algorithms}

\newtheorem{theorem}{Theorem}[section]
\newtheorem{definition}[theorem]{Definition}

\newtheorem{lemma}[theorem]{Lemma}
\newtheorem{corollary}[theorem]{Corollary}
\newtheorem{claim}[theorem]{Claim}
\newtheorem{remark}[theorem]{Remark}
\newtheorem{fact}[theorem]{Fact}

\newtheorem*{theorem*}{Theorem}
\newtheorem*{lemma*}{Lemma}

\makeatletter
\def\thm@space@setup{%
  \thm@preskip=\parskip \thm@postskip=0pt
}
\makeatother

\usepackage[parfill]{parskip}

\usepackage{paralist}
\usepackage{turnstile}
\usepackage{mdframed}
\usepackage{tikz}
\usepackage{caption}
\usepackage{newfloat}

\usepackage[
letterpaper,
top=1.2in,
bottom=1.2in,
left=1in,
right=1in]{geometry}
\usepackage{newpxtext} %
\usepackage{textcomp} %
\usepackage[varg,bigdelims]{newpxmath}
\usepackage[scr=rsfso]{mathalfa}%
\usepackage{bm} %
\let\mathbb\varmathbb
\usepackage{microtype}

\allowdisplaybreaks

\newcommand{\Paren}[1]{\left(#1\right)}

\newcommand{\Brac}[1]{\left[#1\right]}

\newcommand{\abs}[1]{\lvert#1\rvert}
\newcommand{\Abs}[1]{\left\lvert#1\right\rvert}

\newcommand{\card}[1]{\lvert#1\rvert}
\newcommand{\Card}[1]{\left\lvert#1\right\rvert}

\newcommand{\set}[1]{\{#1\}}
\newcommand{\Set}[1]{\left\{#1\right\}}

\newcommand{\norm}[1]{\lVert#1\rVert}
\newcommand{\Norm}[1]{\left\lVert#1\right\rVert}

\newcommand{\iprod}[1]{\langle#1\rangle}
\newcommand{\Iprod}[1]{\left\langle#1\right\rangle}

\newcommand{\Esymb}{\mathbb{E}}
\newcommand{\Psymb}{\mathbb{P}}

\DeclareMathOperator*{\E}{\Esymb}
\DeclareMathOperator*{\ProbOp}{\Psymb}
\renewcommand{\Pr}{\ProbOp}

\newcommand{\Mid}{\nonscript\;\middle\vert\nonscript\;}

\DeclareMathOperator*{\argmax}{\arg\!\max}

\newcommand{\R}{{\mathbb R}}

\newcommand{\e}{\varepsilon}

\newcommand{\seteq}{\mathrel{\mathop:}=}

\newcommand{\1}{\mathbf{1}}

\newcommand{\Tr}{\mathrm{Tr}}

\newcommand{\Var}{\mathrm{Var}}

\newcommand{\poly}{\mathrm{poly}}

\newcommand{\mper}{\,.}
\newcommand{\mcom}{\,,}

\newcommand{\Id}{\mathbb{I}}

\newcommand{\rank}{\operatorname{rank}}

\newcommand{\cA}{\mathcal A}

\newcommand{\from}{\colon}

\usepackage{xspace}
\usepackage{comment}
\usepackage{dsfont}
\newcommand{\tcl}{$3$-coloring\xspace}
\newcommand{\kcl}{$k$-coloring\xspace}

\newcommand{\avg}{\operatorname{avg}}
\newcommand{\ALG}{\mathcal{A}}%

\makeatletter
\newcommand{\AVG}{%
  \@ifnextchar\bgroup{\AVG@i}{\operatorname{AVG}}%
}
\newcommand{\AVG@i}[1]{\AVG@ii{#1}}
\newcommand{\AVG@ii}[2]{\operatorname{AVG}[#1, #2]}
\makeatother

\makeatletter
\newcommand{\AVGT}{%
  \@ifnextchar\bgroup{\AVGT@i}{\operatorname{AVG}}%
}
\newcommand{\AVGT@i}[1]{\AVGT@ii{#1}}
\newcommand{\AVGT@ii}[2]{\operatorname{AVG}[#1, #2]}
\makeatother

\makeatletter
\newcommand{\nAVG}{%
  \@ifnextchar\bgroup{\nAVG@i}{\widetilde{\operatorname{AVG}}}%
}
\newcommand{\nAVG@i}[1]{\nAVG@ii{#1}}
\newcommand{\nAVG@ii}[2]{\widetilde{\operatorname{AVG}}[#1, #2]}
\makeatother

\makeatletter
\newcommand{\VAR}{%
  \@ifnextchar\bgroup{\VAR@i}{\operatorname{VAR}}%
}
\newcommand{\VAR@i}[1]{\VAR@ii{#1}}
\newcommand{\VAR@ii}[2]{\operatorname{VAR}[#1, #2]}
\makeatother

\makeatletter
\newcommand{\VART}{%
  \@ifnextchar\bgroup{\VART@i}{\operatorname{VAR}}%
}
\newcommand{\VART@i}[1]{\VART@ii{#1}}
\newcommand{\VART@ii}[2]{\operatorname{VAR}[#1, #2]}
\makeatother

\newcommand{\D}[2]{d_{#1#2}}

\renewcommand{\Id}{\mathbf{I}}
\newcommand{\Allones}{\mathbf{J}}

\newcommand{\btw}[2]{\Brac{#1\mathclose{:}#2}}
\newcommand{\whp}{w.h.p.\ }

\newcommand{\Ind}{\mathds{1}}

\usepackage{mdframed}
\usepackage[linesnumbered, ruled, vlined]{algorithm2e}

\usepackage{fancyhdr}

\begin{document}

\newcommand{\FormatAuthor}[3]{
\begin{tabular}{c}
#1 \\ {\small\texttt{#2}} \\ {\small #3}
\end{tabular}
}
\input{title}

\maketitle
\thispagestyle{empty}

\begin{abstract}
  \input{content/abstract}

\end{abstract}
\pagestyle{plain}
\setcounter{page}{1}

\clearpage
\tableofcontents{}
\thispagestyle{empty}
\clearpage

\input{content/intro}

\input{content/results}

\input{content/techniques}

\input{content/notations}

\input{content/spectral_result}
\input{content/k_partitioning}
\input{content/variance}
\input{content/algorithms}
\input{content/hardness}

\section*{Acknowledgements}
This project has received funding from the European Research Council (ERC) under the European Union’s Horizon 2020 research and innovation programme (grant agreement No 815464).

\bibliographystyle{alpha}
\bibliography{references}

\appendix
\input{content/threshold_rank}
\input{content/random_planting}

\end{document}

%% file: title.tex
\title{Finding Colorings in One-Sided Expanders %
}

\author{
\begin{tabular}[h!]{cc}
  \FormatAuthor{Rares-Darius Buhai}{rares.buhai@inf.ethz.ch}{ETH Zurich}
  &
  \FormatAuthor{Yiding Hua}{yiding.hua@inf.ethz.ch}{ETH Zurich}\\\\
  \FormatAuthor{David Steurer}{dsteurer@inf.ethz.ch}{ETH Zurich}
  &
  \FormatAuthor{Andor Vári-Kakas}{avarikakas@ethz.ch}{ETH Zurich}
\end{tabular}
} %
\date{\today}

%% file: content/abstract.tex
We establish new algorithmic guarantees with matching hardness results for coloring and independent set problems in one-sided expanders and related classes of graphs.
For example, given a \(3\)-colorable regular one-sided expander, we compute in polynomial time either an independent set of relative size at least \(\tfrac 12 - o(1)\)
or a proper \(3\)-coloring for all but an \(o(1)\) fraction of the vertices,
where \(o(1)\) stands for a function that tends to \(0\) with the second largest eigenvalue of the normalized adjacency matrix.
This result improves on recent seminal work of Bafna, Hsieh, and Kothari (STOC 2025) developing an algorithm that efficiently finds independent sets of relative size at least \(0.01\) in such graphs.
We also obtain an efficient \(1.6667\)-factor approximation algorithm for VERTEX COVER in sufficiently strong regular one-sided expanders, improving over a previous \((2-\e)\)-factor approximation in such graphs for an unspecified constant \(\e>0\).

We propose a new stratification of $k$-COLORING in terms of \(k\)-by-\(k\) matrices akin to predicate sets for constraint satisfaction problems.
We prove that whenever this matrix has repeated rows, the corresponding coloring problem is NP-hard for one-sided expanders under the Unique Games Conjecture.
On the other hand, if this matrix has no repeated rows, our algorithms can solve the corresponding coloring problem on one-sided expanders in polynomial time.
When this \(k\)-by-\(k\) matrix has repeated rows, we furthermore characterize the maximum fraction of vertices on which a proper $k$-coloring can be found by polynomial-time algorithms under the Unique Games Conjecture.

As starting point for our algorithmic results, we show a property of graph spectra that, to the best of our knowledge, has not been observed before:
The number of negative eigenvalues smaller than \(-\tau\) is at most \(O(1/\tau^{4})\) times the number of eigenvalues larger than \(\tau^{2}/2\).

While this result allows us to bound the number of eigenvalues bounded away from \(0\) in one-sided spectral expanders, this property alone is insufficient for our algorithmic results.
For example, given a one-sided regular expander with a balanced $3$-coloring, we can efficiently find a $3$-coloring for all but a \(o(1)\) fraction of vertices.
At the same time, if we only know that the graph has a balanced $3$-coloring and a bounded number of significant eigenvalues, it is NP-hard under the Unique Games Conjecture to find a $3$-coloring for all but a \(0.1\) fraction of vertices.

%% file: content/intro.tex
\section{Introduction}

Finding proper \(k\)-colorings in graphs for \(k\ge 3\) is among the first computational problems shown to be NP-hard \cite{MR4679190-Karp72}.
As a consequence, much research has focused on approximation algorithms for this problem and identifying assumptions on the input that allow us to solve the problem efficiently \cite{MR819128-83, MR1369207-94, MR1344544-Blum95, MR1439867-97, MR1484153-Alon97, MR1623197-98, MR1774844-00, MR1948749-01, MR2277147-06, MR2538841-09, MR2863244-Arora11, MR3536556-David16, MR3634492-17, kumar-louis-tulsiani, MR4764817-24}.

Most relevant to this paper, Blum and Spencer~\cite{MR1344544-Blum95} and Alon and Kahale~\cite{MR1484153-Alon97} studied an average-case setting in which a $k$-coloring is planted in a random graph,
and showed that polynomial-time algorithms can recover the planted $k$-coloring with high probability.
This result was subsequently derandomized by David and Feige~\cite{MR3536556-David16}:
They obtain the same guarantees if the $k$-coloring is planted \emph{at random} in a (two-sided) expander graph, where they require all but one eigenvalue of its normalized adjacency matrix to be close to \(0\).
Furthermore, they can recover a $k$-coloring for a $0.99$ fraction of the vertices even if the $k$-coloring is planted \emph{adversarially} in such a graph.

Kumar, Louis, and Tulsiani~\cite{kumar-louis-tulsiani} provide a different derandomization result:
given a graph that admits a $3$-coloring satisfying a particular pseudo-randomness condition, their algorithms finds a $3$-coloring for a $0.99$ fraction of the vertices in time $n^{O(t)}$, where $t$ is the number of eigenvalues of the graph smaller than $-0.51$.

\paragraph{Two-sided vs one-sided expansion.}

The results of Alon-Kahale and David-Feige crucially exploit the fact that both the positive and negative part of the spectrum of the graph are severely restricted.

Recent seminal work developed the first algorithms for one-sided expander graphs, where the negative part of the spectrum is completely unrestricted \cite{bafna-hsieh-kothari}.
The current work further develops our understanding of the computational landscape of coloring and independent set problems on one-sided expander graphs and provides answers to key questions left open by this prior work.

%% file: content/results.tex
\subsection*{Results}

Our results concern the computational tractability of the coloring and independent-set problems in one-sided expanders and related graphs.
While many combinatorial problems are computationally easy on expander graphs (e.g. \cite{arora-khot-kolla-steurer-tulsiani-vishnoi}),
the complexity landscape on these graphs reveals unexpected diversity.

We will restrict our discussion to regular graphs unless explicitly stated otherwise.
Many of our results extend in a straightforward way to non-regular graphs as well by measuring the sizes of sets in terms of the sum of vertex degrees.
(See the discussion in \cref{sec:techniques}).

\paragraph{Algorithms for independent sets and vertex covers in one-sided expanders and related graphs.}

For \(\lambda\ge 0\), we say that a graph \(G\) is a \(\lambda\)-one-side-expander if the second largest (nonnegative) eigenvalue of its normalized adjacency matrix is at most \(\lambda\).
For the coloring and independent set problems, it is important to leave the negative part of the spectrum of the adjacency matrix unrestricted.
Graphs on \(n\) vertices can have independent sets of size at least \(\alpha n\) only if they have at least one eigenvalue at most \(-\tfrac{\alpha}{1-\alpha}\).

\begin{theorem}[Large independent sets in one-sided expanders, see \Cref{cor:independentsetrecovery} for full version]
  For every \(\lambda>0\), there exists a polynomial-time algorithm that given a \(n\)-vertex regular \(\lambda\)-one-sided expander containing an independent set of size at least \((\tfrac 1 2 - \gamma) \cdot n\) for \(\gamma\in [0,\tfrac 1 {10}]\),
  outputs an independent set of size at least \(\Paren{\tfrac 12-(5+o_{\lambda\to 0}(1))\cdot \gamma}\cdot n\).
\end{theorem}

In the above theorem, we use \(o_{\lambda\to 0}(1)\) as a placeholder for an absolute function in \(\lambda\) that tends to \(0\) as \(\lambda\to 0\).
See \cref{cor:independentsetrecovery} for an explicit form of this function.

For sufficiently small \(\lambda\), we can simplify the lower bound on the independent set size found by the algorithm to \((\tfrac 12 -6\gamma)\cdot n\).
In particular, the fraction of vertices in the independent set tends to \(\tfrac 12\) for \(\gamma\to 0\).
Even in this setting, previous algorithms for large independent sets in one-sided expander only guarantee finding sets of size at least \(C\cdot n\) for some absolute constant \(C\) \cite{bafna-hsieh-kothari}.\footnote{The explicit bound is \(C\approx \tfrac 14 \cdot (\frac 1 {12}-\tfrac 1 {16})\approx 0.005\) \cite{bafna-hsieh-kothari}.}

Since the complement of independent sets are vertex covers, a direct corollary of the above theorem is a better-than-$2$ approximation algorithm for \textsc{vertex cover} in one-sided spectral expanders.

\begin{corollary}[Vertex cover approximation in one-sided expanders]
  There exist \(\lambda_{0}>0\) and a polynomial-time algorithm for \textsc{vertex cover} that on all regular \(\lambda_{0}\)-one-sided expanders achieves an approximation ratio at most \(1.6667\).
\end{corollary}

We remark that previous algorithms also imply a better-than-$2$ approximation ratio for \textsc{vertex cover} on one-sided spectral expanders but with an approximation ratio much closer to \(2\).

Furthermore, we remark that the above algorithms for independent sets and vertex covers in one-sided expanders extend to graphs with \emph{bounded threshold rank}.
These graphs may have more than one eigenvalue larger than \(\lambda\) but we require their number to be bounded by a constant independent of the size of the graph.

\paragraph{Coloring algorithms for one-sided expanders.}

The following theorem resolves one of the main questions left open by previous works on algorithms for 3-colorable one-sided expanders \cite{bafna-hsieh-kothari}.
For one-sided spectral expanders that have a balanced proper 3-coloring, we can find such a coloring for all but a tiny fraction of the vertices.

\begin{theorem}[Finding balanced 3-colorings in one-sided expanders, see \Cref{thm:3alg} for full version]
  For every \(\lambda>0\), there exists a polynomial-time algorithm that given an \(n\)-vertex regular \(\lambda\)-one-sided expander that admits a proper \(3\)-coloring with all color classes of size at most \((\tfrac 12 -o_{\lambda\to 0}(1))\cdot n\), outputs three independent sets that cover all but a \(o_{\lambda\to 0}(1)\) fraction of vertices.
\end{theorem}

As we will discuss in the next paragraph, the restriction on the sizes of the color classes is necessary assuming the Unique Games Conjecture:
Any algorithm that could solve the above problem when all of the color classes have size at most \(n/2\) could also solve problems shown to be NP-hard under the Unique Games Conjecture.

Since the above algorithm fails for \(3\)-colorable one-sided expanders only if one of the color classes contains almost half of the vertices, we can combine the above result with the previously discussed algorithm for finding independent sets that contain almost half of the vertices in one-sided expanders.

\begin{corollary}[Independent sets for 3-colorable one-sided expanders]
  For every \(\lambda>0\), there exists a polynomial-time algorithm given a \(3\)-colorable \(n\)-vertex regular \(\lambda\)-one-sided expander, outputs either one independent set of size at least \((\tfrac 1 2-o_{\lambda\to 0}(1))\cdot n\) or three independent sets that cover all but \(o_{\lambda\to 0}(1)\cdot n\) vertices.
\end{corollary}

Previous work on 3-colorable one-sided expanders are only guaranteed to find an independent set of size at least \(C\cdot n\) for a small absolute constant \(C>0\) even when the color classes have perfectly balanced sizes \cite{bafna-hsieh-kothari}.

\paragraph{Hardness results for coloring one-sided expanders and related graphs.}

As mentioned above, we obtain the following hardness result that matches the guarantees of our algorithms on one-sided expanders.
Our hardness results are formulated for almost 3-colorable graphs, which is a well-known artifact of \textsc{unique games} hardness reductions.
We expect these problems to remain intractable even for exactly 3-colorable graphs.
Moreover, we prove all of our algorithmic results for 3-colorable one-sided expanders also for the almost 3-colorable case.

\begin{theorem}[Hardness for almost 3-colorable one-sided expanders, see \Cref{thm: unbalancedhard} for full version]
  For every \(\e>0\), the following problem is NP-hard assuming the Unique Games Conjecture:
  Given a regular \(\e\)-one-sided expander with three independent sets that cover a \(1-\e\) fraction of its vertices,
  it is NP-hard to find a 3 independent sets that cover at least a \(\tfrac 12 +\e\) fraction of its vertices.
\end{theorem}

One might also expect that our algorithmic result for 3-colorable one-sided expanders extends to graphs that have a bounded number of significant eigenvalues.
However, this expectation is wrong in the sense that such algorithms would also refute the Unique Games Conjecture.

\begin{theorem}[Hardness for almost 3-colorable graphs with few significant eigenvalues, see \Cref{thm: lowtopthresholdrankhard} for full version]
  For every \(\e>0\), there exists a constant \(R\ge 1\) such that the following problem is NP-hard assuming the Unique Games Conjecture:
  Given a regular graph \(G\) with the following properties,
  \begin{enumerate}
  \item \(G\) has a balanced 3-coloring for all but an \(\e\) fraction of vertices, i.e., three disjoint independent sets of size at least \((1/3-\e)\cdot n\) each,
  \item \(G\) has at most \emph{two} (positive) eigenvalues larger than \(\e\),
  \item \(G\) has at most \(R\) eigenvalues larger than \(\e\) in absolute value,
  \end{enumerate}
  find three independent sets that cover at least a \(0.9\) fraction of the vertices.
\end{theorem}

The above theorem shows that even a single additional positive eigenvalue changes completely the tractability of graphs with (balanced) 3-colorings.

The above hardness result also refutes a prediction in previous work that hard coloring instances must have a large (super-constant) number of significant eigenvalues \cite{bafna-hsieh-kothari}.

\paragraph{Threshold-ranks of graphs.}

While the above hardness result shows that a small number of significant eigenvalues does not guarantee the tractability of finding good coloring,
our algorithmic results for one-sided expanders do rely on bounds on the number of these eigenvalues.
Indeed, we show the perhaps unexpected fact that one-sided expanders always have bounded number of significant (negative) eigenvalues.
More, generally we show the following relationship between the positive and negative spectrum of graphs.

\begin{theorem}[Relation between top and bottom threshold rank, see \Cref{thm:spectral_large_to_large} for full version]
  For all thresholds \(\tau>0\) and \(\tau'<\tau^{2}\) and every graph \(G\), the bottom-threshold rank of \(G\) at \(\tau\) is bounded by the top-threshold rank of \(G\) at \(\tau'\) as follows,
  \begin{equation}    
    \Paren{
      \tfrac{\tau^{2}-\tau'}{1-\tau'}}^2
    \cdot    \Card{ \sigma(G)\cap (-\infty,-\tau]
    }
    \le
 \Card{
\sigma(G)\cap [\tau',\infty)
    }
    \mper
  \end{equation}
  Here, \(\sigma(G)\) denotes the spectrum of the normalized adjacency matrix of \(G\) as a multiset of reals, taking into account the algebraic multiplicity of eigenvalues.
\end{theorem}

To the best of our knowledge, this basic result for spectra of graphs is new.

\paragraph{Algorithmic meta-theorem and dichotomy for coloring one-sided expanders.}

Prior work showed that \(3\) colors is the end of the line:
Already in balanced-\(4\)-colorable one-sided expanders, it is NP-hard under the Unique Games Conjecture to find an independent set with a constant fraction of vertices \cite{bafna-hsieh-kothari}.

Nevertheless, we generalize our results for \(3\)-colorable one-sided expanders to more than \(3\) colors by stratifying coloring instances in a particular way.
This stratification is akin to those for constraint satisfaction problems in terms of the set of allowed predicates, where every finite set of predicates defines a computational problem in its own right and we aim to characterize their computational complexity depending on the choice of predicates.

Our stratification of coloring problems on graphs \(G=(V,E)\) is in terms of \(k\)-by-\(k\) matrices \(M\).
Concretely, a reversible,\footnote{Here, \(M\) being reversible means that there exists a diagonal matrix \(D\) such that \(D M\) is a symmetric matrix.} row stochastic matrix \(M\in \R_{\ge 0}^{k\times k}\) with all zeroes on the diagonal and \(\delta>0\), we say that \(\chi\from V\to [k]\) is a \emph{\((M,\delta)\)-coloring} for \(G\) if \(\chi\) is proper for all but a \(\delta\) fraction of vertices and for all pairs of colors \(a,b\in [k]\),
\begin{equation}
  \Abs{M_{ab} - \Pr_{\bm x\bm y \sim G}\Set{\chi(\bm y)=b\Mid \chi(\bm x)=a}}
  \le \delta\mper
\end{equation}
Here, \(\bm x\bm y\sim G\) denotes a uniformly random edge such that \(\bm x\) and \(\bm y\) have the same marginal distribution.
We say that \(G\) is \((M,\delta)\)-\emph{colorable} if there exists a \((M,\delta)\)-coloring for \(G\).

Now we ask: given an \((M,\delta)\)-colorable one-sided expander for sufficiently small \(\delta\), can we efficiently find a proper \(k\) coloring for a large fraction of vertices?

Formally, we define an algorithm \(\cA\) to solve \(M\)-\textsc{expander-coloring} if for every \(\e>0\), there exists a \(\delta>0\) such that given an \((M,\delta)\)-colorable regular \(\delta\)-one-sided-expander \(G\), the algorithm \(\cA\) outputs $k$ independent sets that cover all but an \(\e\) fraction of vertices of \(G\).
In order for this problem to be non-trivial, we need the matrix \(M\) to have second largest eigenvalue at most \(0\).
In this way, \((M,\delta)\)-colorable \(\delta\)-one-sided-expanders exist for every \(\delta>0\).
In particular, a stochastic block model graph with large enough degree parameter (as a function of \(\delta\)) and suitable block probabilities (as a function of \(M\)) has the desired properties with high probability.

We show that, assuming the Unique Games Conjecture, \(M\)-\textsc{expander-coloring} admits a polynomial-time algorithm in the sense above if and only if \(M\) has no repeated rows.

\begin{theorem}[Algorithmic meta-theorem for coloring one-sided expanders, see \Cref{thm:kalg} and \Cref{thm: identicalrowshard} for full version]
\label{thm:alg-meta-thm}
  For every reversible, row stochastic matrix \(M\) with zeros on the diagonal and second largest eigenvalue at most \(0\),
  \begin{enumerate}
  \item if \(M\) has no repeated rows, then there exists a polynomial-time algorithm for \(M\)-\textsc{expander-coloring}, and
  \item if \(M\) has repeated rows, then \(M\)-\textsc{expander-coloring} is NP-hard under the Unique Games Conjecture.
  \end{enumerate}
\end{theorem}

The above meta-theorem can be readily applied to deduce new results. 
For example, for almost $2$-colorable graphs, we have trivially that any $2 \times 2$ row stochastic matrix with zeros on the diagonal must be reversible, have second largest eigenvalue at most $0$, and have distinct rows.
Then \Cref{thm:alg-meta-thm} implies that, for every $\epsilon > 0$, there exists a $\delta > 0$ such that given a regular $\delta$-one-sided expander graph that is $2$-colorable after the removal of a $\delta$-fraction of the vertices, a polynomial-time algorithm outputs two independent sets that cover all but an $\epsilon$ fraction of the graph vertices.

Furthermore, the meta-theorem provides a satisfying explanation for why the sizes of color classes matter for \(3\)-colorable one-sided expanders in terms of what is possible algorithmically.
For example, the following choice of \(M\) corresponds to perfectly balanced \(3\)-colorings in regular graphs,
\begin{equation}
\begin{bmatrix}
0 & \frac{1}{2} & \frac{1}{2} \\
\frac{1}{2} & 0 & \frac{1}{2} \\
\frac{1}{2} & \frac{1}{2} & 0
\end{bmatrix}
\mper
\end{equation}
Now, if we increase the size of the first color class while keeping the sizes of the other color classes the same, we obtain the following matrices for every \(\gamma\in [0,\tfrac 12]\),
\begin{equation}
\begin{bmatrix}
0 & \frac{1}{2} & \frac{1}{2} \\
\frac{1}{2}+\gamma & 0 & \frac{1}{2}-\gamma \\
\frac{1}{2}+\gamma & \frac{1}{2}-\gamma & 0
\end{bmatrix}
\mper
\end{equation}
The above matrix has three non-positive eigenvalues \(1,-\tfrac 12 +\gamma,-\tfrac 12 -\gamma\).
We can make this matrix symmetric by multiplying from the left the diagonal matrix with entries \(1,\tfrac{1/2}{1/2+\gamma},\tfrac{1/2}{1/2+\gamma}\).
Thus, the relative sizes of the color classes are \(\tfrac{1+2\gamma}{3+2\gamma},\frac1{3+2\gamma},\frac1{3+2\gamma}\).
Hence, for \(\gamma=\tfrac12\), we obtain relative color class sizes \(\tfrac 12,\tfrac 14,\tfrac 14\), but the matrix \(M\) degenerates and the two last rows are identical,
\begin{equation}
\begin{bmatrix}
0 & \frac{1}{2} & \frac{1}{2} \\
1 & 0 & 0 \\
1 & 0 & 0
\end{bmatrix}
\mper
\end{equation}
This matrix even provides an explanation of why it is possible to find an independent set with almost half of the vertices:
The first color class has relative size \(1/2\) and the corresponding row is not repeated.

In general, for $M$-colorable one-sided expanders, our algorithms allow us to approximately recover all color classes for which the corresponding rows in \(M\) are unique.
We also recover an optimal fraction (assuming the Unique Games Conjecture) of the color classes for which the corresponding rows \emph{are} repeated.
Formally, we define an algorithm \(\cA\) to solve \((M, \alpha)\)-\textsc{expander-coloring} if for every \(\e>0\), there exists a \(\delta>0\) such that given an \((M,\delta)\)-colorable regular \(\delta\)-one-sided expander \(G\), the algorithm \(\cA\) outputs $k$ independent sets that cover all but an \(\alpha + \e\) fraction of vertices of \(G\).

\begin{theorem}[Algorithmic approximation meta-theorem for coloring one-sided expanders, see \Cref{thm:kalg-partial} and \Cref{thm: finegrainedrowshard} for full version]
  For every reversible, row stochastic matrix \(M\) with zeros on the diagonal and second largest eigenvalue at most \(0\), there exists some $\alpha \in [0, 1]$ such that
  \begin{enumerate}
  \item there exists a polynomial-time algorithm for \((M, \alpha)\)-\textsc{expander-coloring}, and
  \item for every $\gamma > 0$, \((M, \alpha-\gamma)\)-\textsc{expander-coloring} is NP-hard under the Unique Games Conjecture.
  \end{enumerate}
  Furthermore, $\alpha$ is calculated as follows: Let $\pi = (\pi_1, \ldots, \pi_k)$ be the stationary distribution corresponding to $M$.
  Then, denoting the rows of $M$ as $M^{(1)}, \ldots, M^{(k)}$,
  \begin{equation}
    \label{eq:approx-alpha}
    \alpha = \sum_{a=1}^k \min\Paren{\pi_a, \sum_{\substack{b \neq a\\ M^{(b)} = M^{(a)}}} \pi_j}\,.
  \end{equation}
\end{theorem}

The error guarantee in \Cref{eq:approx-alpha} is to be interpreted as saying that, roughly, the number of misclassified vertices from color class $a$ (normalized by dividing by $n$) is bounded by the sum of the relative sizes of the other color classes with rows identical to $M^{(a)}$.
(This number of misclassified vertices is also trivially upper bounded by the relative size of color class $a$.)
As a consequence, for any maximal set of identical rows $M^{(a_1)}, \ldots, M^{(a_\ell)}$ of $M$, all the color classes other than the one with the largest relative size may be completely missed by the algorithm.
On the other hand, the number of correctly classified vertices from the color class with the largest relative size among $M^{(a_1)}, \ldots, M^{(a_\ell)}$ is non-trivial whenever its relative size is larger than the sum of the other relative sizes.

\paragraph{Concurrent work}
Concurrently and independently, Jun-Ting Hsieh \cite{tim-concurrent} gave an algorithm that, given an $n$-vertex regular graph that admits a proper $3$-coloring on all but $\delta n$ of its vertices and whose normalized adjacency matrix has at most $r$ eigenvalues larger than $\epsilon/100$, runs in time $n^{O(r/\epsilon^2)}$ and outputs a proper $3$-coloring on at least $(\frac{1}{2} - \epsilon - O(\delta))n$ of its vertices.

%% file: content/techniques.tex
\section{Techniques}
\label{sec:techniques}

To discuss the spectral graph techniques used by our coloring algorithms,
we introduce some convenient notation and terminology commonly used, e.g., in the literature on Markov chains.

Let \(G\) be an undirected, connected graph with vertex set \(V\) and let \(\bm x \bm y\) denote a uniformly random edge of \(G\) with vertices \(\bm x\) and \(\bm y\) as endpoints.
Since the graph is undirected, the random vertices \(\bm x\) and \(\bm y\) have the same distribution without loss of generality.
For a vertex subset \(S\subseteq V\), we denote by \(\mu(S)=\Pr\{\bm x\in S\}\) its probability under this distribution.
(For regular graphs, this probability is the fraction of vertices contained in \(S\).)

Let \(A\) denote the adjacency matrix for \(G\) with every row normalized to sum to \(1\).
This matrix acts as a linear operator on real-valued functions \(f\from V\to \R\),
\begin{equation}
  A f (x)
  = \E \Brac{f(\bm y) \Mid \bm x = x}
  \mper
\end{equation}
This operator is self-adjoint with respect to the following inner product for functions \(f,g\from V\to \R\),
\begin{equation}
  \iprod{f,g}
  \seteq \E f(\bm x) g(\bm x)\mper
\end{equation}
We denote the Euclidean norm with respect to this inner product by \(\norm{f}\seteq \iprod{f,f}^{1/2}=\Paren{\E f^{2}}^{1/2}\).

The linear operator \(A\) has all eigenvalues in the interval \([-1,1]\subseteq \R\) and largest eigenvalue \(1\).
Since the graph \(G\) is connected, the top eigenspace is one-dimensional and contains all constant functions on \(V\).
Let \(\lambda<1\) denote the second largest eigenvalue.
Then, the quadratic form of \(A\) satisfies the inequality,
\begin{equation}
  \iprod{f,Af}\le (1-\lambda)\cdot \Paren{\E f}^{2} + \lambda \cdot \E f^{2}
  \mper
\end{equation}
The \emph{spectral gap} \(1-\lambda\) provides a measure of well-connectedness of \(G\).
We denote the \emph{(edge) expansion} of a vertex subset \(S\) by \(\phi_{G}(S)\seteq \Pr\set{\bm y\not \in S\mid \bm x\in S}\) and the \emph{expansion} of \(G\) as \(\phi(G)\seteq \min_{S\colon \mu(S)\le 1/2 } \phi_{G}(S)\).
By Cheeger's inequality \cite{cheeger, alon-milman}, the expansion of \(G\) agrees up to a quadratic factor with the spectral gap,
\begin{equation}
  \phi(G)^{2}
  \lesssim 1-\lambda
  \lesssim \phi(G)
  \mper
\end{equation}
While the above inequality is most informative for \(\lambda\) close to \(1\), smaller values of \(\lambda\) have interesting combinatorial consequences, e.g., they imply stronger expansion bounds for small vertex subsets \(S\subseteq V\),
\begin{equation}
  \phi_{G}(S)
  = 1-\iprod{\Ind_{S},A \Ind_{S}} / \mu(S)
  \ge 1- (1-\lambda)\cdot \mu(S) - \lambda
  \mper
\end{equation}
Here, \(\Ind_{S}\from V\to \set{0,1}\) denotes the indicator function of \(S\) satisfying \(\E \Ind_{S}=\E\Ind_{S}^{2}=\mu(S)\).

\paragraph{Subspace methods for expansion.}

Several important combinatorial optimization problems are related to finding (small) non-expanding sets in graphs, e.g., \textsc{uniform sparsest cut} \cite{leighton-rao,arora-rao-vazirani}, \textsc{unique games} \cite{khot, arora-barak-steurer, raghavendra-steurer}, and \textsc{small-set expansion} \cite{raghavendra-steurer, raghavendra-steurer-tetali}.
By extending the above discussion, we see that all non-expanding sets must be close to the subspace spanned by eigenvectors with eigenvalue above a certain threshold.

Letting \(P^{\ge \tau}\) denote the orthogonal projector into the subspace spanned by eigenvectors with eigenvalue at least \(\tau\), we generalize the previous upper bound on the quadratic form \(\iprod{f,Af}\)
\begin{equation}
  \iprod{f,Af}
  \le (1-\tau)\cdot \norm{P^{\ge \tau}f}^{2} + \tau \E f^{2}
  \mper
\end{equation}
A concrete consequence of the above inequality is that the indicator \(\Ind_{S}\) of a subset \(S\subseteq V\) with expansion \(\phi_{G}(S)\le \e\)  has projection into the top-most eigenspaces at least \(\norm{P^{\ge 1-\eta}\Ind_{S}}^{2}\ge (1-\e/\eta)\norm{\Ind_{S}}^{2}\) because \(1-\e \le 1-\phi_{G}(S)=\iprod{\Ind_{S},A\Ind_{S}}/\mu(S)\le \eta\cdot \norm{P^{\ge 1-\eta}\Ind_{S}}^{2}/\norm{\Ind_{S}}^{2}+1-\eta\).

Indeed, many algorithmic tasks related to expansion become tractable for graphs that have only few eigenvalues above a certain threshold \cite{arora-khot-kolla-steurer-tulsiani-vishnoi,arora-barak-steurer,barak-raghavendra-steurer,guruswami-sinop} by enumerating a fine enough cover of the Euclidean unit ball of the corresponding subspace.
This observation drives worst-case subexponential-time approximation algorithms for \textsc{unique games} and several related problems \cite{arora-barak-steurer, steurer-d-games} and motivates the notion of \emph{(top-)threshold rank} for graphs,
\begin{equation}
  \rank_{\ge \tau}(G) = \rank P^{\ge \tau} = \Card{\Set{i \mid \lambda_{i}(A)\ge \tau}}
  \mper
\end{equation}
While small threshold rank enables enumeration-based algorithms for many problems, it turns out that large threshold rank implies the existence of small non-expanding sets \cite{arora-barak-steurer, steurer, lee-oveis-gharan-trevisan, raghavendra-vempala-tetali-louis}.
This combinatorial ramification of large threshold rank underlies the aforementioned subexponential-time approximation algorithms, e.g., for the construction of suitable graph decompositions.

\paragraph{Subspace methods beyond expansion?}

In this work, we aim to extend these developments to vertex-based problems like \textsc{coloring}, \textsc{vertex cover} and \textsc{independent set},
strengthening recent seminal work on rounding large independent sets in expanders \cite{bafna-hsieh-kothari}.

By Hoffman's bound \cite{hoffman}, an independent set \(S\) in \(G\) means that the Markov operator \(A\) for \(G\) has at least one eigenvalue at most \(-\tfrac{\mu(S)}{1-\mu(S)}\).
Indeed, the projection \(f=\Ind_{S}-\mu(S)\) of the indicator \(\Ind_{S}\) orthogonal to the top eigenspace of \(A\) satisfies \(\E f^{2}=\mu(S)\cdot (1-\mu(S))\) and
\begin{equation}
  \iprod{f,Af}
  = \iprod{\Ind_{S},A \Ind_{S}}-\mu(S)^{2}
  = -\tfrac{\mu(S)}{1-\mu(S)} \cdot \E f^{2}
  \mper
\end{equation}
In particular, for independent sets \(S\) with volume \(\tfrac {1-\e}2\), the above Rayleigh quotient for \(f\) is \(\tfrac{1-\e}{1+\e}\ge 1-2\e\).
By a similar argument as before for non-expanding sets, most of the function \(f\) is captured by the eigenspaces with eigenvalues at most \(-0.9\) so that \(\norm{P^{\le -0.9}f}\ge (1-O(\e))\cdot \norm{f}^{2}\),
where \(P^{\le -\tau}\) denotes the projector to the span of the eigenvectors of \(A\) with eigenvalue at most \(-\tau\).
Therefore, by enumerating a fine enough cover of the unit ball of the image of \(P^{\le -0.9}\), we can find a function \(O(\e)\)-close to \(f\).
With such a function in our hands, we can round it to a vertex subset \(O(\e)\)-close to \(S\).
After removing a vertex cover of the edges with both endpoints in that set, we obtain an independent set of volume at least \(\tfrac{1-O(\e)}{2}\).
In this way, we can find independent sets of volume close to \(\tfrac 12\) in time exponential in the \emph{bottom-threshold rank} \(\rank_{\le -0.9}(G)\), defined as
\begin{equation}
  \rank_{\le -\tau} (G)
  = \rank P^{\le -\tau}
  = \card{\set{i \mid \lambda_{i}(A)\le -\tau}}
  \mper
\end{equation}

\paragraph{Coloring graphs with bounded bottom-threshold rank?}

At this point, it is natural to expect efficient algorithms for finding independent sets of volume \(\tfrac 1 k\) or \(k\)-colorings for most vertices on graphs with small bottom-threshold rank.
However, we show in this work that this expectation is wrong and that such algorithms would refute the Unique Games Conjecture \cite{khot}:
\begin{quote}
  For all constants \(\tau,\e>0\), there exists a constant \(R\ge 1\)
  such that it is UG-hard\footnote{
    We say that a problem is UG-hard if there is a reduction from gapped \textsc{unique games} to that problem.
    In particular, it means that the problem is NP-hard assuming the Unique Games Conjecture.
  } to find a \(3\)-coloring for a \(0.9\) fraction of vertices in graphs \(G\) that have a balanced \(3\)-coloring for all but an \(\e\) fraction of vertices and that have bounded bottom-threshold \(\rank_{\le -\tau}(G)\le R\).
\end{quote}
(With the same quantifiers, it is UG-hard to find an independent set of volume \(\e\) in graphs that have a \(4\)-coloring for all but an \(\e\) fraction of vertices and have bounded bottom-threshold rank.)
These hardness results also refute an explicit prediction in previous work \cite{bafna-hsieh-kothari} that hardness reductions for coloring or independent set must produce instances with large (super-constant) bottom-threshold rank.

Underlying our hardness results is the issue that for independent sets \(S\) of volume bounded away from \(\tfrac 12\), we can guarantee only that its (centered) indicator function has non-trivial correlation \(\Theta(\mu(S))\) with the bottom most eigenspaces.
This amount of correlation turns out to be enough for solving many edge-based problems like edge expansion or 2-CSP.
Concretely, we could efficiently find colorings with only a tiny fraction of monochromatic edges in low-threshold-rank graphs \cite{barak-raghavendra-steurer}.
The computational intractability arises if we try to solve vertex-based problems in such graphs, e.g., finding colorings that are proper for a significant fraction of vertices.
(In other words, the monochromatic edges have a small vertex cover.)

A key technical contribution of our work identifies simple structural conditions on coloring instances to guarantee that the bottom-most eigenspaces capture enough information about colorings to enable efficient algorithms.
Our starting point is generalization of the notion of pseudorandom colorings \cite{kumar-louis-tulsiani}.
In order to understand what information we can extract from the bottom-most eigenspaces, we construct from the input graph \(G\) and its optimal \(k\)-coloring \(\chi\from V\to [k]\), a \emph{model graph} \(H\) whose eigenspaces are easy to understand and capture enough information to be able to reconstruct the coloring \(\chi\).
Then, we aim to show that \(H\) approximates \(G\) well-enough to guarantee that the bottom-most eigenspaces of \(G\) capture almost as much information about \(\chi\) as those of the model graph \(H\).

\paragraph{Recovering partitions from non-degenerate low-variance models.}
In order to specify our construction of the model \(H\) from the input graph \(G\) and any (not necessarily proper) \(k\)-coloring \(\chi\from V\to [k]\), we describe how to obtain the normalized adjacency matrix \(B\) for \(H\):
Partition the rows and columns of the normalized adjacency \(A\) according to the coloring \(\chi\) and replace in the resulting \(k\)-by-\(k\) block matrix every block by its average entry.
While one can view the model \(H\) as a graph on \(k\) vertices (corresponding to the blocks of the partition \(\chi\)), it will be more convenient to view \(H\) as a graph on the same vertex set as \(G\).
Concretely, the Markov operator \(B\) of the model \(H\) acts on functions \(f\from V\to \R\) as follows,
\begin{equation}
  Bf(x') = \E \Brac{f(\bm y') \Mid \chi(\bm x)=\chi(x')\,,~\chi(\bm y) = \chi(\bm y') }
  \mper
\end{equation}
Here, \(\bm x \bm y\) denotes a random edge of \(G\) as before and \(\bm y'\) is an independent random vertex with the same marginal distribution as \(\bm y\) (and \(\bm x\)).

The functions in the image of \(B\) are constant on the blocks of the partition \(\chi\).
Indeed, let \(\chi_{1},\ldots,\chi_{k}\) be the indicator functions of the blocks of \(\chi\) so that \(\chi_{a}=\Ind_{\chi^{-1}(a)}\).
Then, every function orthogonal to \(\chi_{1},\ldots,\chi_{k}\) is in the kernel of \(B\).
Therefore, every function in the image of \(B\) is a linear combination of \(\chi_{1},\ldots,\chi_{k}\).
However, the image of \(B\) may be a strict subspace of the span of these indicator functions and the indicator functions are not necesarily in the image of \(B\).

When can we recover the partition \(\chi\) from the image of \(B\)?
For a function \(f\from V\to \R\) in the image of \(B\), let \(f(c)\) denote the value that \(f\) assigns to block \(c\).
Then all of the following statements are equivalent:

\begin{enumerate}
\item The image of \(B\) uniquely determines the partition \(\chi\).
\item For every distinct pair of colors \(a,b\in [k]\), the image of \(B\) contains a function \(f\) that assigns different values to block \(a\) and block \(b\) so that \(f(a)\neq f(b)\).
\item For every distinct pair of colors \(a,b\in [k]\), the function \(\chi_{a}/\E\chi_{a}-\chi_{b}/\E\chi_{b}\) is \emph{not} orthogonal to the image of \(B\).
\item For every distinct pair of colors \(a,b\in [k]\), the function \(\chi_{a}/\E\chi_{a}-\chi_{b}/\E\chi_{b}\) is \emph{not} in the kernel of \(B\).
\item For every distinct pair of colors \(a,b\in [k]\), there exists a color \(c\in [k]\) such that
  \begin{equation}
    \Pr\Set{\chi(\bm y)=c \mid \chi(\bm x)=a}
    \neq \Pr\Set{\chi(\bm y)=c \mid \chi(\bm x)=b}
    \mper
  \end{equation}
\item The row stochastic matrix \(M\in\R_{\ge 0}^{k\times k}\) with entries \(M_{a,b}=\Pr\Set{\chi(\bm y)=b \mid \chi(\bm x)=a}\) has no repeated rows.
\end{enumerate}

Let us assume that the image of \(B\) uniquely determines the partition \(\chi\).
A priori, it is unclear how to exploit this property of the model graph algorithmically because we can construct it explicitly only if we have access to the partition \(\chi\).

Next we will discuss that a strong enough variance bound implies that the image of \(B\) is captured by eigenspaces of \(A\) with eigenvalue bounded away from \(0\).
In this way, we can simulate the enumeration of the image of \(B\) by enumerating subspaces defined in terms of eigenspaces of the input matrix \(A\).

Our notion of variance captures how well \(A\) approximates \(B\) acting on the indicator functions \(\chi_{1},\ldots,\chi_{k}\) of the partition \(\chi\).
Concretely, we define the variance of \(G\) with respect to the model \(H\) as
\begin{equation}
  \Var(G;H)\seteq \sum_{a=1}^{k} \norm{(A-B)\chi_{a}}^{2}/\E \chi_{a}
  \mper
\end{equation}
Suppose that \(B\) has rank \(r\le k\) and that \(h_{1},\ldots,h_{r}\) are orthonormal eigenfunctions of \(B\) with non-zero eigenvalues \(\lambda_{1}(B),\ldots,\lambda_{r}(B)\).
Let \(\lambdamin(B)=\min\Set{\abs{\lambda_{i}(B)} \mid i\in [r]}\) denote the smallest non-zero eigenvalue of \(B\) in absolute value.
Let \(P_{\eta}\) denote the projection operator to the subspace spanned by eigenfunctions of \(A\) with eigenvalue at most \(\lambdamin(B)-\eta\) in absolute value.
In particular, every eigenvalue of \(A\) captured by \(P_{\eta}\) has distance at least \(\eta\) from every non-zero eigenvalue of \(B\).
We aim to show that all projections \(\norm{P_{\eta} h_{i}}^{2}\) are small as it allows us to conclude that most of \(h_{i}\) is captured by eigenspaces of \(A\) with eigenvalue at least \(\lambdamin(B)-\eta\) in absolute value (which we may afford to search exhaustively).
Indeed, we can upper bound this projection in terms of \(\norm{(A-B)h_{i}}^{2}\),
\begin{align}
  \norm{(A-B)h_{i}}^{2}
  &= \norm{(A-\lambda_{i}(B)\cdot I)h_{i}}^{2}\\
  &= \sum_{j} \Paren{\lambda_{j}(A)-\lambda_{i}(B)}^{2} \cdot \iprod{g_{j},h_{i}}^{2}\\
  &\ge \eta^{2}\sum_{j\colon \abs{\lambda_{j}(A)-\lambda_{i}(B)}\ge \eta} \iprod{g_{j},h_{i}}^{2}\\
  &\ge \eta^{2} \norm{P_{\eta} h_{i}}^{2}
  \mper
\end{align}
Here, we use that \(h_{i}\) is an eigenfunction of \(B\) with eigenvalue \(\lambda_{i}(B)\) and we let \(\set{g_{j}}\) be a full eigenbasis for the Markov operator \(A\) of \(G\).
At the same time, we can upper bound \(\norm{(A-B)h_{i}}^{2}\) in terms of our notion of variance.
Let us extend the \(h_{1},\ldots,h_{r}\) to an orthonormal basis \(h_{1},\ldots,h_{k}\) of the span \(U\) of \(\chi_{1},\ldots,\chi_{k}\).
Then,
\begin{align}
  \sum_{i=1}^{r}\norm{(A-B)h_{i}}^{2}
  & \le   \sum_{i=1}^{k}\norm{(A-B)h_{i}}^{2} \\
  & =  \sum_{i=1}^{k}\iprod{h_{i},(A-B)^{2}h_{i}} \\
  & = \Tr \Paren{I_{U} (A-B)^{2} I_{U}}\\
  &  =   \sum_{i=1}^{k}\norm{(A-B)\chi_{a}}^{2}/\E \chi_{a} = \Var(G;H)
    \mper
\end{align}
Here, \(I_{U}\) denotes the projection operator to the subspace \(U\) and we use that both \(h_{1},\ldots,h_{k}\) and \(\tfrac {1}{(\E\chi_{1})^{1/2}}\chi_{1},\ldots,\tfrac {1}{(\E\chi_{k})^{1/2}}\chi_{k}\) are two orthonormal bases for the same subspace \(U\).
Putting the previous two inequalities, we obtain the following bound on the error we suffer by factoring out the subspace spanned by eigenvectors of \(A\) with eigenvalue close to \(0\) in absolute value.
\begin{equation}
  \sum_{i=1}^{r}\norm{P_{\eta} h_{i}}^{2} \le \eta^{-2}\cdot \Var(G; H)\mper
\end{equation}
We show by a concentration argument that for randomly planted colorings in worst-case host graphs, the resulting coloring instance has a model graph that uniquely determines the planted coloring and it has variance tending to \(0\) with growing average degree.
Perhaps surprisingly, we show that a crisp deterministic property of the input graph, namely one-sided spectral expansion, also implies strong bounds on our notion of variance.

\paragraph{One-sided spectral expansion implies low variance.}

For \(\lambda>0\), we say that a graph \(G\) is a \(\lambda\)-one-sided-expander if the second largest eigenvalue of its Markov operator is at most \(\lambda\).
We use the term one-sided expander to emphasize that we do not restrict the bottom eigenvalues smaller than \(0\).
As we discussed above, this distinction between positive and negative eigenvalues is crucial for coloring and independent set problems.
Indeed, a two-sided bound on the eigenvalues different from \(1\) would rule out the existence of large independent sets.

Consider a \(\lambda\)-one-sided expander \(G\) with a proper \(k\)-coloring \(\chi\from V\to [k]\).
Construct the model graph \(H\) for \(G\) and \(\chi\) as above.
We show that the variance of \(G\) with respect to \(H\) satisfies the bound
\begin{equation}
  \Var(G;H) \le \lambda \cdot k
  \mper
\end{equation}
As before, we let \(A\) and \(B\) be the Markov operators of \(G\) and \(H\) and we let \(\chi_{1},\ldots,\chi_{k}\from V\to \set{0,1}\) be the indicator functons of the blocks of coloring \(\chi\).
To bound the variance, we are to show that \(A\chi_{a}\) is close to \(B\chi_{a}\) for all colors \(a\in [k]\).
The function \(A\chi_{a}\) is equal to \(\psi_{a}(x)\seteq \Pr\Set{\chi(\bm y)=a \mid \bm x=x}\),
where \(\bm x \bm y\) denotes a uniformly random edge of \(G\) as before.
At the same time, in every block \(b\) of \(\chi\), the function \(B\chi_{a}\) is constant and equal to \(\Pr\Set{\chi(\bm y)=a\mid \chi(\bm x)=b}=\E\Brac{\psi_{a}(\bm x) \mid \chi(\bm x)=b}\).
Hence, we can express the contribution of \(\norm{(A-B)\chi_{a}}^{2}\) to \(\Var(G;H)\),
\begin{equation}
  \norm{(A-B)\chi_{a}}^{2}
  = \sum_{b=1}^{k}\Pr\Set{\chi(\bm x)=b}
  \cdot \Var\Brac{\psi_{a}(\bm x) \mid \chi(\bm x)=b}
  \mper
\end{equation}
If the variance of \(\psi_{a}\), conditioned on block \(b\), is large,
it implies that a subset of block \(b\) has an unusually large number of neighbors in block \(a\),
which contradicts the one-sided expansion property of \(G\).
Concretely, we claim that every pair of distinct colors \(a,b\in[k]\) satisfies \(\Var\Brac{\psi_{a}(\bm x) \mid \chi(\bm x)=b}\le \lambda \tfrac{\E \chi_{a}}{\E \chi_{b}}\).
Indeed, consider the function \(f=\chi_{b}\cdot (\psi_{a}-\alpha)+\lambda\cdot \chi_{a}\) for \(\alpha=\E\Brac{\psi_{a}(\bm x)\mid \chi(\bm x)=b}\).

The expectation of this function satisfies \(\E f = \lambda \E\chi_{a}\).
Since the functions \(\chi_{b}\cdot (\psi_{a}-\alpha)\) and \(\chi_{a}\) are orthogonal, we can compute the norm of \(f\) as,
\begin{equation}
  \E f^{2}= \E \chi_{b}\cdot(\psi_{a}-\alpha)^{2} + \lambda^{2}\E\chi_{a}
  \mper
\end{equation}

Since \(\chi\) is a proper coloring and there no monochromatic edges, we can compute the quadratic form of \(A\) evaluated at \(f\),
\begin{align}
  \iprod{f,Af}
  &=2\lambda \cdot\E \chi_{b} (\psi_{a}-\alpha) A\chi_{a} \\
  &=2\lambda \cdot\E \chi_{b} (\psi_{a}-\alpha) \psi_{a} \\
  &=2\lambda \cdot\E \chi_{b} (\psi_{a}-\alpha)^{2}
    \mper
\end{align}
Here, we use in the third step that we chose \(\alpha\) such that \(\E \chi_{b}(\psi_{a}-\alpha)=0\).
  
By the inequality for \(\lambda\)-one-sided-expanders,
\begin{align}
    2\lambda \E \chi_{b} \cdot(\psi_{a}-\alpha)^{2}
    & = \iprod{f,Af} \\
    & \le (1-\lambda) (\E f)^{2} + \lambda \E f^{2} \\
    & =   (1-\lambda) \lambda^{2} (\E \chi_{a})^{2} + \lambda \E \chi_{b} (\psi_{a}-\alpha)^{2} + \lambda^{3} \E \chi_{a}
\end{align}
By solving the inequality for the conditional variance \(\E \chi_{b}(\psi_{a}-\alpha)^{2}\),
we obtain the desired bound,
\begin{equation}
  \E\chi_{b}(\psi_{a}-\alpha)^{2}
  \le \lambda \E \chi_{a} \cdot \Paren{(1-\lambda)\E \chi_{a}+\lambda}\le \lambda \E \chi_{a}
  \mper
\end{equation}
We remark that this variance bound is robust and deteriorates gracefully in the presense of a small fraction of monochromatic edges.

\paragraph{Relating positive and negative graph eigenvalues.} 

The discussion so far focused on the question whether the bottom eigenspaces of the graphs Markov operator capture enough information about colorings or not.
However, in order to obtain efficient algorithms, we also need to bound the number of significant bottom eigenvalues so that we can afford to enumerate the subspace their eigenfunctions span.

To this end, we show that, perhaps surprisingly, bottom-threshold ranks can never be much larger than top-threshold ranks.
Concretely, for all thresholds \(\tau>0\) and \(\tau'<\tau^{2}\), we can upper bound the bottom-threshold rank in terms of the top-threshold rank,
\begin{equation}
  \rank_{\le -\tau}(G)\le \Paren{\tfrac{\tau^{2}-\tau'}{1-\tau'}}^{-2}\cdot \rank_{\ge \tau'}(G)
  \mper
\end{equation}
Since one-side spectral expansion and small-set (edge) expansion imply upper bounds on top-threshold ranks, the above inequality allows us to translate these properties also to upper bounds on bottom-threshold ranks.
In the extreme case \(\tau=1\), the bottom-threshold rank counts the number of bipartite connected compontents, which we can clearly upper bound by the total number of connected components, which is the top-threshold rank for \(\tau'=1\).

To illustrate the above inequality, consider orthonormal eigenfunctions \(g_{1},\ldots,g_{t}\) of the Markov operator \(A\) of \(G\) with eigenvalues at most \(-\tau\).
Define a vector embedding of the vertices \(g\from V\to \R^{t}\) with \(g(x)\seteq \tfrac 1{\sqrt t}(g_{1}(x),\ldots,g_{t}(x))\).
By orthonormality of the functions \(g_{1},\ldots,g_{t}\), this embedding satisfies \(\E\norm{g(\bm x)}^{2}=\tfrac 1 t\sum_{i=1}^{t}\norm{g_{i}}^{2}=1\) and \(\E\iprod{g(\bm x),g(\bm x')}=\tfrac 1 t^{2}\sum_{i,j}\iprod{g_{i},g_{j}}^{2}=\tfrac 1 t\), where \(\bm x'\) is an independent random vertex with the same marginal distribution as \(\bm x\).
Furthermore, since \(g_{1},\ldots,g_{t}\) are eigenfunctions with eigenvalue at most \(-\tau\), we have
\begin{equation}
  \E \iprod{g(\bm x),g(\bm y)}=\tfrac 1 t \sum_{i=1}^{t} \iprod{g_{i},A g_{i}} \le -\tau
  \mper
\end{equation}
Using the tensoring trick \cite{arora-khot-kolla-steurer-tulsiani-vishnoi, khot-vishnoi}, let us a construct a new vector embedding \(f\from V\to \R^{t^{2}}\)of the vertices
\begin{equation}
  f(x) \seteq \norm{g(x)} \cdot \bar g(x)^{\otimes 2}\mcom
\end{equation}
where \(\bar g(x)=\tfrac 1{\norm{g(x)}}g(x)\) denotes the normalized vector embedding.
By Cauchy--Schwarz,
\begin{align}
  \tau
  &\le \E \abs{\iprod{g(\bm x),g(\bm y)}}\\
  &=\E \norm{g(\bm x)}\cdot \norm{g(\bm y)}\cdot \abs{\iprod{\bar g(\bm x),\bar g(\bm y)}}\\
  &\le \Paren{\E \norm{g(\bm x)}\cdot \norm{g(\bm y)}}^{1/2} \cdot \Paren{\E \norm{g(\bm x)}\cdot \norm{g(\bm y)}\cdot \iprod{\bar g(\bm x),\bar g(\bm y)}^{2}}^{1/2}
  \\
  &\le  \Paren{\E \norm{g(\bm x)}\cdot \norm{g(\bm y)}\cdot \iprod{\bar g(\bm x),\bar g(\bm y)}^{2}}^{1/2}
  \\
  &=  \Paren{\E \iprod{f(\bm x),f(\bm y)}}^{1/2}
    \mper
\end{align}
At the same time, the vector embedding continues to satisfy \(\E \norm{f(\bm x)}^{2}=\E \norm{g(\bm x)}^{2}=1\) and
\begin{equation}
  \E \iprod{f(\bm x),f(\bm x')}^{2}=  \E \iprod{g(\bm x),g(\bm x')}^{2}\cdot \iprod{\bar g(\bm x),\bar g(\bm x')}^{2}
  \le  \E \iprod{g(\bm x),g(\bm x')}^{2}=\tfrac 1 t
  \mper
\end{equation}
By a characterization of the top-threshold rank in terms of vector embeddings \cite{barak-raghavendra-steurer}, this function \(f\) constitues a witness that for every \(\tau'<\tau^{2}\), the top-threshold rank is at least
\begin{equation}
  \rank_{\ge  \tau'}(G)\ge \Paren{\tfrac {\tau^{2}-\tau'}{1-\tau'}}^{2} \cdot t
  \mper
\end{equation}

\paragraph{Algorithmic meta theorem and dichotomy for coloring one-sided expanders.}

Putting together the results discussed so far, we obtain an appealing algorithmic meta theorem and dichotomy result for a particular class of coloring problems on one-sided spectral expanders.

Let \(G\) be a undirected, connected graph with vertex set \(V\).
For \(\e>0\), we say that \(\chi\from V\to[k]\) is an \emph{\((k,\e)\)-coloring} if \(\chi\) can be made a proper \(k\)-coloring by removing a vertex subset of volume at most \(\e\).
For a reversible,\footnote{Here, \(M\) being reversible means that there exists a diagonal matrix \(D\) such that \(D M\) is a symmetric matrix.} row stochastic matrix \(M\in \R_{\ge 0}^{k\times k}\) with all zeroes on the diagonal and \(\delta>0\), we say that \(\chi\from V\to [k]\) is a \emph{\((M,\delta)\)-coloring} if \(\chi\) is \((k,\delta)\)-coloring and the model graph for \(G\) with respect is \(\delta\)-close to \(M\) so that for all pairs of colors \(a,b\in [k]\),
\begin{equation}
  \Abs{M_{ab} - \Pr_{\bm x\bm y \sim G}\Set{\chi(\bm y)=b\Mid \chi(\bm x)=a}}
  \le \delta\mper
\end{equation}
We say that \(G\) is \((M,\delta)\)-\emph{colorable} if there exists a \((M,\delta)\)-coloring for \(G\).

We will use the matrix \(M\) in order to parametrize graph coloring instances on one-sided expanders.
To solve the problem, it will be enough to find any good enough coloring of the graph.
However, the matrix \(M\) will represent an additional promise on what colorings exist in the graph.
Also note that every \(k\)-coloring \(\chi\) is an \(M\)-coloring for some reversible, row-stochastic matrix \(M\) with zeroes on the diagonal.

Formally, we define an algorithm \(\cA\) to solve \(M\)-\textsc{expander-coloring} if for every \(\e>0\), there exists a \(\delta>0\) such that given a \((M,\delta)\)-colorable \(\delta\)-one-sided-expander \(G\), the algorithm \(\cA\) outputs a \((k,\e)\)-coloring for \(G\).
In order for this problem to be non-trivial, we need the matrix \(M\) to have second largest eigenvalue at most \(0\).
In this way, \((M,\delta)\)-colorable \(\delta\)-one-sided-expander exist for every \(\delta>0\).
In particular, a stochastic block model graph with large enough degree parameter (as a function of \(\delta\)) and suitable block probabilities (as a function of \(M\)) has the desired properties with high probability.

We show that, assuming the Unique Games Conjecture, \(M\)-\textsc{expander-coloring} admits a polynomial-time algorithm in the sense above if and only if \(M\) has no repeated rows.
More precisely, we show that for every reversible, row stochastic matrix \(M\) with zeros on the diagonal and second largest eigenvalue at most \(0\):

\begin{enumerate}
\item If \(M\) has no repeated rows, then there exists a polynomial-time for \(M\)-\textsc{expander-coloring}.
\item If \(M\) has repeated rows, then there exists a polynomial-time reduction from problems predicted to be NP-hard by the Unique Games Conjecture to \(M\)-\textsc{expander-coloring}.
\end{enumerate}

We remark that it is not important for our algorithms to know the underlying model matrix \(M\).
(However, by carefully enumerating all potential model matrices, one could also assume to know the underlying matrix \(M\).)

For the algorithmic part of our meta-theorem, we just need to combine the ideas discussed so far.
Let \(\e>0\).
Suppose that \(M\) has no repeated rows and that \(G\) is a \((M,\delta)\)-colorable \(\delta\)-one-side-expander for sufficiently small \(\delta>0\) (as a function of \(M\) and \(\e\)).
Let \(\chi\) be \((M,\delta)\)-coloring of \(G\).
We can derive that the matrix \(M'_{ab}=\Pr\Set{\chi(\bm y)=b\mid \chi(\bm x)=a}\) has no repeated rows because it is sufficiently close to the matrix \(M\), which has that property.
Thus, as discussed before, the partition \(\chi\) is uniquely determined by the image of the Markov operator \(B\) of the model graph \(H\) for \(G\) with respect to \(\chi\).
Since \(G\) is a \(\delta\)-one-sided-expander, we obtain an upper bound of \(\delta\cdot k\) on the variance \(\Var(G;H)\) of the input graph \(G\) with respect to the model graph \(H\).
As we saw before, this kind of variance bound implies that the bottom-most eigenspaces of the input graph are close to those of the model graph and therefore allow us to approximately recover the partition \(\chi\).
At the same time, we can upper bound the dimension of the span of those eigenspaces by bounding the bottom-threshold rank in terms of the top-threshold rank (equaling \(1\) for one-sided spectral expanders).
We conclude that we can approximately recover the partition \(\chi\) in polynomial time by exhaustively enumerating a fine enough cover of the subspace spanned by the bottom-most eigenspaces of \(G\).

For the hardness part of our meta-theorem, we generalize a previous hardness reduction for coloring one-sided expanders \cite{bafna-hsieh-kothari}.
The starting point of the reduction is the following problem shown to be NP-hard under the Unique Games Conjecture for every \(\e>0\) \cite{MR2648426-Bansal09}:
\begin{quote}
  Given a regular graph \(G\) that contains two disjoint independent sets, each of volume at least \(\tfrac 12-\e\), find an independent set of volume at least \(\e\).
\end{quote}
Let \(M\) be a reversible, row-stochastic \(k\)-by-\(k\) matrix with zeroes on the diagonal and second largest eigenvalue at most \(0\).
Let \(p_{1},\ldots,p_{k}\ge 0\) with \(\sum_{i=1}^{k} p_i = 1\) be the stationary measure for \(M\) so that \(p_{a}M_{ab}=p_{b}M_{ba}\).
Suppose that the first two rows of \(M\) are identical and that \(p_{1}\le p_{2}\).
The idea of the reduction is to create a graph \(G'\) on \(k\) blocks consistent with \(M\) and to embed the hard instance \(G\) within the first two blocks.
Concretely, our reduction outputs the following graph \(G'\), where we choose its number of vertices \(n\) such that \(\tfrac 12 \card{V(G)}=p_{1} n\).
Let \(V_{1,2}=V(G) \cup V'_{1,2}\), where \(V'_{1,2}\) consists of \((p_{2}-p_{1})\cdot n\) fresh vertices, and create disjoint blocks \(V_{3},\ldots,V_{k}\) of fresh vertices such that \(\card{V_{i}}=p_{i}\cdot n\) for all \(i\in \set{3,\ldots,k}\).
This construction ensures \(\card{V_{1,2}}=(p_{1}+p_{2})\cdot n\).
Choose an arbitary bipartition \(V_{1}\) and \(V_{2}\) of \(V_{1,2}\) such that \(\card {V_{1}}=p_{1}n\) and \(\card{V_{2}}=p_{2}n\).
For simplicity, we choose \(G'\) as a weighted graph with the following distribution \(\bm x\bm y\) over edges, where \(\delta>0\) is a paramter of the reduction:
\begin{enumerate}
\item With probability \(\delta>0\), choose \(\bm x \bm y\) as a random edge in the input graph \(G\) for the reduction.
\item With the remaining probability \(1-\delta\), choose two random pair of colors \(\bm a \bm b\) according to the probabilities \(p_{a,b}=p_{a}\cdot M_{a,b}\) and choose \(\bm x\) to be a uniformly random vertex in \(V_{\bm a}\) and \(\bm y\) to be uniformly random vertex in \(V_{\bm b}\).
\end{enumerate}
We remark that, by subsampling the edges of the graph \(G'\), we could obtain a sparse, unweighted graph that inherits enough of the properties of \(G'\) for the reduction to go through.

For sufficiently small \(\delta\), the model graph for \(G'\) with respect to the partition \(V_{1},\ldots,V_{k}\) is arbitrarily close to \(M\).
Since the second largest eigenvalue of \(M\) is at most \(0\) that also means that \(G'\) is an arbitrarily good one-sided spectral expander, i.e., the second largest eigenvalue of \(G'\) tends to \(0\) as \(\delta\to 0\).
If we find a \((k,p_{1})\)-coloring in \(G'\), then at least half of the vertices of \(G\) are covered by \(k\) independent sets and we obtain an independent set of volume at least \(\tfrac{1}{2k}\) for \(G\).
Finally, we need to argue that \(G'\) satisfies the promise of being \((M,\delta')\)-colorable in case that the input graph \(G\) contains two disjoint independent sets, each of volume at least \(1/2-\e\).
To this end, extend the two large independent sets in \(G\) to a bipartition \(V^{*}_{1}\) and \(V^{*}_{2}\) of \(V_{1,2}\) with \(p_{1}n\) and \(p_{2}n\) vertices, respectively.
Note that the edges with both endpoints in the same part of this bipartition have a vertex cover of size at most \(4\e p_{1} n\).
Here, we also use that the edge probability satisfies \(p_{1,2}=p_{1,1}=p_{2,2}=p_{2,1}=0\) because \(M\) has zeros on the diagonal and its first two rows are identical.
It follows that the \(k\)-partition \(V^{*}_{1},V^{*}_{2,}V_{3},\ldots,V_{k}\) is a \((k,4\e)\)-coloring of \(G'\).
Furthermore, the model graph for \(G'\) with respect to this partition tends to \(M\) as \(\delta\to 0\) because the first two rows of \(M\) are identical.
(Indeed, for that property it does not matter which bipartition of \(V_{1,2}\) we choose as long as we respect the two parts have the right cardinalities.)

%% file: content/notations.tex
\section{Definitions and notation}

For a graph $G$ with adjacency matrix $A$ and (diagonal) degree matrix $D$ with $D_{ii} = \deg(i)$, we say it is a $\lambda$-one-sided expander if $\lambda_2(D^{-1/2} A D^{-1/2}) \leq \lambda$.
We say a matrix $A \in \R_{\geq 0}^{n \times n}$ is \emph{row stochastic} if the sum of its entries on every row is $1$: for all $x \in [n]$, $\sum_{y=1}^n A_{xy} = 1$.
We say a matrix $A \in \R_{\geq 0}^{n \times n}$ is \emph{reversible row stochastic} if it is row stochastic and if there exists $\pi \in \R_{\geq 0}^{n}$ such that, for all $x, y \in [n]$, $\pi_x A_{xy} = \pi_y A_{yx}$.

Our main results will be stated in terms of $(k,\delta)$-colorings.

\begin{definition}[$(k,\delta)$-coloring]
    \label[definition]{def:kdelta}
    Let $G$ be an $n$-vertex graph.
    A $(k, \delta)$-coloring of $G$ is a list of $k$ disjoint independent sets in $G$ that cover $(1-\delta)n$ of the vertices of $G$.
\end{definition}

We define now the partition matrices and the model of a graph, which are defined with respect to a $k$-partition $\chi: [n] \to [k]$ of its vertices.

\begin{definition}[Partition matrices of a coloring]
\label[definition]{def:partition_Z}
Let $\chi: [n] \to [k]$ be a $k$-partition of $[n]$.
Then we define the following matrices:
\begin{itemize}
    \item The partition matrix $Z \in \set{0, 1}^{n \times k}$ has $Z_{x\chi(x)} = 1$ for all $x \in [n]$ and $0$ elsewhere.
    \item The partition size matrix $B \in \R^{k \times k}$ is the diagonal matrix $B = Z^\top Z$. Note that $B_{aa} = |\chi^{-1}(a)|$ for all $a \in [k]$.
\end{itemize}
\end{definition}

\begin{definition}[Model of a coloring]
    \label[definition]{def:model}
Let $A \in \R^{n \times n}$, and let $\chi: [n] \to [k]$ be a $k$-partition of $[n]$.
For $x \in [n]$ and $a \in [k]$, let $\D{x}{a} = \sum_{y \in \chi^{-1}(a)} A_{xy}$.
We define the \emph{model} $M(A, \chi) \in \R^{k \times k}$ to be the matrix that, for any $a, b \in [k]$, has entries $M(A,\chi)_{ab} = \frac{1}{|\chi^{-1}(a)|} \sum_{x \in \chi^{-1}(a)}\D{x}{b}$.
\end{definition}

\begin{remark}
Using the definitions in \cref{def:partition_Z}, we have equivalently that $M(A, \chi) = B^{-1} Z^\top A Z$.
\end{remark}

For a reversible row stochastic matrix $M \in \R_{\geq 0}^{k \times k}$, we define an $(M, \delta)$-coloring, which is a coloring with model close to $M$:

\begin{definition}[$(M, \delta)$-colorable graph]
    \label[definition]{def:model-coloring}
    Let $M \in \R^{k \times k}_{\geq 0}$ be a reversible row stochastic matrix with zeros on the diagonal.
    Let $G$ be an $n$-vertex graph with adjacency matrix $A$ and (diagonal) degree matrix $D$ with $D_{ii} = \deg(i)$, and let $\tilde{A} = D^{-1/2}AD^{-1/2}$.
    We say that $G$ is $(M, \delta)$-colorable if there exists a $k$-partition $\chi: [n] \to [k]$ such that $\Norm{M - M(\tilde{A}, \chi)}_{\max} \leq \delta$ and $G$ has a vertex cover of all edges with both endpoints in the same part of size at most $\delta n$ with respect to $\chi$.
\end{definition}

%% file: content/spectral_result.tex
\section{Small top threshold rank implies small bottom threshold rank}

In this section, we prove the following result which shows that, for bounded-norm symmetric matrices with non-negative entries, large bottom threshold rank implies large top threshold rank.
In the rest of the paper, we will be applying this result to normalized adjacency matrices of graphs, which indeed are symmetric, have spectral norm $1$, and have non-negative entries.

\begin{theorem}[Large bottom rank implies large top rank]
\label[theorem]{thm:spectral_large_to_large}
    Let $A \in \R^{n \times n}_{\geq 0}$ be symmetric with $\Norm{A} \leq 1$.
    If
    \begin{equation*}
        \rank_{\leq -\lambda} (A) \geq t\,,
    \end{equation*}
    for some $\lambda \geq 0$ and some positive integer $t$, then for any $\sigma \in (0, 1)$ it follows that
    \begin{equation*}
        \rank_{\geq \frac{\lambda^2-\sigma}{1-\sigma}} (A) \geq \sigma^2 t \,.
    \end{equation*}
\end{theorem}

By taking the contrapositive of \cref{thm:spectral_large_to_large} and re-parameterizing, we get the following direct corollary that shows, for bounded-norm symmetric matrices with non-negative entries and small top threshold rank, that the bottom threshold rank cannot be too large.

\begin{corollary}[Small top rank implies small bottom rank]
\label[corollary]{cor:spectral_small_to_small}
    Let $A \in \R^{n \times n}_{\geq 0}$ be symmetric with $\Norm{A} \leq 1$.
    If
    \begin{equation*}
        \rank_{\geq \tau}(A) \leq s\,,
    \end{equation*}
    for some $\tau \geq 0$ and some non-negative integer $s$, then for any $\sigma \in (0, 1)$ it follows that
    \begin{equation*}
         \rank_{\leq -\sqrt{\tau(1-\sigma) + \sigma}} (A) \leq \frac{s}{\sigma^2}\,.
    \end{equation*}
\end{corollary}

To prove \cref{thm:spectral_large_to_large}, we need the following slight generalization of Lemma 6.1 in \cite{barak-raghavendra-steurer}.
We include a proof of this lemma in \cref{sec:rank-lb} for completeness.
\begin{lemma}[Generalized restatement of Lemma 6.1 in \cite{barak-raghavendra-steurer}]
    \label[lemma]{lemma:BRS11_lower_bound_top_threshold_rank}
    Let $A \in \R^{n \times n}$ be symmetric with $\Norm{A} \leq 1$.
    Suppose there exists $M \in \R^{n \times n}$ positive semidefinite such that
    \begin{equation*}
        \iprod{A, M} \geq 1-\e \,, \quad
        \Norm{M}_F^2 \leq \frac{1}{r} \,, \quad
        \Tr(M) = 1 \,.
    \end{equation*}
    Then for all $C > 1$,
    \begin{equation*}
        \rank_{\geq 1- C \cdot \e} (A) \geq \Paren{1-\frac{1}{C}}^2 r \,.
    \end{equation*}
\end{lemma}

The key observation that leads to the proof of \cref{thm:spectral_large_to_large} is that, if a bounded-norm symmetric matrix with non-negative entries has large bottom threshold rank, then we can construct a matrix $M$ that satisfies \cref{lemma:BRS11_lower_bound_top_threshold_rank}.
This idea is captured by the following lemma.

\begin{lemma}[Construction of witness matrix]
\label[lemma]{lemma:spectral_construction_of_vectors}
    Let $A \in \R^{n \times n}_{\geq 0}$ be symmetric with $\Norm{A} \leq 1$.
    If
    \begin{equation*}
        \rank_{\leq -\lambda} (A) \geq t\,,
    \end{equation*}
    for some $\lambda \geq 0$ and some positive integer $t$, then there exist $V \in \R^{t^2 \times n}$ such that
    \begin{equation*}
        \iprod{A, V^{\top} V} \geq \lambda^2 \,, \quad
        \Norm{V^{\top} V}_F^2 \leq \frac{1}{t} \,, \quad
        \Tr(V^{\top} V) = 1 \,.
    \end{equation*}
\end{lemma}

\begin{proof}
    Let $\set{u_1, u_2, \dots, u_t}$ be $t$ orthonormal eigenvectors of $A$ whose corresponding eigenvalues are at most $-\lambda$, that is, for each $s \in [t]$,
    \begin{equation*}
        \iprod{u_s, A u_s} \leq - \lambda \quad \text{and} \quad
        \norm{u_s}^2 = 1 \,.
    \end{equation*}
    Let $U \in \R^{n \times t}$ be the matrix whose $s$-th column is $\frac{1}{\sqrt{t}} u_s$ for all $s \in [t]$.
    It follows that
    \begin{equation}
    \label{eq:spectral_construction_1}
        \iprod{A, U U^{\top}}
        = \sum_{s \in [t]} \frac{1}{t} \iprod{A, u_s (u_s)^{\top}}
        = \frac{1}{t} \sum_{s \in [t]} \iprod{u_s, A u_s}
        \leq - \lambda \,,
    \end{equation}
    and
    \begin{equation}
        \label{eq:spectral_construction_2}
            \Norm{U U^{\top}}_F^2
            = \Norm{U^{\top} U}_F^2
            = \sum_{s \in [t]} \frac{1}{t^2} \norm{u_s}^4
            = \frac{1}{t} \,,
    \end{equation}
    and
    \begin{equation}
    \label{eq:spectral_construction_3}
        \Tr(U U^{\top})
        = \Tr(U^{\top} U)
        = \sum_{s \in [t]} \frac{1}{t} \norm{u_s}^2
        = 1 \,.
    \end{equation}
    
    Let us denote the $i$-th row of $U$ by $w_i$.
    For $i \in [n]$, define $v_i = \frac{w_i^{\otimes 2}}{\norm{w_i}}$ if $\norm{w_i} \neq 0$ and $v_i = 0$ (the all-zero vector) otherwise. %
    Let $S = \Set{i \in [n]: \norm{w_i} \neq 0}$.
    Finally, let $V \in \R^{t^2 \times n}$ be the matrix whose $i$-th column is $v_i$ for all $i \in [n]$.
    We will show that $V$ satisfies the three conditions of the lemma.
    
    \paragraph{Condition 1.}
    Notice that, by Cauchy-Schwarz, we have
    \begin{align*}
        \Paren{\sum_{i, j \in S} A_{i j} \iprod{w_i, w_j}}^2
        & = \Paren{\sum_{i, j \in S} \frac{ \sqrt{A_{i j}} \cdot \iprod{w_i, w_j}}{\sqrt{\norm{w_i} \norm{w_j}}} \cdot \sqrt{A_{i j} \norm{w_i} \norm{w_j}}}^2 \\
        & \leq \Paren{\sum_{i, j \in S} A_{i j} \frac{\iprod{w_i, w_j}^2}{\norm{w_i} \norm{w_j}}} \cdot \Paren{\sum_{i, j \in S} A_{i j} \norm{w_i} \norm{w_j}} \,.
    \end{align*}
    By rearranging, we get
    \begin{equation}
    \label{eq:spectral_construction_4}
        \iprod{A, V^{\top} V}
        = \sum_{i, j} A_{i j} \iprod{v_i, v_j}
        = \sum_{i, j \in S} A_{i j} \frac{\iprod{w_i, w_j}^2}{\norm{w_i} \norm{w_j}} \\
        \geq \frac{\Paren{\sum_{i, j \in S} A_{i j} \iprod{w_i, w_j}}^2}{\sum_{i, j \in S} A_{i j} \norm{w_i} \norm{w_j}} \,.
    \end{equation}
    For the numerator, by observing that $\iprod{w_i, w_j} = (U U^{\top})_{i j}$ and by \cref{eq:spectral_construction_1}, we get
    \begin{equation}
    \label{eq:spectral_construction_5}
        \Paren{\sum_{i, j \in S} A_{i j} \iprod{w_i, w_j}}^2
        = \Iprod{A, U U^{\top}}^2
        \geq \lambda^2 \,.
    \end{equation}
    For the denominator, because $\Norm{A} \leq 1$ we get
    \begin{align}
    \label{eq:spectral_construction_6}
    \begin{split}
        \sum_{i, j \in S} A_{i j} \norm{w_i} \norm{w_j}
        &\leq  \sum_{i \in S} \norm{w_i}^2
        = \Tr(U U^{\top})
        = 1 \,.
    \end{split}
    \end{align}
    Plugging \cref{eq:spectral_construction_5} and \cref{eq:spectral_construction_6} into \cref{eq:spectral_construction_4}, we get that $V$ satisfies condition 1,
    \begin{equation*}
        \iprod{A, V^{\top} V} \geq \lambda^2 \,.
    \end{equation*}

    \paragraph{Condition 2.}
    By the definition of $V$ and by Cauchy-Schwarz, we get
    \begin{equation*}
        \Norm{V^{\top} V}_F^2
        = \sum_{i, j} \iprod{v_i, v_j}^2
        = \sum_{i, j \in S} \frac{\iprod{w_i, w_j}^4}{\norm{w_i}^2 \norm{w_j}^2}
        \leq \sum_{i, j \in S} \iprod{w_i, w_j}^2 \,.
    \end{equation*}
    Since $\iprod{w_i, w_j} = (U U^{\top})_{i j}$, it follows by \cref{eq:spectral_construction_2} that
    \begin{equation*}
        \Norm{V^{\top} V}_F^2
        \leq \sum_{i, j \in S} \iprod{w_i, w_j}^2
        = \Norm{U U^{\top}}_F^2
        = \frac{1}{t}\,.
    \end{equation*}

    \paragraph{Condition 3.}
    By the definition of $V$ and \cref{eq:spectral_construction_3}, it follows that
    \begin{equation*}
        \Tr(V^{\top} V)
        = \sum_{i} \norm{v_i}^2
        = \sum_{i} \norm{w_i}^2
        = \Tr (U U^{\top})
        = 1 \,.
    \end{equation*}
\end{proof}

Given \cref{lemma:BRS11_lower_bound_top_threshold_rank} and \cref{lemma:spectral_construction_of_vectors}, we can prove \cref{thm:spectral_large_to_large} easily.
\begin{proof}[Proof of \cref{thm:spectral_large_to_large}]
    By \cref{lemma:spectral_construction_of_vectors}, when $\rank_{\leq -\lambda} (A) \geq t$, there exist $V \in \R^{t^2 \times n}$ such that
    \begin{equation*}
        \iprod{A, V^{\top} V} \geq \lambda^2 \,, \quad
        \Norm{V^{\top} V}_F^2 \leq \frac{1}{t} \,, \quad
        \Tr(V^{\top} V) = 1 \,.
    \end{equation*}
    Then, we can apply \cref{lemma:BRS11_lower_bound_top_threshold_rank} by setting $M=V^{\top} V$, $\e = 1-\lambda^2$, $r=t$, and $C=\frac{1}{1-\sigma}$, and obtain
    \begin{equation*}
        \rank_{\geq \frac{\lambda^2-\sigma}{1-\sigma}} (A) \geq \sigma^2 t \,.
    \end{equation*}
\end{proof}

%% file: content/k_partitioning.tex
\section{Recovering partitions}
\label{sec:k-partitioning}

In this section we state and prove our meta-theorem for recovering partitions.
The result in this section is written in a very general form, to make clear what properties we exploit in the proof.
Then, in \Cref{sec:algorithms}, we apply this meta-theorem to our settings of interest and obtain our algorithmic results.

\begin{theorem}[Result for recovering partitions]
\label{thm:part-recovery}
Let $\chi: [n] \to [k]$ be a $k$-partition of $[n]$, and consider its partition matrices $Z \in \{0,1\}^{n \times k}$ and $B \in \R^{k \times k}$ as per \cref{def:partition_Z}.
Suppose $|\chi^{-1}(a)| \geq c n$ for all $a \in [k]$.
Let $M \in \R^{k \times k}$ be a matrix such that 
\begin{itemize}
    \item $M B^{-1}$ is symmetric,
    \item $\Norm{M} \leq \zeta$,
    \item $\Norm{M^{(a)} - M^{(b)}} \geq \alpha$ for any two rows $M^{(a)}$ and $M^{(b)}$ of $M$.
\end{itemize}
Also let $W = Z M B^{-1} Z^\top$ and let $A \in \R^{n \times n}$ be some matrix such that, for all $a \in [k]$,
\[\Norm{A \frac{Z_a}{\Norm{Z_a}} - W \frac{Z_a}{\Norm{Z_a}}}^2 \leq \delta\,,\]
and suppose that $\rank_{\leq -\lambda}(A) + \rank_{\geq \lambda}(A) \leq t$ for some $\lambda > 0$ such that $\lambda \leq \alpha \sqrt{c} / 4$ and $\lambda \leq \zeta/4$.

Then, given $A$, there exists an algorithm that runs in time $n^{O(1)} \cdot \Paren{O_{\alpha, \delta, \zeta, c, k}(1)}^t$ and outputs a list of $\Paren{O_{\alpha, \delta, \zeta, c, k, t}(1)}^t$ $k$-partitions $\hat{\chi}: [n] \to [k]$ such that, for at least one of them, it holds up to a permutation of the color classes that $\mathbb{E}_{\bm{x} \sim \operatorname{Unif}([n])} \1[\hat\chi(\bm{x}) \neq \chi(\bm{x})] \leq O\Paren{\frac{\delta \zeta^2 k^3}{\min\{\alpha^2 c, \zeta^2\} \alpha^2 c}}$.
\end{theorem}

In our applications, we will take $M$ to be a model of the coloring $\chi$, as per \Cref{def:model}.
Our main results for graph coloring, in particular, take $M = M(\tilde{A}, \chi)$, where $\tilde{A}$ is the normalized adjacency matrix of the graph.
Then $M_{ab}$ corresponds to the transition probability from color class $a$ to color class $b$.
(An exception is our result for randomly planted graph colorings in \Cref{thm: random-planting-partial}, where for technical reasons we take $M = \frac{1}{d'} A$, where $A$ is the adjacency matrix of the graph and $d'$ is its \emph{expected} average degree.)
We note that $W_{ij} = M_{\chi(i)\chi(j)} / |\chi^{-1}(\chi(j))|$ is simply the extension of $M$ to the entire matrix, according to the color classes, up to the normalization factor $1/|\chi^{-1}(\chi(j))|$.

The proof of \cref{thm:part-recovery} is a straightforward combination of the following three independent lemmas.

\begin{lemma}[Eigenspace approximation]
\label[lemma]{thm:partition_recovery_theorem2}
    Let $\chi: [n] \to [k]$ be a $k$-partition of $[n]$, and consider its partition matrix $Z \in \{0,1\}^{n \times k}$ as per \cref{def:partition_Z}.
    Let $W \in \mathbb{R}^{n \times n}$ be a symmetric matrix such that $W$ has column span in the span of $\set{Z_a}_{a \in [k]}$.
    Also let $A \in \R^{n \times n}$ be some matrix such that, for all $a \in [k]$,
    \[\Norm{A \frac{Z_a}{\Norm{Z_a}} - W \frac{Z_a}{\Norm{Z_a}}}^2 \leq \delta\,,\]
    and suppose that $\rank_{\leq -\lambda}(A) + \rank_{\geq \lambda}(A) \leq t$ for some $\lambda > 0$.
    Then, given $A$, for any $\eta > 0$, 
    there exists an algorithm that runs in time $n^{O(1)} \cdot \Paren{\delta k /\eta^2}^{-O(t)}$ and outputs a list of size $\Paren{\delta k /\eta^2}^{-O(t)}$ of unit vectors such that, for each unit vector $u \in \R^n$ in the eigenspace of $W$ corresponding to eigenvalues larger than $\lambda + \eta$ in absolute value, there exists some $\hat{u} \in \R^n$ in the list such that $\Norm{\hat{u}-u}^2 \leq O\Paren{\delta k /\eta^2}$.
\end{lemma}

\begin{lemma}[Eigenvector coordinate separation]
\label[lemma]{lem:eigenvalue_separation}
    Let $\chi: [n] \to [k]$ be a $k$-partition of $[n]$, and consider its partition matrices $Z \in \{0,1\}^{n \times k}$ and $B \in \R^{k \times k}$ as per \cref{def:partition_Z}.
    Suppose $|\chi^{-1}(a)| \geq c n$ for all $a \in [k]$.
    Let $M \in \R^{k \times k}$ be a matrix such that 
    \begin{itemize}
        \item $M B^{-1}$ is symmetric,
        \item $\Norm{M} \leq \zeta$,
        \item $\Norm{M^{(a)} - M^{(b)}} \geq \alpha$ for any two rows $M^{(a)}$ and $M^{(b)}$ of $M$.
    \end{itemize}
    For some $\lambda \leq \alpha \sqrt{c} / 2$ and $\lambda \leq \zeta/2$, let the unit vectors $v_1, \ldots, v_r \in \R^k$ be the eigenvectors of $M$ corresponding to eigenvalues larger than $\lambda$ in absolute value.
    Then
    \[\min_{a\neq b\in[k]}\max_{i \in [r]} \frac{1}{\Norm{Z v_i}} \Abs{(v_i)_a- (v_i)_b} \geq \Omega\Paren{\frac{\alpha \sqrt{c}}{\zeta \sqrt{k n}}}\,.\]
\end{lemma}

\begin{lemma}[Clustering with separated rows]
\label[lemma]{lem:cluster-eigen}
    Let $\chi: [n] \to [k]$ be a $k$-partition of $[n]$, and consider its partition matrix $Z \in \{0,1\}^{n \times k}$ as per \cref{def:partition_Z}.
    Suppose $|\chi^{-1}(a)| \geq c n$ for all $a \in [k]$.
    Let unit vectors $v_1, \ldots, v_r \in \R^k$ satisfy 
    \[\min_{a\neq b\in[k]}\max_{i \in [r]} \frac{1}{\Norm{Z v_i}} \Abs{(v_i)_a- (v_i)_b} \geq \frac{\alpha}{\sqrt{n}}\,.\]
    Then, given unit vectors $\hat{u}_1, \ldots, \hat{u}_r \in \R^n$ such that $\Norm{\hat{u}_i - \frac{Z v_i}{\Norm{Z v_i}}}^2 \leq \delta$ for all $i \in [r]$, there exists an algorithm that runs in time $n^{O(1)} \cdot \Paren{\alpha \sqrt{c}}^{-O(k^2)}$ and outputs a list of $\Paren{\alpha \sqrt{c}}^{-O(k^2)}$ $k$-partitions $\hat{\chi}: [n] \to [k]$ such that, for at least one of them, it holds up to a permutation of the color classes that $\mathbb{E}_{\bm{x} \sim \operatorname{Unif}([n])} \1[\hat\chi(\bm{x}) \neq \chi(\bm{x})] \leq O\Paren{\frac{\delta k}{\alpha^2}}$.
\end{lemma}

Let us now prove \cref{thm:part-recovery} given \cref{thm:partition_recovery_theorem2}, \cref{lem:eigenvalue_separation}, and \cref{lem:cluster-eigen}.

\begin{proof}[Proof of \cref{thm:part-recovery}]
We will assume $\lambda = \min\{\alpha \sqrt{c} / 4, \zeta/4\}$, such that $2\lambda$ satisfies the upper bound required to apply \cref{lem:eigenvalue_separation}.

First, let $S:=\Set{u \in \R^n\; \mid\; W u = \tau u, \Norm{u}=1, \Abs{\tau}> 2\lambda}$ be the set of eigenvectors of $W$ with corresponding eigenvalues larger than $2\lambda$ in absolute value.
By applying \cref{thm:partition_recovery_theorem2} with $\eta = \lambda$, we can recover a list of size $\Paren{\delta k/\lambda^2}^{-O(t)}$ such that, for each unit eigenvector $u \in \R^n$ in this subspace, there exists some $\hat{u}$ in the list with $\norm{\hat{u}-u}^2 \leq O(\delta k/\lambda^2)$.

Observe that each eigenvector $u$ of $W$ is by definition in the column span of $Z$, so we can write $u = Z v$ for some $v \in \R^{k}$.
Suppose $Wu = \tau u$.
Then $M v = B^{-1} Z^\top W Z v = \tau B^{-1} Z^\top Z v = \tau v$, so $v$ is an eigenvector of $M$ with eigenvalue $\tau$.
Hence, if $v_1, \ldots, v_r$ are the eigenvalues of $M$ with corresponding eigenvalues larger than $2\lambda$ in absolute value, then $\frac{Z v_1}{\Norm{Zv_1}}, \ldots, \frac{Z v_r}{\Norm{Zv_r}} \in S$.

By \cref{lem:eigenvalue_separation}, we have for these eigenvectors $v_1, \ldots, v_r$ of $M$ that 
\[\min_{a\neq b\in[k]}\max_{i \in [r]} \frac{1}{\Norm{Z v_i}} \Abs{(v_i)_a- (v_i)_b} \geq \Omega\Paren{\frac{\alpha \sqrt{c}}{\zeta \sqrt{k n}}}\,,\]
and by the previous discussion for each $v_i$ there exists some $\hat{u} \in \R^n$ in the recovered list such that $\Norm{\hat{u} - \frac{Z v_i}{\Norm{Zv_i}}}^2 \leq O(\delta k / \lambda^2)$.
To know which of the unit vectors $\hat{u}$ in the list correspond to these eigenvectors, we can afford to try all possibilities of at most $k$ vectors in the list; there are at most $\Paren{\delta k/\lambda^2}^{-O(tk)}$ such possibilities.
The guarantee then follows from \cref{lem:cluster-eigen}, giving probability of misclassification $O(\frac{\delta \zeta^2 k^3}{\min\{\alpha^2 c, \zeta^2\} \alpha^2 c} )$.
\end{proof}

The next three subsection prove \cref{thm:partition_recovery_theorem2}, \cref{lem:eigenvalue_separation}, and \cref{lem:cluster-eigen}.

\subsection{Eigenspace approximation}

In this section we prove \cref{thm:partition_recovery_theorem2}.
First, we state a generic result that proves that if two matrices are close in a subspace that for one of the matrices is an eigenspace with eigenvalues of large absolute value, then the small eigenvalues of the other matrix are far from this subspace.

\begin{lemma}[Eigenspace closeness]
    \label[lemma]{lemma:lin-alg-subspace}
    Let $X \in \R^{n \times n}$ be a symmetric matrix, and let $v_1, \ldots, v_k$ be any $k$ of its eigenvectors with non-zero eigenvalues $\lambda_1, \ldots, \lambda_{k}$.
    Let $\lambda = \min\{|\lambda_1|, \ldots, |\lambda_k|\}$.
    Also let $Y \in \R^{n \times n}$ be a symmetric matrix, and let $P \in \R^{n \times n}$ be the orthogonal projection to the span of the eigenvectors of $Y$ that are smaller in absolute value than $\lambda - \eta$ for some $\eta > 0$.
    Then, for any orthonormal basis $u_1, \ldots, u_{k} \in \R^n$ of $\operatorname{span}(v_1, \ldots, v_k)$,
    \[\sum_{i=1}^{k} \Norm{P u_i}^2 \leq \frac{1}{\eta^2} \sum_{i=1}^{k} \Norm{(X-Y)u_i}^2\,.\]
\end{lemma}
\begin{proof}
First, for the eigenvectors $v_1, \ldots, v_{k}$ of $X$ we have that 
\begin{equation}
\label{lemma:lin-alg-subspace-eq1}
\Norm{(X-Y)v_i}^2 = \Norm{\lambda_i v_i - Y v_i}^2 = \Norm{(\lambda_i I_n - Y) v_i}^2 \geq \Norm{P(\lambda_i I_n - Y) v_i}^2\,.
\end{equation}
Let now $v_1', \ldots, v_{m}' \in \R^n$ be the eigenvectors of $Y$ in $P$ with corresponding eigenvalues $\lambda_1', \ldots, \lambda_{m}'$ smaller in absolute value than $\lambda - \eta$.
Continuing \cref{lemma:lin-alg-subspace-eq1},
\begin{equation}
\label{lemma:lin-alg-subspace-eq2}
\Norm{(X-Y)v_i}^2 \geq \sum_{j=1}^{m} (\lambda_i - \lambda_j')^2 \langle v_j', v_i \rangle^2 \geq \eta^2 \sum_{j=1}^{m} \langle v_j', v_i \rangle^2 = \eta^2 \Norm{Pv_i}^2\,.
\end{equation}

Second, we claim that $\sum_{i=1}^k \Norm{(X-Y)v_i}^2 = \sum_{i=1}^{k} \Norm{(X-Y)u_i}^2$ and $\sum_{i=1}^k \Norm{P v_i}^2 = \sum_{i=1}^{k} \Norm{P u_i}^2$.
The proof is simply that, for any matrix $M \in \R^{n \times n}$ and any two orthonormal bases of the same subspace $\{x_i\}$ and $\{y_i\}$,
\[\sum_{i} \Norm{Mx_i}^2 = \sum_{i} \Tr(M^\top M x_i x_i^\top) = \Tr\Paren{M^\top M \sum_{i} x_i x_i^\top} = \Tr\Paren{M^\top M \sum_{i} y_i y_i^\top}\,,\]
and by the same sequence of transformations we get back that $\Tr\Paren{M^\top M \sum_{i} y_i y_i^\top} = \sum_{i} \Norm{My_i}^2$. 
Using this fact for the orthonormal bases $\{v_i\}$ and $\{u_i\}$, we prove the desired fact, and combining it with \cref{lemma:lin-alg-subspace-eq2} we complete the proof.
\end{proof}

We state now the well-known bounds of subspace enumeration.

\begin{lemma}[Subspace enumeration]
    \label[lemma]{lemma:subspace_enumeration}
    There exists an algorithm that outputs with high probability an $\delta$-net of the unit $(d-1)$-dimensional ball in time $(1/\delta)^{O(d)}$.
\end{lemma}

\begin{proof}
    By standard arguments, there exists such a cover of size $m = O(1/\delta)^d$ (see Example 5.8 in~\cite{MR3967104-Wainwright19}).
    Furthermore, it suffices to sample $O(m \log m)$ random points on the unit ball in order to obtain with high probability such a cover.
\end{proof}

Finally, we prove \cref{thm:partition_recovery_theorem2}.

\begin{proof}[Proof of \cref{thm:partition_recovery_theorem2}]
    We have $\Norm{A \frac{Z_a}{\Norm{Z_a}} - W \frac{Z_a}{\Norm{Z_a}}}^2 \leq \delta$ for all $a \in [k]$, so $\sum_{a=1}^k \Norm{\Paren{A - W}\frac{Z_a}{\Norm{Z_a}}}^2 \leq \delta k$.
    Because $\frac{Z_a}{\Norm{Z_a}}$ are orthonormal, for any orthonormal $u_1, \ldots, u_k$ that span the same subspace we also have $\sum_{i=1}^k \Norm{\Paren{A -W}u_i}^2
    \leq \delta k$.
    Also recall that $W$ has column span in the span of $\{Z_a\}_{a \in [k]}$, so all of its eigenvectors with non-zero eigenvalues lie in the span of $\{u_1, \ldots, u_k\}$.
    Let $P$ be the orthogonal projection to the eigenvectors of $A$ with eigenvalue smaller than $\lambda$ in absolute value.
    Then, by \cref{lemma:lin-alg-subspace} applied with $W, A$ as $X, Y$, we get for any orthonormal vectors $u_1, \ldots, u_r$ that span the eigenspace of $W$ with eigenvalues larger than $\lambda+\eta$ in absolute value, for any $\eta >0$, that $\sum_{i=1}^r \Norm{P u_i}^2 \leq \delta k / \eta^2$.
    Hence, $\Norm{P u_i}^2 \leq \delta k / \eta^2$ for all $i \in [r]$, so each $u_i$ is $O(\sqrt{\delta k / \eta^2})$-close to the subspace corresponding to $P^{\bot} = I_n - P$.

    Then, by \cref{lemma:subspace_enumeration}, by finding an $O(\sqrt{\delta k / \eta^2})$-net of the eigenspace of $A$ corresponding to $P^\bot$, we are guaranteed to find a vector that is $O(\sqrt{\delta k / \eta^2})$-close to each $u_i$.
    This subspace has dimension at most $t$, so applying \cref{lemma:subspace_enumeration} takes time $(\delta k/\eta^2)^{O(t)}$ and produces a list of size at most $(\delta k /\eta^2)^{O(t)}$.
    The overall time complexity is then bounded by $n^{O(1)} \cdot (\delta k/\eta^2)^{O(t)}$.
\end{proof}

\subsection{Eigenvector coordinate separation}

In this section we prove \cref{lem:eigenvalue_separation}.
First, we state a technical lemma.

\begin{lemma}[Non-uniformity of eigenvectors]\label[lemma]{lem: eigenvectorsdistinguish}
    Let $c\geq 0$, $0<\lambda< \zeta$, and $\alpha\geq \sqrt{2}c\lambda$ be parameters.
    Let $D \in \R^k$ with $D_x \leq c$ for all $x \in [k]$.
    Let $M$ be a symmetric $k\times k$ matrix with $\Norm{M}\leq \zeta$ such that for any two distinct columns $M_x$ and $M_y$ it holds that $\Norm{D_xM_x-D_yM_y}^2 \geq \alpha^2$.
    Let $v_1, \ldots, v_k$ be orthonormal eigenvectors of $M$ and let $S:=\Set{v_r \; \mid\; M v_r = \tau v_r, |\tau| > \lambda}$.
    Then
    $$\min_{x\neq y\in[k]}\max_{v_r \in S} \Abs{D_x (v_r)_x-D_y (v_r)_y}\geq \sqrt{\frac{\alpha^2-2c^2\lambda^2}{k\Paren{\zeta^2-\lambda^2}}} \,.$$
\end{lemma}

\begin{proof}
    We write $M = \sum_{r=1}^k \lambda_r v_r v_r^{\top}$ and get that, for any $x \neq y \in [k]$,
    \begin{equation*}
        D_x M_x - D_y M_y=\sum_{r=1}^k \lambda_r (D_x (v_r)_x - D_y (v_r)_y)v_r \,.
    \end{equation*}
    Using that $\Norm{M}\leq \zeta$ and $\Abs{\lambda_r}\leq \lambda$, we get
    \begin{align}
    \label{eq:rounding_split_1}
        \begin{split}
            \Norm{D_x M_x- D_y M_y}^2
            &=\sum_{r=1}^k \lambda_r^2 (D_x (v_r)_x - D_y(v_r)_y)^2 \\
            &=\sum_{r\in S} \lambda_r^2 (D_x (v_r)_x - D_y(v_r)_y)^2+\sum_{r\notin S} \lambda_r^2 (D_x (v_r)_x - D_y(v_r)_y)^2 \\
            &\leq \zeta^2\sum_{r\in S} (D_x (v_r)_x- D_y (v_r)_y)^2+\lambda^2\sum_{r\notin S} (D_x (v_r)_x- D_y (v_r)_y)^2 \,.
        \end{split}
    \end{align}
    Because the eigenvectors are orthonormal, the rows of the matrix $V \in \R^{k \times k}$ whose columns are $\set{v_{r}}_{r \in [k]}$ are also orthonormal: that is, $\sum_{r=1}^{k} (v_r)_x (v_r)_y = 0$ for $x \neq y$ and $\sum_{r=1}^{k} (v_r)_x^2 = 1$ for $x \in [k]$.
    Therefore, we obtain
    \begin{align*}
        \sum_{r \notin S} (D_x (v_r)_x- D_y (v_r)_y)^2
        & = \sum_{r=1}^k (D_x (v_r)_x- D_y (v_r)_y)^2 - \sum_{r \in S} (D_x (v_r)_x- D_y (v_r)_y)^2 \\
        & = D_x^2 + D_y^2 - \sum_{r \in S} (D_x (v_r)_x- D_y (v_r)_y)^2 \\
        & \leq 2c^2 - \sum_{r \in S} (D_x (v_r)_x- D_y (v_r)_y)^2 \,.
    \end{align*}
    Plugging it into \cref{eq:rounding_split_1}, we get
    \begin{equation*}
        \Norm{D_x M_x- D_y M_y}^2
        \leq \Paren{\zeta^2-\lambda^2} \sum_{r \in S} (D_x (v_r)_x- D_y (v_r)_y)^2 + 2 c^2 \lambda^2 \,.
    \end{equation*}
    Since $\Norm{D_xM_x-D_yM_y}^2\geq \alpha^2$, it follows that
    \begin{equation*}
        \sum_{r \in S} (D_x (v_r)_x- D_y (v_r)_y)^2 \geq \frac{\alpha^2 - 2c^2\lambda^2}{\zeta^2-\lambda^2} \,.
    \end{equation*}
    Finally, by an averaging argument, there must exist some $r \in S$ such that
    \begin{equation*}
        (D_x (v_r)_x- D_y (v_r)_y)^2 \geq \frac{\alpha^2 - 2c^2\lambda^2}{k(\zeta^2-\lambda^2)} \,.
    \end{equation*}
\end{proof}

Next, we prove \cref{lem:eigenvalue_separation}.

\begin{proof}[Proof of \cref{lem:eigenvalue_separation}]
Observe that if $Mv = \tau v$, then we also have that $B^{1/2}v$ is an eigenvector of $\widetilde{M} = B^{1/2} M B^{-1/2}$ of value $\tau$.
Note that $\widetilde{M}$ is symmetric and continues to satisfy $\Norm{\widetilde{M}}\leq \zeta$.
Then, we apply \cref{lem: eigenvectorsdistinguish} with the matrix $\widetilde{M}$ and with the eigenvectors $B^{1/2} v$ of $\widetilde{M}$ with eigenvalue at least $\lambda$ in absolute value (normalized to be unit vectors).
We apply the lemma with $D \in \R^k$ such that $D_x = (\sqrt{n} B^{-1/2})_x$, whose largest entry is upper bounded by $1/\sqrt{c}$.
Recall that we assume any two rows of $M$ satisfy $\Norm{M^{(a)} - M^{(b)}} \geq \alpha$, so also any two rows of $\sqrt{n} M B^{-1/2} = \sqrt{n} B^{-1/2} \widetilde{M}$ satisfy $\Norm{\sqrt{n}B^{-1/2} \Paren{M^{(a)} - M^{(b)}}} \geq \alpha$.
Taking the transpose and using that $\widetilde{M}$ is symmetric, any two columns of $\sqrt{n} \widetilde{M} B^{-1/2}$ satisfy $\Norm{\sqrt{n}B^{-1/2}(a) \widetilde{M}_a - \sqrt{n}B^{-1/2}(b) \widetilde{M}_b} \geq \alpha$.
Applying \cref{lem: eigenvectorsdistinguish}, and using that $\alpha^2 \geq 4 \lambda^2 /c$ and $\zeta^2 \geq 4 \lambda^2$, we get 
\[\min_{a\neq b\in[k]}\max_{v \in S} \frac{1}{\Norm{B^{1/2}v}} \sqrt{n} \Abs{v_a- v_b} \geq \Omega\Paren{\sqrt{\frac{\alpha^2 - 2\lambda^2/c}{k(\zeta^2-\lambda^2)}}} \geq \Omega\Paren{\frac{\alpha}{\zeta\sqrt{k}}}\,.\]

The above shows that, for all $a \neq b \in [k]$, there exists $\frac{v}{\Norm{v}} \in S$ with $\Norm{B^{1/2}v}=1$ such that $\sqrt{n}\Abs{v_a - v_b} \geq \Omega\Paren{\frac{\alpha}{\zeta\sqrt{k}}}$, so $\Abs{v_a - v_b} \geq \Omega\Paren{\frac{\alpha}{\zeta\sqrt{kn}}}$.
We note that $\Norm{B^{1/2}v} = 1$ implies that $\frac{1}{\sqrt{n}} \leq \Norm{v} \leq \frac{1}{\sqrt{cn}}$.
We also have $\sqrt{c n} \norm{v} \leq \norm{Zv} \leq \sqrt{n} \norm{v}$, so $\sqrt{c} \leq \norm{Zv} \leq \frac{1}{\sqrt{c}}$.
Then, if we consider for each such $v$ some $v' \propto v$ such that $\norm{Zv'} = 1$, we have that $\Abs{v'_a - v'_b} \geq \Omega\Paren{\frac{\alpha\sqrt{c}}{\zeta\sqrt{k n}}}$
\end{proof}

\subsection{Clustering with separated rows}

In this section we prove \cref{lem:cluster-eigen}.
First, we state the guarantees of a clustering algorithm used in the proof.

\begin{lemma}[Partitioning based on distinguishing eigenvectors (based on Lemma E.7 in \cite{MR3536556-David16})]\label[lemma]{lem:clusteringalmosteigenvectors}
    Let $\chi: [n] \to [k]$ be a $k$-partition of $[n]$.
    Suppose $|\chi^{-1}(a)| \geq c n$ for all $a \in [k]$.
    For $\ell \leq k$, let $\set{v_{i}}_{i \in [\ell]}$ be $k$-dimensional vectors that satisfy
    \[\min_{a\neq b\in[k]}\max_{i\in [\ell]} \Abs{(v_i)_a-(v_i)_b}\geq \frac{\alpha}{\sqrt{n}}\,.\]
    Also let $\set{\hat{u}_{i}}_{i \in [\ell]}$ be $n$-dimensional unit vectors that satisfy $\Norm{\hat{u}_i-Zv_i}^2\leq \delta$ for all $i \in [\ell]$.
    Then \cref{alg: spectralclustering} runs in time $n^{O(1)} \cdot (\alpha \sqrt{c})^{-O(k^2)}$ and returns a list of $(\alpha \sqrt{c})^{-O(k^2)}$ $k$-partitions $\hat{\chi}: [n] \to [k]$ such that, for at least one of them, it holds up to a permutation of the color classes that $\mathbb{E}_{\bm{x} \sim \operatorname{Unif}([n])} \1[\hat\chi(x) \neq \chi(x)] \leq O\Paren{\frac{\delta k}{\alpha^2}}$.
\end{lemma}

\begin{algorithm}[H]
    \caption{Spectral partitioning algorithm (based on Algorithm 10 in \cite{MR3536556-David16})}
    \label{alg: spectralclustering}
    \DontPrintSemicolon
    Guess $k$-dimensional vectors $\{v_i^*\}_{i \in [\ell]}$ with entries in $[-1/\sqrt{cn}, 1/\sqrt{cn}]$ up to entrywise resolution $\alpha / \sqrt{12 k n}$.\;
    Let $S_1 = \emptyset, \ldots, S_k = \emptyset$.\;
    Insert each $x \in [n]$ in $S_a$ where
    \[
    a = \argmax_{a \in [k]} \sum_{j=1}^\ell \Paren{(\hat{u}_j)_x-(v^*_j)_a}^2\,.
    \]
\end{algorithm}

\begin{proof}
    To begin with, we remark that the number of guesses is bounded by
    \[\Paren{\Paren{2/\sqrt{cn}} / \Paren{\alpha/\sqrt{12kn}}}^{k\ell} \leq O(\sqrt{k}/(\alpha\sqrt{c}))^{k^2}\,,\] 
    so the algorithm leads to the desired time complexity bound and list size bound.

    Let $u_j=Zv_j$, and note that if $Z_{ia}=1$ then $(u_j)_i=(v_j)_a$. Define the set 
    $$\mathrm{Good}:=\Set{x\in [n]\;\mid\; \Paren{(\hat{u}_j)_x-(u_j)_x}^2\leq \frac{\alpha^2}{12kn}, \forall j \in [\ell]}\,.$$
    Note that, because $\Norm{\hat{u}_j-Zv_j}^2\leq \delta$ for all $j \in [\ell]$, we have that $n - |\mathrm{Good}| \leq \frac{\delta}{\alpha^2/(12kn)} = \frac{12\delta k}{\alpha^2} n$, so $|\mathrm{Good}| \geq n - \frac{12\delta k}{\alpha^2} n$.

    Consider the guess of the algorithm in which $|(v^*_i)_a - (v_i)_a| \leq \alpha / \sqrt{12 k n}$ for all $a \in [k]$, $i \in [\ell]$.
    Consider some $x \in \mathrm{Good}$ with $Z_{x a}=1$ for some $a \in [k]$.
    Then $\Paren{(\hat{u}_j)_x-(u_j)_x}^2\leq \frac{\alpha^2}{12kn}$ and $(u_j)_x = (v_j)_a$, so
    \begin{align*}
        \Paren{(\hat{u}_j)_x-(v^*_j)_a}^2
        &= \Paren{\Paren{(\hat{u}_j)_x-(u_j)_x}+\Paren{(u_j)_x-(v_j)_a}+\Paren{(v_j)_a-(v^*_j)_a}}^2\\
        &= \Paren{\Paren{(\hat{u}_j)_x-(u_j)_x}+\Paren{(v_j)_a-(v^*_j)_a}}^2\\
        &\leq 2 \Paren{(\hat{u}_j)_x-(u_j)_x}^2 + 2 \Paren{(v_j)_a-(v^*_j)_a}^2\\
        &\leq \frac{\alpha^2}{3kn}\,.
    \end{align*}

    Then, overall, we have 
    $$\sum_{j=1}^\ell \Paren{(\hat{u}_j)_x-(v^*_j)_a}^2 \leq \ell \frac{\alpha^2}{3kn} \leq \frac{\alpha^2}{3n}\,.$$

    On the other hand, for any $b \neq a \in [k]$, there exists $i\in[\ell]$ with $\Abs{(u_i)_x-(v_i)_b}=\Abs{(v_i)_a-(v_i)_b}\geq \frac{\alpha}{\sqrt{n}}$, so similarly we have
    \begin{align*}
    \Paren{(\hat{u}_i)_x-(v_i^*)_b}^2
    &= \Bigg(\Paren{(\hat{u}_i)_x-(u_i)_x}+\Paren{(u_i)_x-(v_i)_b}+\Paren{(v_i)_b-(v^*_i)_b}\Bigg)^2\\
    &\geq \Paren{\frac{\alpha}{\sqrt{n}}-2\frac{\alpha}{\sqrt{12kn}}}^2\\
    &> \frac{\alpha^2}{3n}\,.
    \end{align*}
    Hence, 
    $$\sum_{j=1}^\ell \Paren{(\hat{u}_j)_x-(v_j^*)_b}^2 > \frac{\alpha^2}{3n}\,.$$
    This shows that there exists a guess of the algorithm $v_1^*, \ldots, v_\ell^*$ for which all vertices in $\mathrm{Good}$ get correctly partitioned.
    Then, because $|\mathrm{Good}| \geq n - \frac{12\delta k}{\alpha^2} n$, for this guess at most $\frac{12k\delta}{\alpha^2}n$ vertices are misclassified.
\end{proof}

Finally, we prove \cref{lem:cluster-eigen}.

\begin{proof}[Proof of \cref{lem:cluster-eigen}]

    For each $v_i$, let $v'_i \propto v_i$ such that $\Norm{Z v'_i} = 1$.
    Then the assumptions guarantee that 
    \[\min_{a\neq b\in[k]}\max_{i \in [r]} \Abs{(v'_i)_a- (v'_i)_b} \geq \frac{\alpha}{\sqrt{n}}\,.\]
    We also have that $\Norm{\hat{u}_i - Z v'_i}^2 \leq \delta$ for all $i \in [r]$.
    Then the proof follows by \cref{lem:clusteringalmosteigenvectors}: we can recover a list of partitions such that one of them contains (up to a permutation) at most $O\Paren{\frac{\delta k}{\alpha^2}} n$ wrongly partitioned vertices.
    The time complexity is $n^{O(1)} \cdot (\alpha \sqrt{c})^{-O(k^2)}$ the list size bound is $(\alpha \sqrt{c})^{-O(k^2)}$.
\end{proof}

%% file: content/variance.tex
\section{Variance of colorings in expander graphs}\label{sec: 3coloring}

In this section we prove a structural result for colorable expander graphs concerning the variance of the connections between any two color classes.
This result is essential to our proof of \Cref{thm:kalg-partial}.

\begin{lemma}[Variance of $k$-colorings in expanders]
\label[lemma]{lemma: expansionimpliesvariance}
Let $G$ be an $n$-vertex $d$-regular graph with adjacency matrix $A$, and let $\tilde{A} = \frac{1}{d}A$.
Let $\lambda_2 > 0$ be the second largest eigenvalue of $\tilde{A}$.
Let $\chi: [n] \to [k]$ be a $k$-partition of its vertices with $\min_{a \in [k]} |\chi^{-1}(a)| \geq c n$.
Suppose that $G$ has at most $\delta d n$ edges with both endpoints in the same part with respect to $\chi$.
For $x \in [n]$ and $a \in [k]$, let $\D{x}{a} = \sum_{y \in \chi^{-1}(a)} \tilde{A}_{xy}$.
Then, for any $a, b \in [k]$ and $\bm{x}$ uniformly random in $\chi^{-1}(a)$,
\[\E_{\bm{x} \sim \chi^{-1}(a)} \Paren{\D{\bm{x}}{b} - \E_{\bm{x} \sim \chi^{-1}(a)} \D{\bm{x}}{b}}^2 \leq O\Paren{\frac{\lambda_2}{c} + \frac{\delta}{\lambda_2 c }}\,.\]
\end{lemma}
\begin{proof}
    Fix some $a, b \in [k]$.
    For simplicity, we denote $\sigma^2 = \E_{\bm{x} \sim \chi^{-1}(a)} \Paren{\D{\bm{x}}{b}-\E_{\bm{x} \sim \chi^{-1}(a)} \D{\bm{x}}{b}}^2$.

    If $a = b$, we have 
    \begin{align*}
        \sigma^2
        \leq \E_{\bm{x} \sim \chi^{-1}(a)} \Paren{\D{\bm{x}}{a}}^2
        \leq \E_{\bm{x} \sim \chi^{-1}(a)} \D{\bm{x}}{a}
        \leq \frac{2 \delta n}{|\chi^{-1}(a)|}
        \leq \frac{2\delta}{c}\,.
    \end{align*}

    We focus now on the case $a \neq b$.
    By regularity, $\lambda_1=1$ corresponds to eigenvector $\1$.
    Let $\beta \in [0, 1]$ be some parameter to be chosen later.
    We define $v \in \R^n$ as follows:
    \begin{itemize}
        \item $v_x:=\frac{1}{\sqrt{|\chi^{-1}(a)|}}\Paren{\D{x}{b}-\E_{\bm{x} \sim \chi^{-1}(a)} \D{\bm{x}}{b}}$ for $x\in \chi^{-1}(a)$\,,
        \item $v_x:=\frac{1}{\sqrt{|\chi^{-1}(a)|}}\beta$ for $x\in \chi^{-1}(b)$\,,
        \item $v_x:=0$ otherwise\,.
    \end{itemize}

    Let us calculate $v^\top \tilde{A}v$.
    Let $\tilde{A}_{(a)} \in \R^{|\chi^{-1}(a)| \times |\chi^{-1}(a)|}$ be the restriction of $\tilde{A}$ to rows and columns in $\chi^{-1}(a)$, and let $v_{(a)} \in \R^{|\chi^{-1}(a)|}$ be the restriction of $v$ to entries in $\chi^{-1}(a)$.
    Then
    \begin{align}
    \begin{split}
    \label{eq:3_col_num}
    v^\top \tilde{A}v
        &= v_{(a)}^\top \tilde{A}_{(a)} v_{(a)} + v_{(b)}^\top \tilde{A}_{(b)} v_{(b)} + 2 \beta \frac{1}{|\chi^{-1}(a)|} \sum_{x \in \chi^{-1}(a)} \sum_{\substack{y \in \chi^{-1}(b)\\y \sim x}} \Paren{\D{x}{b}-\E_{\bm{x} \sim \chi^{-1}(a)} \D{\bm{x}}{b}}\\
        &= v_{(a)}^\top \tilde{A}_{(a)} v_{(a)} + v_{(b)}^\top \tilde{A}_{(b)} v_{(b)} + 2 \beta \E_{\bm{x} \sim \chi^{-1}(a)} \D{\bm{x}}{b} \Paren{\D{\bm{x}}{b}-\E_{\bm{x} \sim \chi^{-1}(a)} \D{\bm{x}}{b}}\,.
    \end{split}
    \end{align}

    For the first two terms, we have 
    \[v_{(a)}^\top \tilde{A}_{(a)} v_{(a)} \geq - \Paren{\sum_{x, y \in \chi^{-1}(a)} \tilde{A}_{xy}} \cdot \Norm{v_{(a)}}^2_\infty \geq - \frac{2\delta n}{|\chi^{-1}(a)|} \geq - \frac{2\delta}{c}\]
    and 
    \[v_{(b)}^\top \tilde{A}_{(b)} v_{(b)} \geq - \Paren{\sum_{x, y \in \chi^{-1}(b)} \tilde{A}_{xy}} \cdot \Norm{v_{(b)}}^2_\infty \geq - \frac{2\delta n}{|\chi^{-1}(a)|} \geq - \frac{2\delta}{c}\,.\]

    For the cross-term, we use that
    \begin{align*}
    &2 \beta \E_{\bm{x} \sim \chi^{-1}(a)} \D{\bm{x}}{b} \Paren{\D{\bm{x}}{b}-\E_{\bm{x} \sim \chi^{-1}(a)} \D{\bm{x}}{b}}\\
    &= 2 \beta \E_{\bm{x} \sim \chi^{-1}(a)} \Paren{\D{\bm{x}}{b}-\E_{\bm{x} \sim \chi^{-1}(a)} \D{\bm{x}}{b}}^2 + 2 \beta \E_{\bm{x} \sim \chi^{-1}(a)} \Paren{\E_{\bm{x} \sim \chi^{-1}(a)} \D{\bm{x}}{b}} \Paren{\D{\bm{x}}{b}-\E_{\bm{x} \sim \chi^{-1}(a)} \D{\bm{x}}{b}}\\
    &= 2 \beta \E_{\bm{x} \sim \chi^{-1}(a)} \Paren{\D{\bm{x}}{b}-\E_{\bm{x} \sim \chi^{-1}(a)} \D{\bm{x}}{b}}^2\\
    &= 2\beta \sigma^2\,.
    \end{align*}
    Then, plugging these bounds back into \cref{eq:3_col_num}, we get 
    \begin{align*}
    v^\top \tilde{A}v
    \geq 2 \beta \sigma^2 - 4 \delta/c\,.
    \end{align*}
    We also note that $\norm{v}^2 = \sigma^2 + \frac{|\chi^{-1}(b)|}{|\chi^{-1}(a)|} \beta^2 \leq \sigma^2 + \beta^2/c$ and $\langle v, \frac{\1}{\sqrt{n}}\rangle^2 = \frac{|\chi^{-1}(b)|^2}{|\chi^{-1}(a)| n}\beta^2 \leq \beta^2 / c$, so
    \[v^\top \tilde{A} v \leq \frac{1}{n} \langle v, \1\rangle^2 + \lambda_2 \Paren{\norm{v}^2 - \frac{1}{n} \langle v, \1\rangle^2} \leq \lambda_2 \sigma^2 + \beta^2/c\,.\]
    Putting everything together,
    \begin{align*}
        2 \beta \sigma^2 - 4 \delta/c \leq v^\top \tilde{A} v \leq \lambda_2 \sigma^2 + \beta^2/c\,,
    \end{align*}
    and taking $\beta=\lambda_2$ we get
    \[\sigma^2 \leq \lambda_2/c + \frac{4\delta }{\lambda_2c} \leq O\Paren{\frac{\lambda_2}{c} + \frac{\delta}{\lambda_2 c }}\,.\]
\end{proof}

%% file: content/algorithms.tex
\section{Finding colorings and independent sets}
\label{sec:algorithms}

\subsection{Colorable models}

Our most general result for $k$-colorings is the following:
\begin{theorem}[Result for $k$-coloring]
    \label{thm:kalg-partial}
    Let $M \in \R^{k \times k}_{\geq 0}$ be a reversible row stochastic matrix with zero on the diagonal and non-positive second eigenvalue, and let $\pi$ be its stationary distribution.
    Let $S_1, \ldots, S_{k'}$ be a partition of $[k]$ such that, for all $a, b \in S_i$ the rows $M^{(a)}$ and $M^{(b)}$ of $M$ are equal, and for all $a \in S_i, b \in S_j$ with $i \neq j$ the rows $M^{(a)}$ and $M^{(b)}$ of $M$ are distinct.
    For all $i \in [k']$, let $a_i^* = \argmax_{a \in S_i} \pi(a)$, and let $p_i = \max\Paren{0, \pi(a_i^{*}) - \sum_{\substack{a \neq a_i^{*} \in S_i}} \pi(a)}$.
    Then, given a $d$-regular $(M, \delta)$-colorable $\lambda$-one-sided expander $G$, where $\delta, \lambda > 0$ are small enough, there exists an algorithm that runs in time $n^{O(1)} \cdot O_{\delta, \lambda, M}(1)$ and outputs a $(k', 1 - \sum_{i \in [k']} p_i + O_M(\sqrt{\delta} + \lambda + \delta/\lambda))$-coloring for $G$.
\end{theorem}

As a direct corollary, we approach a full coloring in the case when all rows of $M$ are distinct:
\begin{corollary}[Result for $k$-coloring, distinct rows]
    \label[corollary]{thm:kalg}
    Let $M \in \R^{k \times k}_{\geq 0}$ be a reversible row stochastic matrix with zero on the diagonal, non-positive second eigenvalue, and distinct rows.
    Then, given a $d$-regular $(M, \delta)$-colorable $\lambda$-one-sided expander $G$, where $\delta, \lambda > 0$ are small enough, there exists an algorithm that runs in time $n^{O(1)} \cdot O_{\delta, \lambda, M}(1)$ and outputs a $(k, O_M(\sqrt{\delta} + \lambda + \delta/\lambda))$-coloring for $G$.
\end{corollary}

We start by proving two lemmas that we will use in the proof.
First, we show that the models we consider in \cref{thm:kalg-partial} are well-behaved: they have a unique stationary distribution in which each vertex has non-zero probability.

\begin{lemma}[Model color class size lower bound]
\label[lemma]{lemma:model-color-size}
Let $M \in \R^{k \times k}_{\geq 0}$ be a reversible row stochastic matrix with zeros on the diagonal.
Then the stationary distribution $\pi$ of $M$ is unique and satisfies $\min_{a \in [k]} \pi_a \geq \Omega_M(1)$.
Furthermore, for every $A \in \R^{n \times n}$ and every $k$-partition $\chi : [n] \to [k]$ of $[n]$ such that $\Norm{M - M(A, \chi)} \leq \delta$, then also $\Abs{\frac{1}{n}|\chi^{-1}(a)| - \pi_a} \leq 8\sqrt{\delta}k$ for all $a \in [k]$.
\end{lemma}
\begin{proof}
First, because $M$ has non-positive second eigenvalue, its non-zero entries describe a connected graph.
Furthermore, for a stationary distribution $\pi$ of $M$, for all $a \neq b \in [k]$ we have $\pi_a M_{ab} = \pi_b M_{ba}$, so if $M_{ab},M_{ba}>0$ also 
\[\Omega_M(1) \leq \frac{\pi_a}{\pi_b} = \frac{M_{ab}}{M_{ba}} \leq O_M(1)\,.\]
Then, because the non-zero entries of $M$ describe a connected graph, we can chain inequalities of this form to obtain for all $a \neq b \in [k]$ the ratios $\pi_a/\pi_b$.
Together with the fact that $\sum_{a \in [k]} \pi_a = 1$, we get that $\pi$ is uniquely determined by $M$ and satisfies $\min_{a \in [k]} \pi_a \geq \Omega_M(1)$.

Second, we have that $M(A, \chi)_{ab}$ is $\delta$-close to $M_{ab}$ for all $a,b\in [k]$. 
Also, for all $a \neq b \in [k]$ we have $\frac{1}{n}|\chi^{-1}(a)| M(A, \chi)_{ab} = \frac{1}{n}|\chi^{-1}(b)| M(A, \chi)_{ba}$ and $\sum_{a \in [k]} \frac{1}{n}|\chi^{-1}(a)| = 1$.
Then, if $M_{ab},M_{ba}>0$ and $\delta > 0$ is small enough such that $\delta \leq \sqrt{\delta} M_{ab}, \sqrt{\delta} M_{ba}$, also 
\[(1-4\sqrt{\delta}) \frac{\pi_a}{\pi_b} \leq \frac{(1-\sqrt{\delta})M_{ab}}{(1+\sqrt{\delta})M_{ba}} \leq \frac{\frac{1}{n}|\chi^{-1}(a)|}{\frac{1}{n}|\chi^{-1}(b)|} = \frac{M(A, \chi)_{ab}}{M(A, \chi)_{ba}} \leq \frac{(1+\sqrt{\delta})M_{ab}}{(1-\sqrt{\delta})M_{ba}} \leq (1+4\sqrt{\delta})\frac{\pi_a}{\pi_b}\,.\]
This implies that for all $a, b \in[k]$ we have 
\[(1-8\sqrt{\delta}k) \frac{\pi_a}{\pi_b} \leq (1-4\sqrt{\delta})^{k-1} \frac{\pi_a}{\pi_b} \leq \frac{\frac{1}{n}|\chi^{-1}(a)|}{\frac{1}{n}|\chi^{-1}(b)|} \leq (1+4\sqrt{\delta})^{k-1} \frac{\pi_a}{\pi_b} \leq (1+4\sqrt{\delta}k) \frac{\pi_a}{\pi_b}\,,\]
so summing over all $a \in [k]$ we get 
\[(1-8\sqrt{\delta}k)\frac{1}{\pi_b} \leq \frac{1}{\frac{1}{n}|\chi^{-1}(b)|} \leq (1+4\sqrt{\delta}k) \frac{1}{\pi_b}\,,\]
which finally implies that $\Abs{\frac{1}{n}|\chi^{-1}(b)| - \pi_b} \leq 8\sqrt{\delta}k$ for all $b \in [k]$.
\end{proof}

Second, we prove that a $k$-colorable graph that has small variance as per \cref{lemma: expansionimpliesvariance} also has $\tilde{A}$ that is close to $W = Z M(\tilde{A}, \chi) B^{-1} Z^\top$ on the indicators of the color classes.

\begin{lemma}[Small variance implies close indicators]
\label[lemma]{lem:small-var-to-indic}
Let $G$ be an $n$-vertex graph with adjacency matrix $A$, and let $\tilde{A} = \frac{1}{d}A$ for some $d>0$.
Let $\chi: [n] \to [k]$ be a $k$-partition of its vertices with $\min_{a \in [k]} |\chi^{-1}(a)| \geq c n$, and consider its partition matrices $Z \in \{0,1\}^{n \times k}$ and $B \in \R^{k \times k}$ as per \cref{def:partition_Z}.
For $x \in [n]$ and $a \in [k]$, let $\D{x}{a} = \sum_{y \in \chi^{-1}(a)} \tilde{A}_{xy}$.
Suppose that, for any $a, b \in [k]$ and $\bm{x}$ uniformly random in $\chi^{-1}(a)$,
\[\E_{\bm{x} \sim \chi^{-1}(a)} \Paren{\D{\bm{x}}{b} - \E_{\bm{x} \sim \chi^{-1}(a)} \D{\bm{x}}{b}}^2 \leq \delta\,.\]
Then, for $W = Z M(\tilde{A}, \chi) B^{-1} Z^\top$ with $M(\tilde{A}, \chi)$ defined as in \cref{def:model} with respect to $\tilde{A}$, it holds for all $a \in [k]$ that
\[\Norm{\tilde{A} \frac{Z_a}{\Norm{Z_a}} - W  \frac{Z_a}{\Norm{Z_a}}}^2 \leq \frac{\delta}{c}\,.\]
\end{lemma}
\begin{proof}
Note that $W_{xy} = \frac{E(\chi^{-1}(\chi(x)), \chi^{-1}(\chi(y)))}{d \Abs{\chi^{-1}(\chi(x))} \Abs{\chi^{-1}(\chi(y))}}$.
Then we get 
\begin{align*}
    \Norm{\tilde{A} Z_a - W Z_a}^2 
    &= \sum_{b \in [k]} \sum_{x \in \chi^{-1}(b)} \Paren{\sum_{y \in \chi^{-1}(a)} \tilde{A}_{x,y} - \frac{1}{d} \frac{E(b,a)}{|\chi^{-1}(b)|} }^2\\
    &= \sum_{b \in [k]} \sum_{x \in \chi^{-1}(b)} \Paren{\D{x}{a} - \E_{\bm{x} \sim \chi^{-1}(b)} \D{\bm{x}}{a}}^2\\
    &= \sum_{b \in [k]} |\chi^{-1}(b)|\; \E_{\bm{x} \sim \chi^{-1}(b)} \Paren{\D{\bm{x}}{a} - \E_{\bm{x} \sim \chi^{-1}(b)} \D{\bm{x}}{a}}^2\\
    &\leq \delta n\,,
\end{align*}
so using that $cn \leq \Norm{Z_a}^2 \leq n$ we get that $\Norm{\tilde{A} \frac{Z_a}{\Norm{Z_a}} - W \frac{Z_a}{\Norm{Z_a}}}^2 \leq \frac{\delta}{c}$.
\end{proof}

Let us now prove \cref{thm:kalg-partial}.

\begin{proof}[Proof of \cref{thm:kalg-partial}]
    We note by \cref{lemma:model-color-size} that $\pi$ is unique and satisfies $\min_{a \in [k]} \pi_a \geq \Omega_M(1)$.

    Denote by $\chi: [n] \to [k]$ the $k$-partition with respect to which $G$ is $(M, \delta)$-colorable.
    Let $\phi: [k] \to [k']$ map each $a \in [k]$ to the $i \in [k']$ for which $a \in S_i$, and then define $\chi' : [n] \to [k']$ as $\chi'(x) = \phi(\chi(x))$ for all $x \in [n]$.
    That is, $\chi'$ is a $k'$-partition of the graph that starts with the $k$-partition $\chi$ and then for each $i \in [k']$ contracts all colors in $S_i$ into a single color.

    We aim to apply \cref{thm:part-recovery} with respect to the $k'$-partition $\chi'$ and approximately recover the partition of the vertices corresponding to $\chi'$.
    For that, we need to argue that the $k'$-partition $\chi'$ satisfies the conditions of \cref{thm:part-recovery}.
    Afterward, a simple rounding step will lead to the desired guarantee.

    Let $A$ be the adjacency matrix of the graph, and let $\tilde{A} = \frac{1}{d}A$.
    We define the model matrices $M(\tilde{A}, \chi)$ and $M(\tilde{A}, \chi')$ as in \cref{def:model}.
    We will apply \cref{thm:part-recovery} with matrices $M(\tilde{A}, \chi')$ and $\tilde{A}$ (to play the role of $M$ and $A$ in \cref{thm:part-recovery}, respectively).

    We will use the following two claims about the relation between the models defined with respect to $\chi$ and $\chi'$.
    First, we relate the entries of $M(\tilde{A}, \chi)$ and $M(\tilde{A}, \chi')$.

    \begin{claim}[Closeness of $M(\tilde{A}, \chi)$ and $M(\tilde{A}, \chi')$]
    \label[claim]{claim:close-entries-M}
    For any $i, j \in [k']$, we have 
    \begin{equation*}
    M(\tilde{A}, \chi')_{ij} = \frac{\sum_{b \in S_j} |\chi^{-1}(b)|}{|\chi^{-1}(b^*_j)|} M(\tilde{A}, \chi)_{a^*_ib^*_{j}} \pm O_M(\delta)\,.
    \end{equation*}
    \end{claim}
    \begin{proof}
    A calculation shows that, for $i, j \in [k']$,
    \begin{equation}
    \label{eq:model-form}
    M(\tilde{A}, \chi')_{ij} = \sum_{a \in S_i} \frac{|\chi^{-1}(a)|}{\sum_{a' \in S_i} |\chi^{-1}(a')|} \sum_{b \in S_j} M(\tilde{A}, \chi)_{ab}\,.
    \end{equation}
    We note that, for any two rows $M^{(a)}$ and $M^{(a')}$ that are identical, we have that\linebreak $\Norm{M(\tilde{A}, \chi)^{(a)} - M(\tilde{A}, \chi)^{(a')}} \leq O_M(\delta)$.
    Hence, for any $a, a' \in S_i$, we have that $\sum_{b \in S_j} M(\tilde{A}, \chi)_{ab}$ is $O_M(\delta)$-close to $\sum_{b \in S_j} M(\tilde{A}, \chi)_{a'b}$.
    Denote $b_i^* = a_i^*$ for all $i \in [k']$.
    Then, by \cref{eq:model-form} and the above,
    \begin{align*}
    \begin{split}
    \label{eq:mchi}
    M(\tilde{A}, \chi')_{ij}
    &= \sum_{b \in S_j} M(\tilde{A}, \chi)_{a^*_i b} \pm O_M(\delta)\\
    &= \frac{1}{|\chi^{-1}(a^*_i)|} \sum_{b \in S_j} |\chi^{-1}(b)| M(\tilde{A}, \chi)_{b a^*_i} \pm O_M(\delta)\\
    &= \frac{\sum_{b \in S_j} |\chi^{-1}(b)|}{|\chi^{-1}(a^*_i)|} M(\tilde{A}, \chi)_{b^*_{j} a^*_i} \pm O_M(\delta)\\
    &= \frac{\sum_{b \in S_j} |\chi^{-1}(b)|}{|\chi^{-1}(b^*_j)|} M(\tilde{A}, \chi)_{a^*_ib^*_{j}} \pm O_M(\delta)\,,
    \end{split}
    \end{align*}
    where we used that $|\chi^{-1}(a)| M(\tilde{A},\chi)_{ab} = |\chi^{-1}(b)| M(\tilde{A},\chi)_{ba}$ for all $a, b \in [k]$.
    \end{proof}

    Second, let $W = Z M(\tilde{A}, \chi) B^{-1} Z^\top$, with the partition matrices $Z, B$ corresponding to the $k$-partition $\chi$, and let $W' = Z' M(\tilde{A}, \chi') (B')^{-1} (Z')^\top$, with the partition matrices $Z', B'$ corresponding to the $k'$-partition $\chi'$.
    Then we also relate the entries of $W$ and $W'$.

    \begin{claim}[Closeness of $W$ and $W'$]
    \label[claim]{claim:close-entries-W}
    For any $x, y \in [n]$, we have 
    \begin{equation*}
    W'_{xy} = W_{xy} \pm O_M(\delta/n)\,.
    \end{equation*}
    \end{claim}
    \begin{proof}
    A calculation shows that, for $x,y \in [n]$,
    \[W_{xy} = \frac{M(\tilde{A}, \chi)_{\chi(x)\chi(y)}}{|\chi^{-1}(\chi(y))|}\] 
    and, by \cref{claim:close-entries-M}, 
    \begin{align*}
    W'_{xy}
    &= \frac{M(\tilde{A}, \chi')_{\chi'(x)\chi'(y)}}{|(\chi')^{-1}(\chi'(y))|}\\
    &= \frac{ \frac{\sum_{b \in S_{\chi'(y)}} |\chi^{-1}(b)|}{|\chi^{-1}(b_{\chi'(y)})|} M(\tilde{A}, \chi)_{a_{\chi'(x)}b_{\chi'(y)}}}{\sum_{b \in S_{\chi'(y)}} |\chi^{-1}(b)|} \pm O_M(\delta/n)\\
    &= \frac{ M(\tilde{A}, \chi)_{a_{\chi'(x)}b_{\chi'(y)}}}{|\chi^{-1}(b_{\chi'(y)})|} \pm O_M(\delta/n)\,.
    \end{align*}
    In \cref{claim:close-entries-M} $a^*_{\chi'(x)}$ and $b^*_{\chi'(y)}$ could in fact be any arbitrary colors in $S_{\chi'(x)}$ and $S_{\chi'(y)}$, respectively.
    Then, letting $a_{\chi'(x)} = \chi(x)$ and $b_{\chi'(y)} = \chi(y)$, we get that $W'_{xy} = W_{xy} \pm O_M(\delta/n)$.
    \end{proof}

    Let us prove now that we satisfy the conditions of \cref{thm:part-recovery} with matrices $M(\tilde{A}, \chi')$ and $\tilde{A}$.
    First, we have trivially that $(B')^{1/2} M(\tilde{A}, \chi') (B')^{-1/2}$ is symmetric and that $\Norm{M(\tilde{A}, \chi')} = 1$.
    We also have that $|\chi^{-1}(a)| \geq \Omega_M(n)$ for all $a \in [k]$ by \cref{lemma:model-color-size}.
    
    \paragraph{Row separation of $M(\tilde{A}, \chi')$.}
    We prove now that the rows of $M(\tilde{A}, \chi')$ are separated.
    Consider two rows $M(\tilde{A}, \chi')^{(i)}$ and $M(\tilde{A}, \chi')^{(i')}$ with $i \neq i' \in [k']$.
    Because $M^{(a^*_i)} \neq M^{(a^*_{i'})}$, we have that $\Norm{M^{(a^*_i)} - M^{(a^*_{i'})}} \geq \Omega_M(1)$, so then by the closeness of $M$ and $M(\tilde{A}, \chi)$ we also have that $\Norm{M(\tilde{A}, \chi)^{(a^*_i)} - M(\tilde{A}, \chi)^{(a^*_{i'})}} \geq \Omega_M(1)$.
    Therefore, there exists some $b \in [k]$ such that\linebreak $\Abs{M(\tilde{A}, \chi)_{a^*_ib} - M(\tilde{A}, \chi)_{a^*_{i'}b}} \geq \Omega_M(1)$.
    By \cref{claim:close-entries-M} this implies that, for $j = \phi(b)$, we also have $\Abs{M(\tilde{A}, \chi')_{i j} - M(\tilde{A}, \chi')_{i' j}} \geq \Omega_M(1)$.
    Hence, $\Norm{M(\tilde{A}, \chi')^{(i)} - M(\tilde{A}, \chi')^{(i')}} \geq \Omega_M(1)$.

    \paragraph{Closeness of $\tilde{A}$ and $W'$.}
    By definition, we only consider $G$ and partitions with respect to which it has a vertex cover of all edges with both endpoints in the same part of size at most $\delta n$, so the number of edges with both endpoints in the same part is bounded by $\delta d n$.
    By \cref{lemma:model-color-size}, the size of each color class is lower bounded by some $\Omega_M(n)$.
    Then by \cref{lemma: expansionimpliesvariance} and \cref{lem:small-var-to-indic} we have that $\Norm{\tilde{A} \frac{Z_a}{\Norm{Z_a}} - W \frac{Z_a}{\Norm{Z_a}}}^2 \leq O_M\Paren{\lambda + \frac{\delta}{\lambda}}$ for all $a \in [k]$.

    By \cref{claim:close-entries-W}, we have that $\Norm{W'-W} \leq \Norm{W'-W}_F \leq O_M(\delta)$.
    Then, for each $S_i$, denoting its elements as $S_i = \{a_1, \ldots, a_{\ell}\}$, we are interested in the closeness of $\tilde{A}$ and $W'$ with respect to the indicator $\sum_{j=1}^{\ell} Z_{a_j}$.
    We have
    \[\Norm{(\tilde{A} - W') \sum_{j=1}^{\ell} Z_{a_j}} \leq \sum_{j=1}^{\ell} \Norm{(\tilde{A} - W') Z_{a_j}} \leq \sum_{j=1}^{\ell} \Norm{(\tilde{A} - W) Z_{a_j}} + \sum_{j=1}^{\ell}\Norm{(W'-W)Z_{a_j}}\,.\]
    For the first term on the right-hand side, we have that $\Norm{\tilde{A} \frac{Z_a}{\Norm{Z_a}} - W \frac{Z_a}{\Norm{Z_a}}}^2 \leq O_M\Paren{\lambda + \frac{\delta}{\lambda}}$ for all $a \in [k]$, so using that $\Norm{Z_a} \leq \sqrt{n}$ we can bound
    \[\sum_{j=1}^{\ell} \Norm{(\tilde{A} - W) Z_{a_j}} \leq \sqrt{n} \cdot O_M\Paren{\sqrt{\lambda + \frac{\delta}{\lambda}}}\,.\]
    For the second term, we have that
    \[\sum_{j=1}^{\ell} \Norm{(W'-W)Z_{a_j}} \leq \sum_{j=1}^{\ell} \Norm{W'-W} \Norm{Z_{a_j}} \leq \sqrt{n} \cdot O_M(\delta)\,.\]
    Then, also using that $\Norm{\sum_{j=1}^{\ell} Z_{a_j}} \geq \Omega_M(\sqrt{n})$, we have
    \[\Norm{(\tilde{A} - W') \frac{\sum_{j=1}^{\ell} Z_{a_j}}{\Norm{\sum_{j=1}^{\ell} Z_{a_j}}}}^2 \leq O_M \Paren{\lambda + \frac{\delta}{\lambda}}\,,\]
    so we satisfy the second condition in \cref{thm:part-recovery} with constant $O_M \Paren{\lambda + \frac{\delta}{\lambda}}$.

    \paragraph{Low rank of $\tilde{A}$.}
    Because $\rank_{\geq \lambda}(\tilde{A}) \leq 1$, by applying \cref{cor:spectral_small_to_small} with $\sigma=\lambda$, we get that the bottom threshold rank of $\tilde{A}$ is bounded by $\rank_{\leq -\sqrt{2\lambda}}(\tilde{A}) \leq 1/\lambda^2$.
    Note that, using that $\lambda \leq 1$, we have that $\lambda \leq \sqrt{2 \lambda}$.
    Then, for the purposes of \cref{thm:part-recovery}, we apply the theorem with eigenvalue threshold $\lambda' = \max\{\sqrt{2\lambda}, \Omega_M(1)\}$ for some small enough $\Omega_M(1)$ such that the conditions of the theorem are satisfied, and this allows us to take $t = O(1/(\lambda')^2) = O_M(1)$.

    \paragraph{Rounding.}
    Then, by \cref{thm:part-recovery} applied with matrices $M(\tilde{A}, \chi')$ and $\tilde{A}$, we can recover in time $n^{O(1)} \cdot O_{\delta,\lambda,M}(1)$ a list of $O_{\delta,\lambda,M}(1)$ $k'$-partitions $\hat{\chi'}: [n] \to [k']$ such that, for at least one of them, it holds up to a permutation of the color classes that $\mathbb{E}_{\bm{x} \sim \operatorname{Unif}([n])} \1[\hat\chi'(x) \neq \chi'(x)] \leq O_M(\lambda+\delta/\lambda)$.
    Let $\e = O_M(\lambda + \delta/\lambda)$ be equal to this upper bound.
    At this point, for each candidate $k'$-partition, we apply the rounding algorithm in \cref{lemma: generalrounding} on each of its color classes to obtain $k'$ independent sets, and return the set of $k'$ independent sets that cover most vertices of $G$.

    Let us analyze the approximation guarantee.
    For the best $k'$-partition, the rounding algorithm in \cref{lemma: generalrounding} will remove at most $2\e n$ vertices due to misclassifications and at most $2\delta n$ vertices due to edges with both endpoints in the same part with respect to the $k$-partition $\chi$.
    The rounding algorithm may also need to remove a significant portion of vertices because we recover $k'$-partitions instead of $k$-partitions --- this is necessary even if the $k'$-partition $\chi'$ is perfectly recovered.

    To analyze these removals, let us start by defining $\bar{a}_i^* = \argmax_{a \in S_i} |\chi^{-1}(a)|$ for each $i \in [k']$.
    Then, for each $i \in [k']$, the rounding algorithm in \cref{lemma: generalrounding} may remove at most\linebreak $\min\Paren{\sum_{a \in S_i} |\chi^{-1}(a)|, 2\sum_{\substack{a \neq a_i^{*} \in S_i}} |\chi^{-1}(a)|}$ vertices from $(\chi')^{-1}(i)$.
    Then for each $i \in [k']$ the number of surviving vertices is at least $\bar{p}_i = \max\Paren{0, |\chi^{-1}(a_i^*)| - \sum_{\substack{a \neq a_i^{*} \in S_i}} |\chi^{-1}(a)|}$.
    Then the algorithm returns a $(k', 1 - \frac{1}{n} \sum_{i \in [k']} \bar{p}_i + 2\delta + 2\e)$-coloring for $G$.
    Finally, by \cref{lemma:model-color-size}, for $\delta > 0$ small enough, $\frac{1}{n} |\chi^{-1}(a)|$ and $\pi_a$ are $O(\sqrt{\delta}k)$-close for all $a \in [k]$, so the algorithm also returns a $(k', 1 - \sum_{i \in [k']} p_i + O_M(\sqrt{\delta}) + 2\e)$-coloring for $G$, where $p_i$ is defined as in the statement of the theorem.
    Recalling that $\e = O_M(\lambda + \delta/\lambda)$ gives the stated result.
\end{proof}

\subsection{$3$-colorable models}

For the case $k=3$, it turns out that the model has distinct rows whenever its stationary distribution is strictly bounded pointwise by $1/2$. 

The first observation is that we can compute explicitly the entries of $M$ in terms of the stationary distribution.

\begin{lemma}[Model entries for $3$-colorings]
\label[lemma]{lemma: 3coloringfixedaveragedegree}
Let $M \in \R^{3 \times 3}_{\geq 0}$ be a reversible row stochastic matrix with zeros on the diagonal and stationary distribution $\pi$.
Then for all $a \neq b \in [3]$ 
\[M_{ab} = \frac{(2\pi(a)+2\pi(b))-1}{2\pi(a)}.\]
\end{lemma}
\begin{proof}
    For all $a \in [3]$, we have $\sum_{b\in [3]}M_{ab}=1$.
    Also for all $b\in [3]\setminus\Set{a}$, we have $\pi(a) M_{ab} = \pi(b) M_{ba}$.
    Hence, taking the off-diagonal entries of the $M$ matrix as unknown, we have $6$ variables and $6$ independent linear equations, so there is a unique solution.
    Using $\sum_{a\in[3]}\pi(a)=1$, for $a\neq b\neq c\in[3]$ we can express it as
    $$M_{ab}= \frac{(2\pi(a)+2\pi(b))-1}{2\pi(a)}\,.$$
\end{proof}

We then use this observation to deduce that, whenever the largest probability of a color class is upper bounded by $1/2-\gamma$, any two distinct rows of $M$ have separation $\gamma$.

\begin{lemma}[Row separation for $3$-colorings]
\label[lemma]{lemma: 3coloringrowseparation}
Let $M \in \R^{3 \times 3}_{\geq 0}$ be a reversible row stochastic matrix with zero on the diagonal and stationary distribution $\pi$ that satisfies $\max_{a \in [3]} \pi_a \leq 1/2 - \gamma$.
Then any two distinct rows $M^{(a)}$ and $M^{(b)}$ of $M$ satisfy $\Norm{M^{(a)} - M^{(b)}} \geq \gamma$.
\end{lemma}
\begin{proof}
It suffices to prove that all non-diagonal entries of $M$ are bounded away from $0$.
From \cref{lemma: 3coloringfixedaveragedegree}, we get for $a \neq b$ that $M_{ab} = \frac{(2\pi_a+2\pi_b)-1}{2\pi_a}$.
Because $\max_{a \in [3]} \pi_a \leq 1/2-\gamma$, it follows that $\pi_a+\pi_b \geq 1/2+\gamma$, so $M_{ab} \geq \frac{\gamma}{\pi_a} \geq \gamma$.
Then the conclusion follows because $\Norm{M^{(a)} - M^{(b)}} \geq M_{ab} - M_{bb} = M_{ab} \geq \gamma$.
\end{proof}

Finally, we state and prove our result specialized to $3$-colorings.

\begin{theorem}[Result for $3$-coloring]
\label{thm:3alg}
Let $G$ be an $n$-vertex $d$-regular graph with adjacency matrix $A$, and let $\tilde{A} = \frac{1}{d}A$.
Let $\lambda$ be the second largest eigenvalue of $\tilde{A}$.
Let $\chi: [n] \to [3]$ be a $3$-partition of its vertices, and suppose that $G$ has a vertex cover of all edges with both endpoints in the same part of size at most $\delta n$ with respect to $\chi$.
Suppose $\max_{a \in [3]} \frac{1}{n} |\chi^{-1}(a)| \leq 1/2 - \gamma$.
Then, for $\delta, \gamma, \lambda > 0$ small enough, there exists an algorithm that runs in time $n^{O(1)} \cdot O_{\delta,\gamma,\lambda}(1)$ and outputs a $(3, O(\lambda/\gamma^8+\delta/(\lambda\gamma^8)))$-coloring for $G$.
\end{theorem}
\begin{proof}
For simplicity of notation, let $M = M(\tilde{A}, \chi)$ be defined as in \cref{def:model} with respect to $\tilde{A}$.
We aim again to apply \cref{thm:part-recovery} with matrices $M$ and $\tilde{A}$, and we verify that its conditions hold.

\begin{claim}
\label[claim]{claim:3-row-sep}
If $\delta \ll \gamma^2$, then any two rows $M^{(a)}$ and $M^{(b)}$ of $M$ satisfy $\Norm{M^{(a)} - M^{(b)}} \geq \Omega(\gamma)$.
\end{claim}
\begin{proof}
We first argue that $M$ is close to a reversible row stochastic matrix with zero on the diagonal, to which we will apply \cref{lemma: 3coloringrowseparation}.
Let $M' \in \R^{3 \times 3}$ be obtained from $M$ by setting the diagonal entries to zero and by rescaling each row such that it sums up to $1$.
Then $M'$ is row stochastic with zero on the diagonal.
Let us now understand the stationary distribution of $M'$.
Let $\pi$ be the stationary distribution of $M$, which coincides with $\Set{\frac{1}{n} |\chi^{-1}(a)|}_{a \in [3]}$.
Then $\pi_{a} M_{ab} = \pi_b M_{ba}$ for all $a \neq b \in [3]$, so by construction also $\pi_a \Paren{\sum_{c \neq a} M_{ac}} M'_{ab} = \pi_b \Paren{\sum_{c \neq b} M_{bc}} M'_{ba}$ for all $a \neq b \in [3]$.
Therefore, the stationary distribution of $M'$ has probability mass function proportional to $\Set{\pi_a \Paren{\sum_{c \neq a} M_{ac}}}_{a \in [3]}$, normalized to sum up to $1$.
Finally, we note that $\max_{a \in [3]} M_{aa} \leq O(\delta / \gamma)$, which also implies that $\sum_{c \neq a} M_{ac} \geq 1 - O(\delta / \gamma)$.
Then $\Norm{M - M'} \leq O(\delta / \gamma)$, and also the stationary distribution of $M'$ has maximum probability mass bounded by $1/2 - \gamma + O(\delta/\gamma) \leq 1/2 - \Omega(\gamma)$.
By \cref{lemma: 3coloringrowseparation}, any two distinct rows of $M'$ have separation $\Omega(\gamma)$.
Then, by closeness of $M$ and $M'$, we also have that any two rows of $M$ have separation $\Omega(\gamma) - O(\delta/\gamma) \geq \Omega(\gamma)$.
\end{proof}

Note that we can assume that $\delta \ll \gamma^2$, as the guarantees of the theorem are vacuous otherwise, so we can apply \Cref{claim:3-row-sep}.
For the rest of the proof, we will largely repeat the analysis in the proof of \cref{thm:kalg-partial}, but using that any two distinct rows of $M$ are separated and taking into account explicitly the parameter $\gamma$.
Let $W = Z M B^{-1} Z^\top$, with the partition matrices $Z, B$ corresponding to the $3$-partition $\chi$.
We have trivially that $B^{1/2} M B^{-1/2}$ is symmetric and that $\Norm{M} = 1$.
We also have that $|\chi^{-1}(a)| \geq 2\gamma n$ for all $a \in [3]$.

By \Cref{claim:3-row-sep} we have that any two distinct rows of $M$ have separation $\Omega(\gamma)$.
For the closeness of $\tilde{A}$ and $W$, we have by \cref{lemma: expansionimpliesvariance} and \cref{lem:small-var-to-indic} that $\Norm{\tilde{A} \frac{Z_a}{\Norm{Z_a}} - W \frac{Z_a}{\Norm{Z_a}}}^2 \leq O\Paren{\lambda/\gamma^2 + \frac{\delta}{\gamma^2 \lambda}}$ for all $a \in [3]$.
For the low rank of $\tilde{A}$, we have by applying \cref{cor:spectral_small_to_small} with $\sigma=\lambda$ that $\rank_{\leq -\sqrt{2\lambda}}(\tilde{A}) \leq 1/\lambda^2$, and also $\lambda \leq \sqrt{2 \lambda}$. 

Then we apply \cref{thm:part-recovery} with matrices $M$ and $\tilde{A}$ and eigenvalue threshold\linebreak $\lambda' = \max\{\sqrt{2 \lambda}, \Omega(\gamma^{3/2})\}$, such that it satisfies the condition of the theorem, and such that we can take $t = O(1/(\lambda')^2) = O(1/\gamma^3) = O_{\gamma}(1)$.
Then, we can recover in time $n^{O(1)} \cdot O_{\delta,\gamma,\lambda}(1)$ a list of $O_{\delta,\gamma,\lambda}(1)$ $3$-partitions $\hat{\chi}: [n] \to [3]$ such that, for at least one of them, it holds up to a permutation of the color classes that $\mathbb{E}_{\bm{x} \sim \operatorname{Unif}([n])} \1[\hat\chi(x) \neq \chi(x)] \leq O(\lambda/\gamma^8+\delta/(\lambda\gamma^8))$.
As in the proof of \cref{thm:kalg-partial}, for each candidate $3$-partition we apply the rounding algorithm in \cref{lemma: generalrounding} on each of its color classes to obtain $3$ independent sets, and return the set of $3$ independent sets that cover most vertices of $G$.
The rest of the analysis is identical to that in the proof of \cref{thm:kalg-partial}.
\end{proof}

\subsection{Random planting}

Instead of a regular $k$-colorable one-sided expander graph, in this section we consider randomly planting a \kcl in a regular one-sided expander graph.
To aid presentation, let us first define the model and some of the notation that we will use throughout this section.

\begin{definition}[Randomly planted \kcl] \label[definition]{def: randomplanted}
Let $H$ be an $n$-vertex $d$-regular graph. 
Then we say the graph $G$ is obtained by randomly planting a \kcl in $H$ if it is sampled according to the following procedure:
\begin{enumerate}
    \item Sample $\chi: [n] \to [k]$ by sampling $\chi(x) \sim \operatorname{Unif}([k])$ independently for each $x \in [n]$,
    \item Let $G$ be a copy of $H$, and then remove from $G$ each edge $\{u,v\}$ for which $\chi(u) = \chi(v)$.
\end{enumerate}
\end{definition}

The following is our main recovery result.
We note that in \cref{thm: random-planting-full}, for the case when the host graph is a one-sided expander, we give an algorithm that recovers a coloring on \emph{all} vertices of the graph.

\begin{theorem}[Result for random planting, partial recovery]
    \label{thm: random-planting-partial}
    Let $H$ be an $n$-vertex $d$-regular graph with adjacency matrix $A_H$.
    Suppose $\rank_{\geq \lambda}(\frac{1}{d}A_H) \leq t$ for some $0 < \lambda < 1$ and $t \in \mathbb{N}$.
    Let $G$ be the graph obtained by randomly planting a \kcl in $H$, and let $\chi: [n] \to [k]$ be the associated \kcl.
    Then, for $\lambda > 0$ small enough, there exists an algorithm that runs in time $n^{O(1)} \cdot \Paren{O_{k}(1)}^{t}$ and outputs with high probability a $(k, O_k(1/d^{\Omega(1)}))$-coloring for $G$.
\end{theorem}

To prove \cref{thm: random-planting-partial}, we first prove some properties of the graph $G$ obtained after randomly planting a \kcl in $H$.

\begin{lemma}[Color class concentration in random planting]
    \label[lemma]{fact: balancedparts}
    Let $H$ be an $n$-vertex graph.
    Let $G$ be the graph obtained by randomly planting a \kcl in $H$ for some $k=O(1)$, and let $\chi: [n] \to [k]$ be the associated \kcl.
    Then, with high probability, $|\chi^{-1}(a)| = \frac{n}{k} \pm O\Paren{\sqrt{n \log n}}$ for all $a \in [k]$.
\end{lemma}
\begin{proof}
    The proof follows straightforwardly from a Hoeffding bound and a union bound.
\end{proof}

\begin{lemma}[Model entry concentration in random planting]\label[lemma]{lemma: plantedmodelconcentration}
    Let $H$ be an $n$-vertex $d$-regular graph.
    Let $G$ be the graph obtained by randomly planting a \kcl in $H$ for some $k=O(1)$, and let $\chi: [n] \to [k]$ be the associated \kcl.
    Let $A \in \R^{n \times n}$ be the adjacency matrix of $G$, and let $\tilde{A} = \frac{k}{(k-1)d} A$.
    Let $M(\tilde{A}, \chi)$ be defined as in \cref{def:model}.
    Then, with high probability, $M_{ab}=\frac{1}{k-1} \pm o(1)$ for all $a\neq b\in[k]$.
\end{lemma}

\begin{proof}
    Let $d' = \frac{(k-1)d}{k}$, and for $x \in [n]$ and $a \in [k]$, let $\D{x}{a} = \sum_{y \in \chi^{-1}(a)} \tilde{A}_{xy}$.
    By linearity of expectation,
    \[
    \begin{aligned}
    \E\Brac{\sum_{x\in \chi^{-1}(a)} \D{x}{b}}
    &= \frac1{d'}\;\E\Brac{\sum_{x\sim y\in E(H)}
        \Paren{\1[\chi(x)=a\land\chi(y)=b]
            +\1[\chi(x)=b\land\chi(y)=a]}} \\[6pt]
    &= \frac1{d'}\;\frac{nd}{2}\;\frac{2}{k^2}
    \;=\;\frac{n}{k(k-1)}\,.
    \end{aligned}
    \]

    By the $d$-regularity of $H$, changing the color of a single vertex can increase or decrease $\sum_{x\in \chi^{-1}(a)} \D{x}{b}$ by at most $O(1)$. Hence, McDiarmid's inequality gives that $\sum_{x\in \chi^{-1}(a)} \D{x}{b}$ is with probability $O(n^{-2})$ within $O(\sqrt{n\log{n}})$ of its expectation. Using \cref{fact: balancedparts} and putting it all together, we get
    $$M_{ab}=\frac{\frac{n}{k(k-1)}\pm O(\sqrt{n\log{n}})}{\frac{n}{k}\pm O(\sqrt{n\log{n}})}=\frac1{k-1}\pm O\Paren{\sqrt{\frac{\log n}{n}}}$$
    with high probability for all $a \neq b \in [k]$.
\end{proof}

Next, we show that with high probability the graph $G$ has a good variance bound.

\begin{lemma}[Variance of random planting]
    \label[lemma]{lemma: randomplantinglowvariance}
    Let $H$ be an $n$-vertex graph.
    Let $G$ be the graph obtained by randomly planting a \kcl in $H$ for some $k=O(1)$, and let $\chi: [n] \to [k]$ be the associated \kcl.
    Let $A \in \R^{n \times n}$ be the adjacency matrix of $G$, and let $\tilde{A} = \frac{k}{(k-1)d} A$.
    For $x \in [n]$ and $a \in [k]$, let $\D{x}{a} = \sum_{y \in \chi^{-1}(a)} \tilde{A}_{xy}$.
    Then, with high probability, for any $a, b \in [k]$ and $\bm{x}$ uniformly random in $\chi^{-1}(a)$, 
    \[\E_{\bm{x} \sim \chi^{-1}(a)} \Paren{\D{\bm{x}}{b} - \E_{\bm{x} \sim \chi^{-1}(a)} \D{\bm{x}}{b}}^2 \leq 1/d^{\Omega(1)}\,.\]
\end{lemma}

\begin{proof}
    If $a=b$, we have $\E_{\bm{x} \sim \chi^{-1}(a)} \Paren{\D{\bm{x}}{b} - \E_{\bm{x} \sim \chi^{-1}(a)} \D{\bm{x}}{b}}^2 = 0$ by construction, so we assume $a\neq b$.
    
    Writing $\D{\bm{x}}{b} - \E_{\bm{x} \sim \chi^{-1}(a)} \D{\bm{x}}{b}$ as $\Paren{\D{\bm{x}}{b} - \frac1{k-1}} + \Paren{\frac1{k-1} - \E_{\bm{x} \sim \chi^{-1}(a)} \D{\bm{x}}{b}}$ and using that $(a+b)^2\leq 2(a^2+b^2)$, we get
    \begin{equation}\label{eq:variance_split}
        \E_{\bm{x} \sim \chi^{-1}(a)} \Bigl(\D{\bm{x}}{b} - \E_{\bm{x} \sim \chi^{-1}(a)} \D{\bm{x}}{b}\Bigr)^2
        \;\le\;
        2\Bigl(\E_{\bm{x} \sim \chi^{-1}(a)} \bigl(\D{\bm{x}}{b} - \tfrac1{k-1}\bigr)^2
        +\bigl(\tfrac1{k-1}-M_{ab}\bigr)^2\Bigr).
    \end{equation}

    Let us define $S=\sum_{x \in \chi^{-1}(a)} \Paren{\D{x}{b} - \frac1{k-1}}^2$ to analyze the first term. 
    Let $d' = \frac{(k-1)d}{k}$.
    For any $x\in V(G)$, conditioned on $x\in \chi^{-1}(a)$, by $d$-regularity we have $\D{x}{b}\sim \frac1{d'}\mathrm{Bin}(d, \frac1{k})$.
    Hence, by the law of total probability and linearity of expectation, we get
    $$\E_{\chi}\Brac{S}=\frac{n}{k}\frac1{d'^2}d\frac1{k}\Paren{1-\frac1{k}}=\frac{n}{(k-1)kd} \leq O(n/d)\,.$$ 

    It is easy to verify that changing the color of a single vertex can increase or decrease $S$ by at most $O(1)$. Therefore, by McDiarmid's inequality,
    \[
    \Pr\Brac{\,\bigl|S-\E[S]\bigr|\ge \sqrt{n\log n}}
    \;\le\;
    2\exp\!\Paren{-\Omega\!\Paren{\tfrac{n\log n}{n}}}
    =o(1).
    \]
    Hence, with high probability
    \[
    S = \E[S]\pm O\Paren{\sqrt{n\log n}}
    = O(n/d) \pm O\Paren{\sqrt{n\log n}}.
    \]
    Then, using \cref{fact: balancedparts}, with high probability
    \[\E_{\bm{x} \sim \chi^{-1}(a)} \Paren{\D{\bm{x}}{b} - \frac1{k-1}}^2
    =\frac{S}{|\chi^{-1}(a)|}
    =O\!\Paren{1/d}\;+\;O\!\Paren{\sqrt{\tfrac{\log n}{n}}}
    = 1/d^{\Omega(1)}.
    \]
    
    By \cref{lemma: plantedmodelconcentration}, the second term of \cref{eq:variance_split} is also bounded by $\Paren{\frac1{k-1}-M_{ab}}^2 \leq 1/d^{\Omega(1)}$, finishing the proof.
\end{proof}

The last ingredient to prove \cref{thm: random-planting-partial} is a lemma that shows how low threshold rank is approximately preserved after planting.

\begin{lemma}[Threshold rank kept after planting]\label[lemma]{lemma: threshold_rank_kept}
    Let $H$ be an $n$-vertex graph with adjacency matrix $A_H$.
    Let $G$ be the graph obtained by planting a \kcl in $H$ (not necessarily randomly), and let $A_G$ be its adjacency matrix.
    Let $\bar{A}_H:=\frac1{c}A_H$ and $\bar{A}_G:=\frac1{c}A_G$ be adjacency matrices of $H$ and $G$ respectively normalized by the same parameter $c > 0$.
    For $r_1, r_2 \geq 0$ and $t_1, t_2 \geq 1$, assuming that
    \[
    \rank_{\ge r_1}(\bar{A}_H)\le t_1
    \quad\text{and}\quad
    \rank_{\le -r_2}(\bar{A}_H)\le t_2,
    \]
    it holds that
    \[
    \rank_{\le-(r_1+r_2)}(\bar{A}_G)\;\le\;k\,t_1 + t_2,
    \quad
    \rank_{\ge \,r_1+r_2}(\bar{A}_G)\;\le\;t_1 + k\,t_2.
    \]
    \end{lemma}
    
    \begin{proof}
    Decompose \(\bar{A}_G=\bar{A}_H - \bar{A}_F\), where \(\bar{A}_F\) is block‐diagonal over the \(k\) color classes.  
    By Cauchy's interlacing theorem, each block has \(\le t_1\) eigenvalues \(\ge r_1\) and \(\le t_2\) eigenvalues \(\le -r_2\), so \(\rank_{\le -r_1}(-\bar{A}_F)=\rank_{\ge r_1}(\bar{A}_F)\le k\,t_1\) and \(\rank_{\ge r_2}(-\bar{A}_F)=\rank_{\le -r_2}(\bar{A}_F)\le k\,t_2\).
    Hence, applying Weyl's inequalities
    \begin{gather*}
        \rank_{\le -(r_1+r_2)}(\bar{A}_H - \bar{A}_F)\le\rank_{\le -r_2}(\bar{A}_H)+\rank_{\le -r_1}(-\bar{A}_F)\,, \\
        \rank_{\ge r_1+r_2}(\bar{A}_H - \bar{A}_F)\le\rank_{\ge r_1}(\bar{A}_H)+\rank_{\ge r_2}(-\bar{A}_F)\,,
    \end{gather*}
    we get the two claimed bounds.
    \end{proof}
        
We are now ready to prove \cref{thm: random-planting-partial}.

\begin{proof}[Proof of \cref{thm: random-planting-partial}]
    Let $A \in \R^{n \times n}$ be the adjacency matrix of $G$, and let $\tilde{A} = \frac{k}{(k-1)d} A$.
    Also let $d' = \frac{(k-1)d}{k}$.
    Let $M = M(\tilde{A}, \chi)$ be defined as in \cref{def:model}.
    We aim again to apply \cref{thm:part-recovery} with matrices $M$ and $\tilde{A}$, and we verify that its conditions hold.
    
    Let $W = Z M B^{-1} Z^\top$, with the partition matrices $Z, B$ corresponding to the $k$-partition $\chi$.
    We have trivially that $B^{1/2} M B^{-1/2}$ is symmetric.
    By \cref{lemma: plantedmodelconcentration} and the triangle inequality, $\Norm{M}\leq \Norm{\frac1{k-1}\Paren{\Allones-\Id}}+\Norm{M-\frac1{k-1}\Paren{\Allones-\Id}}=1+o(1)$, where $\Allones \in \R^{k \times k}$ is the matrix with all entries equal to $1$.
    We also have by \cref{fact: balancedparts} that with high probability $|\chi^{-1}(a)| \geq (1/k-o(1))n$ for all $a \in [k]$.
    
    In addition, as the diagonal entries of $M$ are zero and by \cref{lemma: plantedmodelconcentration} the off-diagonal entries are $\frac1{k-1}\pm o(1)$, we have $\Norm{M^{(a)} - M^{(b)}} \geq \sqrt{2}/k$ for any two rows $M^{(a)}$ and $M^{(b)}$ of $M$.
    For the closeness of $\tilde{A}$ and $W$, because of  \cref{fact: balancedparts} and \cref{lemma: randomplantinglowvariance}, we can apply \cref{lem:small-var-to-indic} and get that, for all $a\in [k]$,
    \[\Norm{\tilde{A} \frac{Z_a}{\Norm{Z_a}} - W  \frac{Z_a}{\Norm{Z_a}}}^2 = \frac{1}{d^{\Omega(1)} \Paren{\frac{1}{k}-o(1)}} \leq O_k\Paren{1/d^{\Omega(1)}}\,.\]

    For the low rank of $\tilde{A}$, we have $\rank_{\geq \lambda}(\frac{1}{d}A_H) \leq t$ by assumption, and as $\Norm{\frac1{d}A_H}\leq 1$ because of $d$-regularity, we can apply \cref{cor:spectral_small_to_small} with $\sigma=\lambda$ to get $\rank_{\leq -\sqrt{2\lambda}}(\frac{1}{d}A_H) \leq t/\lambda^2$ as well. 
    Note that, using that $\lambda \leq 1$, we have that $\lambda \leq \sqrt{2 \lambda}$.
    Let $\lambda' = \max\{\sqrt{2\lambda}, \Omega_k(1)\}$ for some small enough $\Omega_k(1)$, such that $\rank_{\geq \lambda'}(\frac{1}{d}A_H)=O_k(t)$ and $\rank_{\leq -\lambda'}(\frac{1}{d}A_H)=O_k(t)$.
    By applying \cref{lemma: threshold_rank_kept} with normalization parameter $d$, we get $\rank_{\geq 2\lambda'}\Paren{\frac1{d}A_G}=O_k(t)$ and $\rank_{\leq -2\lambda'}\Paren{\frac1{d}A_G}=O_k(t)$.
    As $k\geq 3$, we have $\frac{d}{d'}2\lambda'\leq 3\lambda'$, so $\rank_{\geq 3\lambda'}(\tilde{A})+\rank_{\leq -3\lambda'}(\tilde{A})=O_k(t)$.

    Then, by the above, with high probability we satisfy the conditions necessary to apply \cref{thm:part-recovery} with $1/k-o(1)$, $1+o(1)$, $2/k$, $O_k(1/d^{\Omega(1)})$, $O(t)$ and $O(\lambda')$ in the roles of $c$, $\zeta$, $\alpha$, $\delta$, $t$ and $\lambda$, respectively.
    Then, we can recover in time $n^{O(1)} \cdot \Paren{O_{k}(1)}^{t}$ a list of $\Paren{O_{k}(1)}^{t}$ $k$-partitions $\hat{\chi}: [n] \to [k]$ such that, for at least one of them, it holds up to a permutation of the color classes that $\mathbb{E}_{\bm{x} \sim \operatorname{Unif}([n])} \1[\hat\chi(x) \neq \chi(x)] \leq O_k(1/d^{\Omega(1)})$.
    Finally, as in the proof of \cref{thm:kalg-partial}, for each candidate $k$-partition we apply the rounding algorithm in \cref{lemma: generalrounding} on each of its color classes to obtain $k$ independent sets, and return the set of $k$ independent sets that cover most vertices of $G$.
    The rest of the analysis is identical to that in the proof of \cref{thm:kalg-partial}.
\end{proof}

\subsubsection{Rounding to a full coloring}
\label{sec:full-col}

Next, we show that if $H$ actually is a one-sided expander (instead of having just low top threshold rank), and if the degree $d$ is at least a large enough constant, then we can actually obtain a full \kcl by exploiting the full randomness of the model.

We start from the guarantees of \cref{thm: random-planting-partial-list}, which is a version of \cref{thm: random-planting-partial} that omits the last rounding step: instead, it guarantees the recovery of a \emph{list} of candidate $k$-colorings, such that at least one of them is close to the planted $k$-coloring.
We remark that \cref{thm: random-planting-partial} itself does not guarantee that the output $k$-coloring is close to the planted $k$-coloring, while our strategy for obtaining a full coloring relies on this property --- hence our use of \cref{thm: random-planting-partial-list}.
Then, we show how to convert this list of partial colorings to a full coloring.
Our main result is the following:

\begin{theorem}[Result for random planting, full recovery]
    \label{thm: random-planting-full}
    Let $H$ be an $n$-vertex $d$-regular graph with adjacency matrix $A_H$.
    Let $\lambda$ be the second largest eigenvalue of $\frac{1}{d}A_H$.
    Let $G$ be the graph obtained by randomly planting a \kcl in $H$, and let $\chi: [n] \to [k]$ be the associated \kcl.
    Then, for $\lambda > 0$ small enough and $d$ at least a large enough constant compared to $k$, there exists an algorithm that runs in time $n^{O(1)} \cdot O_{k}(1)$ and outputs with high probability a $(k, 0)$-coloring for $G$.
\end{theorem}

The proof and algorithm for \cref{thm: random-planting-full} adapt Theorem B.10 of \cite{MR3536556-David16}. For full details, see \cref{app: plantedcoloring}; below we outline the key modifications.

The key difference between our setting and that of \cite{MR3536556-David16} is that they assume two-sided expansion, while we only assume one-sided expansion.
The two ways in which Theorem B.10 of \cite{MR3536556-David16} uses two-sided expansion are via the Expander Mixing Lemma (EML) (used in Lemma B.17, Lemma B.19, and Lemma B.20 of \cite{MR3536556-David16})
$$e(S, T)\leq d|S|\Paren{\frac{|T|}{n}+c\sqrt{\frac{|T|}{|S|}}}$$
and via the vertex expansion lower bound (used in Lemma B.21 of \cite{MR3536556-David16}): for $|S|=\alpha n$, 
$$\Abs{N(S)\setminus S}\geq \Paren{\frac{1}{(1-\alpha)c^2+\alpha}-1}|S|\,.$$

For the former, we note that Lemmas B.19 and B.20 of \cite{MR3536556-David16} use the EML with $S=T$, and that we can avoid Lemma B.17 of \cite{MR3536556-David16} by assuming a large enough (but still constant) $d$. As the lemma below shows, for $S=T$, the EML holds even under one-sided expansion:

\begin{lemma}[Edge density upper bound]
    \label[lemma]{lemma: emlonesided}
    Let $G$ be an $n$-vertex $d$-regular graph with adjacency matrix $A$ satisfying $\lambda_2(\frac{1}{d}A)\leq c$ for some $0<c<1$.
    Then, for any set $S \subseteq [n]$,
    $$\Abs{E(G_S)}\leq\frac{d|S|}{2}\Paren{\frac{|S|}{n}+c}\,,$$
    where $E(G_S)$ is the number of edges in the graph induced by $S$.
\end{lemma}

For the latter use of two-sided expansion, we can show under one-sided expansion a bound that is only slightly weaker, and suffices for our purposes by assuming slightly stronger expansion:

\begin{lemma}[Vertex expansion lower bound]
    \label[lemma]{lemma: ssveonesided}
    Let $G$ be an $n$-vertex $d$-regular graph with adjacency matrix $A$ satisfying $\lambda_2(\frac{1}{d}A)\leq c$ for some $0<c<1$.
    Then, for any set $S \subseteq [n]$ with $|S|\leq \alpha n$,
    $$\Abs{N(S)\setminus S}\geq \Paren{\frac{1}{c+\sqrt{\alpha}}-1}|S|\,,$$
    where $N(S)$ is the number of vertices outside $S$ adjacent to at least one vertex in $S$.
\end{lemma}

The proofs of \cref{lemma: emlonesided} and \cref{lemma: ssveonesided} are also in \cref{app: plantedcoloring}.

\subsection{Independent sets}

We state now our result for independent sets.

\begin{theorem}[Independent sets, small bottom rank]
\label[theorem]{thm:independentsetrecovery}
    Let $G$ be an $n$-vertex $d$-regular graph with adjacency matrix $A$, and let $\tilde{A} = \frac{1}{d}A$.
    For $0 < \gamma < 1/4$, suppose $G$ contains an independent set $I \subseteq [n]$ of size $\Paren{\frac{1}{2}-\gamma}n$.
    Suppose $\rank_{\leq -\lambda}(\tilde{A}) \leq t$ for some $0 < \lambda < 1$ and $t \in \mathbb{N}$.
    Then there exists an algorithm with runtime $\poly(n) \cdot (\gamma/(1-\lambda))^{-O(t)}$ that returns an independent set of size at least $\Paren{1/2-\gamma - (4+c)\gamma/(1-\lambda)}n$ for any constant $c > 0$.
\end{theorem}
    
By \cref{cor:spectral_small_to_small}, we immediately get a result under the assumption that the top threshold rank is small.

\begin{corollary}[Independent sets, small top rank]
\label[corollary]{cor:independentsetrecovery}
    Let $G$ be an $n$-vertex $d$-regular graph with adjacency matrix $A$, and let $\tilde{A} = \frac{1}{d}A$.
    For $0 < \gamma < 1/4$, suppose $G$ contains an independent set $I \subseteq [n]$ of size $\Paren{\frac{1}{2}-\gamma}n$.
    Suppose $\rank_{\geq \lambda}(\tilde{A}) \leq t$ for some $0 < \lambda < 1$ and $t \in \mathbb{N}$.
    Then there exists an algorithm with runtime $\poly(n) \cdot (\gamma/(1-\sqrt{2\lambda}))^{-O(t/\lambda^2)}$ that returns an independent set of size at least $\Paren{1/2-\gamma - (4+c)\gamma/(1-\sqrt{2\lambda})}n$ for any constant $c > 0$.
\end{corollary}
\begin{proof}
    Because $\rank_{\geq \lambda}(\tilde{A}) \leq t$, we get by applying \cref{cor:spectral_small_to_small} with $\sigma = \lambda$ that $\rank_{\leq -\sqrt{2\lambda}}(\tilde{A}) \leq t/\lambda^2$.
    The result then follows by \cref{thm:independentsetrecovery}.
\end{proof}

To prove \cref{thm:independentsetrecovery}, we first prove that the indicator vector of the independent set is close to the subspace spanned by the bottom $t$ eigenvectors of the adjacency matrix:

\begin{lemma}[Closeness to bottom eigenvectors]
    \label[lemma]{lemma:strongerprojectionargument}
    In the same setting as \cref{thm:independentsetrecovery}, let $u \in \R^n$ have $u_x = \frac{1}{\sqrt{n}}$ for $x \in I$ and $u_x = -\frac{1}{\sqrt{n}}$ for $x \not\in I$.
    Then, there exists a vector $\tilde{u} \in \R^n$ with $\norm{\tilde{u}} \leq \norm{u}$ and $||\tilde{u}-u||^2\leq \frac{4\gamma}{1-\lambda}$ such that $\tilde{u}$ is in the span of the bottom $t$ eigenvectors of $\tilde{A}$.
\end{lemma}
\begin{proof}
    First, we have 
    \[u^\top \tilde{A} u = \frac{2}{nd} \Paren{ E(\bar{I}, \bar{I}) - E(I, \bar{I}) } = \frac{2}{nd} \Paren{\gamma n d - \Paren{\frac{1}{2}-\gamma}nd } = -(1-4\gamma)\,.\]
    Second, we derive a lower bound on $u^\top \tilde{A} u$.
    Let $P \in \R^{n \times n}$ be the orthogonal projection to the subspace spanned by the bottom $t$ eigenvectors of $\tilde{A}$.
    Then, because $\rank_{\leq -\lambda}(\tilde{A}) \leq t$ and all eigenvalues of $\tilde{A}$ lie in $[-1, 1]$, we get
    \begin{align*}
    u^\top \tilde{A} u
    &\geq -\lambda \cdot \Paren{1-\Norm{Pu}^2} + (Pu)^\top \tilde{A} (Pu)\\
    &\geq -\lambda \cdot \Paren{1-\Norm{Pu}^2} - \Norm{Pu}^2\\
    &= -\lambda - (1-\lambda) \Norm{Pu}^2\,.
    \end{align*}
    Then, plugging in $u^\top \tilde{A} u = -(1-4\gamma)$ and rearranging, we get
    \[\Norm{Pu}^2 \geq \frac{1-\lambda-4\gamma}{1-\lambda} = 1 - \frac{4\gamma}{1-\lambda}\,.\]
    In particular, $\Norm{u - Pu}^2 = 1 - \Norm{Pu}^2 \leq \frac{4\gamma}{1-\lambda}$.
    Letting $\tilde{u} = Pu$ gives the desired conclusion.
\end{proof}

Furthermore, we can round a vector close to the indicator vector of the independent set:
\begin{lemma}[Independent set rounding]
    \label[lemma]{lemma:independentsetrounding}
    Let $G$ be an $n$-vertex graph with vertex set $V$ that contains an independent set $I \subseteq [n]$.
    Let $u \in \R^n$ have $u_x = \frac{1}{\sqrt{n}}$ for $x \in I$ and $u_x = -\frac{1}{\sqrt{n}}$ for $x \notin I$.
    Also let $\hat{u} \in \R^n$ satisfy $\Norm{\hat{u}-u}^2\leq \alpha$ for some $\alpha \geq 0$.
    Then, given as input $G$ and $\hat{u}$, there exists an algorithm with runtime $\poly(n)$ that returns an independent set of size at least $|I|-\alpha n$.
\end{lemma}
\begin{proof}
    Letting $\hat{I} = \Set{x \in V \mid \hat{u}_x \geq 0}$, by $\Norm{\hat{u}-u}^2\leq \alpha$ we have $\Abs{I\triangle \hat{I}}\Paren{\frac{1}{\sqrt{n}}}^2\leq \alpha$.
    Noting that $\hat{I}$ is a superset of the independent set $\hat{I}\cap I$, by applying \cref{lemma: generalrounding} on $\hat{I}$, we get an independent set of size at least $\Abs{\hat{I}\cap I} - \Abs{\hat{I}\setminus \Paren{\hat{I}\cap I}}=\Abs{I} - \Abs{I\triangle \hat{I}}\geq \Abs{I} - \alpha n$.
\end{proof}

Given \cref{lemma:strongerprojectionargument} and \cref{lemma:independentsetrounding}, we prove~\Cref{thm:independentsetrecovery}.

\begin{proof}[Proof of~\Cref{thm:independentsetrecovery}]
We iterate over a net of the subspace spanned by the bottom $t$ eigenvectors of $\tilde{A}$ in order to find a vector close to $u$ as defined in~\Cref{lemma:independentsetrounding}, and then use \cref{lemma:independentsetrounding} to round it to an independent set.

By standard arguments, we can construct a $c \sqrt{\gamma/(1-\lambda)}$-cover of the unit ball in the subspace spanned by the bottom $t$ eigenvectors of $\tilde{A}$ in time $(\gamma/(1-\lambda))^{-O(t)}$, for any constant $c > 0$.
The cover will have size $O(\sqrt{\gamma/(1-\lambda)})^{-t}$.
Then, we iterate over all $\hat{u}$ in this cover and for each of them we apply the rounding algorithm in \cref{lemma:independentsetrounding} to construct an independent set.
We return the largest independent set thus obtained.

For correctness, we note by~\Cref{lemma:strongerprojectionargument} that there exists some $\tilde{u}$ in the span of the bottom $t$ eigenvectors with $\norm{\tilde{u}} \leq \norm{u} \leq 1$ that is $2\sqrt{\gamma/(1-\lambda)}$-close to $u$.
Then the $c \sqrt{\gamma/(1-\lambda)}$-net that we construct is guaranteed to find some $\hat{u}$ that is $c \sqrt{\gamma/(1-\lambda)}$-close to $\tilde{u}$ and thus $(2+c)\sqrt{\gamma/(1-\lambda)}$-close to $u$.
Then, by~\Cref{lemma:independentsetrounding}, the returned independent set has size at least $\Paren{1/2-\gamma - (4+O(c))\gamma/(1-\lambda)}n$. The time complexity is polynomial in $n$ and $(\gamma/(1-\lambda))^{-O(t)}$.
\end{proof}

\subsection{Rounding tools}

First we state a general rounding algorithm that extracts an independent set from an approximate independent set.

\begin{lemma}[Rounding an approximate independent set]
    \label[lemma]{lemma: generalrounding}
    Let $G$ be an $n$-vertex graph and let $I \subseteq [n]$ be an independent set in it.
    Then, given a set $S\supseteq I$, there is a polynomial-time algorithm that returns an independent set $I'$ such that $|I'| \geq |I| - |S \setminus I|$.
\end{lemma}
\begin{proof}
    Note that $S\setminus I$ is a vertex cover of $G_S$.
    Hence, using the $2$-approximation algorithm for the minimum vertex cover problem, we can compute a vertex cover $C$ of $G_S$ in polynomial time such that $|C| \leq 2|S \setminus I|$.
    Then we output $I':=S\setminus C$, which by construction is an independent set with size at least $|S| - 2|S \setminus I|=|I| - |S \setminus I|$, finishing the proof.
\end{proof}

By applying \cref{lemma: generalrounding} to $k$ color classes, we can extract a coloring from an approximate coloring.

\begin{lemma}[Rounding an approximate $k$-coloring]
    \label[lemma]{cor: kcoloringrounding}
    Let $G$ be an $n$-vertex graph and let $\chi: [n] \to [k]$ be a $k$-partition of its vertices.
    Suppose that $G$ has a vertex cover of all edges with both endpoints in the same part of size at most $\e n$ with respect to $\chi$.
    Then, given $\chi^{-1}(1), \ldots, \chi^{-1}(k)$, there is a polynomial-time algorithm that returns a $(k, 2\e)$-coloring of $G$.
\end{lemma}

\begin{proof}
    We simply run the algorithm from \cref{lemma: generalrounding} on each set $\chi^{-1}(a)$.
    For each $\chi^{-1}(a)$, let $I_a$ be the largest independent set in it.
    Then we obtain an independent set $I'_a$ satisfying $|I'_a| \geq |I_a| - |\chi^{-1}(a) \setminus I_a|$.
    Then 
    \[\sum_{a=1}^k |I'_a| \geq \sum_{a=1}^k |I_a| - \sum_{a=1}^k |\chi^{-1}(a) \setminus I_a| \geq (1 - \e)n - \e n = (1-2\e)n\,.\]
\end{proof}

%% file: content/hardness.tex
\section{Hardness results}

In this section, we give various hardness results against algorithms that are given as input \emph{almost-colorable} graphs.
As in~\cite{bafna-hsieh-kothari}, our hardness reductions proceed from the following hardness result for approximating independent sets by Bansal and Khot~\cite{MR2648426-Bansal09}:

\begin{fact}[Hardness of approximating independent sets~\cite{MR2648426-Bansal09}]
    \label[fact]{thm:propA1}
Let $\e, \gamma > 0$ be any small enough constants.
Assuming the Unique Games Conjecture, given an $n$-vertex graph with maximum degree $o(n)$, it is NP-hard to decide between
\begin{itemize}
    \item the graph contains two disjoint independent sets of size $\Paren{\frac{1}{2}-\e}n$,
    \item the graph contains no independent set of size $\gamma n$.
\end{itemize}
\end{fact}

\cite{bafna-hsieh-kothari} use~\Cref{thm:propA1} to show that it is hard to find an independent set of size $\gamma n$ in a regular one-sided expander graph that admits a $(4, \e)$-coloring.

Our main hardness result concerns graphs colorable according to model matrices with repeated rows (see \cref{def:model-coloring}).
In particular, we prove that under the Unique Games Conjecture it is hard to color all but a small fraction of the vertices of one-sided expander graphs corresponding to models with repeated rows.
Our hardness result is in fact more fine-grained and matches the algorithmic result in \cref{thm:kalg-partial}.

\begin{theorem}[Fine-grained hardness of models]
    \label{thm: finegrainedrowshard}
    Let $M \in \R_{\geq 0}^{k \times k}$ be a reversible row stochastic matrix with zero on the diagonal and non-positive second eigenvalue.
    Let $S_1, \ldots, S_{k'}$ be the partition of $[k]$ such that, for all $a, b \in S_i$ the rows $M^{(a)}$ and $M^{(b)}$ of $M$ are equal, and for all $a \in S_i, b \in S_j$ with $i \neq j$ the rows $M^{(a)}$ and $M^{(b)}$ of $M$ are distinct.
    For all $i \in [k']$, let $a_i^* = \argmax_{a \in S_i} \pi_a$, and let $p_i = \max\Paren{0, \pi_{a_i^{*}} - \sum_{\substack{a \neq a_i^{*} \in S_i}} \pi_a}$.
    Then, assuming the Unique Games Conjecture, the following problem is NP-hard for all $\e, \gamma, \lambda > 0$ small enough:
    Given a regular $\lambda$-one-sided expander graph $G$ that admits an $(M, \e)$-coloring, find a $(k, 1 - \sum_{i \in [k']} p_i - \gamma)$-coloring.
  \end{theorem}

As an immediate corollary, we obtain the following hardness result for coloring all but a small fraction of the vertices of one-sided expander graphs corresponding to models with repeated rows:

\begin{corollary}[Hardness of models with repeated rows]
  \label[corollary]{thm: identicalrowshard}
  Let $M \in \R_{\geq 0}^{k \times k}$ be a reversible row stochastic matrix with zero on the diagonal and non-positive second eigenvalue, and let $\pi$ be its stationary distribution.
  Suppose \(M\) has repeated rows. %
  Then, assuming the Unique Games Conjecture, there exists $\delta > 0$ such that the following problem is NP-hard, for all $\e, \lambda > 0$:
  Given a regular $\lambda$-one-sided expander graph $G$ that admits an $(M, \e)$-coloring, find a $(k, \delta)$-coloring.
\end{corollary}

We deduce two more corollaries of \cref{thm: finegrainedrowshard}. First, assuming the UGC, given an almost-$3$-colorable graph, it is NP-hard to find a proper coloring on most of its vertices.
Hence, the dependence of \cref{thm:3alg} on an upper bound on the color class sizes is necessary.
Technically, this result is captured by \cref{thm: finegrainedrowshard} by noticing that, without an upper bound on the color class sizes, two rows of the corresponding model matrix can be identical.

\begin{corollary}[Hardness of unbalanced $3$-coloring]
\label[corollary]{thm: unbalancedhard}
    Let $\e, \gamma > 0$ be any small enough constants.
    Assuming the Unique Games Conjecture, given an $n$-vertex regular graph with $\lambda_2=o(1)$ that admits a $(3, \e)$-coloring, it is NP-hard to find a $(3, 1/2-\gamma)$-coloring.
\end{corollary}

The second corollary of \cref{thm: finegrainedrowshard} is a generalization of Proposition A.2 in \cite{bafna-hsieh-kothari} to any composite $k$. We show that, given an almost-$k$-colorable graph with identical color class sizes, assuming the UGC, it is NP-hard to find a linear-sized independent set in the graph.\footnote{It is easy to extend the lower bound to all $k \geq 4$ (not only composite $k$) if, instead of requiring the color class sizes to be identical, we merely require them to be lower bounded by $\Omega(1/k)$.}

\begin{corollary}[Generalization of Proposition A.2 in \cite{bafna-hsieh-kothari}]\label[corollary]{thm: kgeq4hard}
    Let $\e, \gamma > 0$ be any small enough constants and let $k$ be composite.
    Assuming the Unique Games Conjecture, given an $n$-vertex regular graph with $\lambda_2=o(1)$ that admits a $(k,\e)$-coloring with identical color class sizes, it is NP-hard to find a $(k, 1-\gamma)$-coloring.
\end{corollary}

In our final hardness result, we prove that if we only assume small $\lambda_3$, as opposed to small $\lambda_2$, then assuming the UGC, given an almost-$3$-colorable one-sided expander graph with identical color class sizes, it is NP-hard to find a proper coloring on most of the vertices.

\begin{theorem}[Hardness of balanced $3$-coloring with small $\lambda_3$]
\label{thm: lowtopthresholdrankhard}
    Let $\e, \gamma > 0$ be any small enough constants.
    Assuming the Unique Games Conjecture, given an $n$-vertex regular graph with $\lambda_3=o(1)$ that admits a $(3, \e)$-coloring with identical color class sizes, it is NP-hard to find a $(3, 2/9-\gamma)$-coloring.
\end{theorem}

We now turn to the proofs of our hardness results. To prove \cref{thm: finegrainedrowshard}, we first give a helper lemma showing that suitable blow-up graphs are expanding.

\begin{lemma}[Blow-up graphs are expanding]\label[lemma]{lem: blowupexpanding}
   Let $M \in \R_{\geq 0}^{k \times k}$ be a reversible row stochastic matrix with zero on the diagonal and non-positive second eigenvalue.
    Let $G$ be an $n$-vertex $d$-regular graph with adjacency matrix $A$, and let $\tilde{A} = \frac{1}{d} A$. 
    Suppose there exists a $k$-partition $\chi: [n] \to [k]$ such that $\Norm{M-M(\tilde{A},\chi)}=o(1)$.
    For $x \in [n]$ and $a \in [k]$, let $\D{x}{a} = \sum_{y \in \chi^{-1}(a)} \tilde{A}_{xy}$.
    For $a \neq b \in [k]$, let $G^{(ab)}$ be the bipartite graph induced between the vertices in $\chi^{-1}(a)$ and the vertices in $\chi^{-1}(b)$, and for $a \in [k]$ let $G^{(aa)}$ be the graph induced by the vertices in $\chi^{-1}(a)$.
    Suppose
    \begin{itemize}
        \item for all $a \in [k]$: $|\chi^{-1}(a)| \geq \Omega(n)$,
        \item for all $a, b \in [k]$: either $\lambda_1\Paren{\frac{1}{d}A_{G^{(ab)}}} = o(1)$, or $G^{(ab)}$ is biregular and $\lambda_2\Paren{\frac{1}{d}A_{G^{(ab)}}} = o(1)$.
    \end{itemize}
    Then $\lambda_2(\tilde{A})=o(1)$.
\end{lemma}

\begin{proof}
    For any $a,b\in[k]$, let $E^{(ab)}$ be the $n\times n$ matrix with entries $E^{(ab)}_{xy}=\tilde{A}_{xy}-\frac{M(\tilde{A},\chi)_{\chi(x)\chi(y)}}{\Abs{\chi^{-1}(\chi(y))}}$ for $\Set{\chi(x), \chi(y)}=\Set{a, b}$ and $E^{(ab)}_{xy}=0$ otherwise.
    Using the definitions of the partition and model matrices in \cref{def:partition_Z} and \cref{def:model}, we decompose $\tilde{A}$ as
    $$\tilde{A}=ZMB^{-1}Z^T+Z\Paren{M(\tilde{A},\chi)-M}B^{-1}Z^T+\sum_{1\leq a\leq b\leq k} E^{(ab)}\,.$$

    Then we argue about each term separately:
    \begin{itemize}
        \item the non-zero eigenvalues of $M$ and $ZMB^{-1}Z^T$ are the same (see the proof of \cref{thm:part-recovery}), so because $\lambda_2(M)\leq 0$ we also have $\lambda_2(ZMB^{-1}Z^T)\leq 0$,
        \item because $\Norm{M-M(\tilde{A},\chi)}=o(1)$ and $\Norm{Z}^2=\Norm{B}$ we have
    $$\Norm{Z\Paren{M(\tilde{A},\chi)-M}B^{-1}Z^T}\leq \Norm{Z}\Norm{M(\tilde{A},\chi)-M}\Norm{B^{-1}}\Norm{Z^T}=\frac{\max_{a\in[k]}\Abs{\chi^{-1}(a)}}{\min_{a\in[k]}\Abs{\chi^{-1}(a)}}\ o(1)=o(1)\,,$$
        \item if for some $1\leq a \leq b\leq k$ we have $\lambda_1\Paren{\frac{1}{d}A_{G^{(ab)}}} = o(1)$, then as $E^{(ab)}$ is obtained from $\frac{1}{d}A_{G^{(ab)}}$ by row-centering (which is an orthogonal projection) and padding with $0$s, we also get $\lambda_1(E^{(ab)})=o(1)$,
        \item if for some $1\leq a \leq b\leq k$ we have that $G^{(ab)}$ is biregular and $\lambda_2\Paren{\frac{1}{d}A_{G^{(ab)}}} = o(1)$, then row-centering sends the leading eigenvalue pair to $0$ while the others remain unchanged, so the spectrum of $E^{(ab)}$ is $\Set{0, 0}\cup \lambda_{\btw{2}{\Paren{\chi^{-1}(a)+\chi^{-1}(b)-1}}}\Paren{\frac{1}{d}A_{G^{(ab)}}}$, implying $\lambda_1(E^{(ab)})=o(1)$.
    \end{itemize}
    Hence, as all terms are symmetric, we can use Weyl's inequality to get $\lambda_2(\tilde{A})=o(1)$, as desired.
\end{proof}

We now prove \cref{thm: finegrainedrowshard}.

\begin{proof}[Proof of \cref{thm: finegrainedrowshard}]
    For the sake of contradiction, suppose an algorithm $\ALG$ solves the problem above efficiently and returns a $(k, 1 - \sum_{i \in [k']} p_i - \gamma)$-coloring.
    Let $r_1, \ldots, r_k$ be even integers bounded by $O_{M, \e, \gamma}(1)$\footnote{By \cref{lemma:model-color-size} we have $\pi_i=\Theta(1)$ for all $i\in[k]$.} such that $\Norm{\hat{\pi} - \pi}_{\infty}\leq \min\Paren{\frac{\gamma}{2k}, \frac{\e}{2}\min_{i\in [k]}\pi_i}$ and $\argmax_{a \in S_i} \pi_a=\argmax_{a \in S_i} \hat{\pi}_a$ for all $i\in[k']$, where $\hat{\pi}_i:=\frac{r_i}{\sum_{i\in[k]} r_i}$.
    Let $\widehat{M}$ be the $k\times k$ matrix defined by $\widehat{M}_{ij}=\frac{\hat{\pi}_j}{\pi_j}M_{ij}$, satisfying $\hat{\pi}_i\widehat{M}_{ij}=\hat{\pi}_j\widehat{M}_{ji}$.
    Finally, let \( f : [k'] \to [k] \) be a choice function selecting an element from each set \( S_i \), i.e., \( f(i) \in S_i \) for all \( i \in [k'] \).
    
    Then, we show that the following derived algorithm solves the problem in \cref{thm:propA1}:

    \begin{algorithm}[H]
        \caption{Fine-grained model hardness algorithm}
        \label{alg: finegrainedrowshard}
        \DontPrintSemicolon
        \textbf{Input:} A $2n$-vertex graph $G$ with maximum degree $d_{\max} = o(n)$ from \cref{thm:propA1} with parameters $\e/2$ and $\frac{\gamma}{2k}$\;
        For each $i \in [k']$ construct a graph $G_i$ as follows:\;
        \quad If $r_{a_i^*} \leq \sum_{\substack{a \neq a_i^{*} \in S_i}} r_a$, let $G_i$ have $\frac{1}{2} \sum_{a \in S_i} r_a$ disjoint copies of $G$\;
        \quad Else, let $G_i$ have $\sum_{a \neq a_i^* \in S_i} r_a$ disjoint copies of $G$ and let it also have $\Paren{r_{a_i^*} - \sum_{a \neq a_i^* \in S_i} r_a}n$ other isolated vertices\;
        For each $1\leq i<j\leq k'$ add a random $(\widehat{M}_{f(i)f(j)}n, \widehat{M}_{f(j)f(i)}n)$-biregular Ramanujan graph $G_{ij}$ between $V(G_i)$ and $V(G_j)$\;
        Decide based on whether running $\ALG$ on the final graph $G^*=\Paren{\bigcup_{i=1}^{k'} G_i} \cup \Paren{\bigcup_{1 \leq i < j \leq k'} G_{ij}}$ returns a $(k, 1 - \sum_{i \in [k']} p_i - \gamma)$-coloring\;
    \end{algorithm}

    Note that by an edge removal step analogous to that in \cref{alg: lowtopthresholdrankhard}, $G^*$ can be made $\Theta(d)$-regular.

    Suppose $G$ has two disjoint independent sets of size $(1-\e/2)n$ each. Then, by construction $G^*$ admits a $(k, \e)$-coloring $\chi$ such that $\Abs{\chi^{-1}(a)}=\hat{\pi}_a \Abs{V(G^*)}$ for all $a\in[k]$. Indeed, if for some $i\in[k']$ we have $r_{a_i^*} \leq \sum_{\substack{a \neq a_i^{*} \in S_i}} r_a$, then the corresponding color classes can be mapped to independent sets in copies of $G$ such that the two independent sets within any copy of $G$ get a different color. If $r_{a_i^*} > \sum_{\substack{a \neq a_i^{*} \in S_i}} r_a$, the same holds after mapping the isolated vertices to the color class corresponding to $a_i^*$.
    
    Let $\tilde{A}^*$ be the normalized adjacency matrix of $G^*$. Let us show that $\Norm{M(\tilde{A}^*, \allowbreak\chi)-M}_{\max}\leq \e$. Let $x\in S_i$ and $y\in S_j$ be arbitrary such that $f(i)=a$ and $f(j)=b$. By the expander mixing lemma,  we have $M(\tilde{A}^*, \chi)_{xy}=\widehat{M}_{ab}\pm o(1)$. On the other hand, using $\Norm{\hat{\pi} - \pi}_{\infty}\leq \frac{\e}{2} \min_{i\in [k]}\pi_i$ it follows that $\Abs{\widehat{M}_{ab}-M_{ab}}=\Abs{\frac{\hat{\pi}_b}{\pi_b}M_{ab} - M_{ab}}\leq \frac{\e}{2} M_{ab}\leq \e/2$. Combining the two by triangle inequality and using $M_{xy}=M_{ab}$ we get $\Abs{M(\tilde{A}^*, \chi)_{xy}-M_{xy}}\leq \e$ as desired.

    Hence, in this case $G^*$ would admit an $(M, \e)$-coloring. Furthermore, as the maximum degree within each $G_{S_j}$ is $o(d)$ and the subgraphs between the $G_{S_j}$s are biregular Ramanujan, we can apply \cref{lem: blowupexpanding} to get $\lambda_2(G^*)=o(1)$. Overall we can conclude that the preconditions of $\ALG$ would be satisfied, and by assumption it would return a $(k, 1 - \sum_{i \in [k']} p_i - \gamma)$-coloring of $G^*$.

    Now for the contrary, assume a $(k, 1 - \sum_{i \in [k']} p_i - \gamma)$-coloring of $G^*$. For $j\in[k']$, let
    
    $$q_j:=\max\Paren{0, r_{a_j^*} - \sum_{\substack{a \neq a_j^{*} \in S_j}} r_a}=\max\Paren{0, \hat{\pi}_{a_j^*} - \sum_{\substack{a \neq a_j^{*} \in S_j}} \hat{\pi}_a} \sum_{i\in[k]} r_i\,.$$
    
    Note that there are $\frac{\Paren{\sum_{i\in[k]} r_i-\sum_{j\in[k']} q_j}}{2}$ copies of $G$ each of size $2n$. Then, by the pigeonhole principle over the copies of $G$ and the colors, using $\Norm{\hat{\pi} - \pi}_{\infty}\leq \gamma/2k$, there must be an independent set of size at least 
    
    \begin{align*}
        &\frac1{k\frac{\Paren{\sum_{i\in[k]} r_i-\sum_{j\in[k']} q_j}}{2}}\Paren{\sum_{i\in[k]} r_i-\sum_{j\in[k']} q_j-\Paren{1-\sum_{j\in[k']} p_j-\gamma}\sum_{i\in[k]} r_i}n\\
        &= \frac{2n\sum_{i\in[k]} r_i}{k\Paren{\sum_{i\in[k]} r_i-\sum_{j\in[k']} q_j}}\Paren{\gamma-\sum_{j\in[k']} \Paren{\max\Paren{0, \hat{\pi}_{a_j^*} - \sum_{\substack{a \neq a_j^{*} \in S_j}} \hat{\pi}_a}-\max\Paren{0, \pi_{a_j^{*}} - \sum_{\substack{a \neq a_j^{*} \in S_j}} \pi_a}}}\\
        &\geq  \frac{2n}{k}\Paren{\gamma-k\frac{\gamma}{2k}}\\
        &=\frac{\gamma}{2k} \Abs{V(G)}
    \end{align*}
    in one of the copies of $G$, finishing the proof.
\end{proof}

We now provide proofs for the two corollaries of \cref{thm: finegrainedrowshard}.

\begin{proof}[Proof of \cref{thm: unbalancedhard}]
    Using \cref{lemma: 3coloringfixedaveragedegree}, consider the unique model $M$ corresponding to a stationary distribution with $\pi_1 = 1/2$, $\pi_2 = 1/4$, and $\pi_3=1/4$. The spectrum of $M$ is $\Set{1, 0, -1}$, so it has non-positive second eigenvalue. Noting that $\pi_2=\pi_3$ and that the rows of $M$ corresponding to $\pi_2$ and $\pi_3$ are identical.

    Therefore, it follows from \cref{thm: finegrainedrowshard} that it is NP-hard under the UGC to find a $(3, 1/2-\gamma)$-coloring of $\lambda$-one-sided expander graphs that admit an $(M, \epsilon)$-coloring.
\end{proof}

\begin{proof}[Proof of \cref{thm: kgeq4hard}]
    As $k$ is composite, it can be written as $k=k'\cdot q$ for some integers $k', q\geq 2$. Let $S_1, \ldots, S_{k'}$ be a partition of $[k]$ into groups of size $q$ each.
    Let $f: [k] \to [k']$ map each $i \in [k]$ to $x \in [k']$ such that $i \in S_x$.
    Let us define the model matrix $M \in \R^{k \times k}$ with zero on the diagonal and with entries
    
    \[
    M_{ij}=
    \begin{cases}
    \dfrac1{k-q}, & i\in S_{f(i)},\;j\notin S_{f(i)},\\[4pt]
    0,                                           & \text{otherwise},
    \end{cases}
    \]
    
    for $i \neq j \in [k]$. It is easy to verify that $M$ is row stochastic and reversible with stationary distribution $\pi = \frac1{k}\1\in \R^k$. The spectrum of $M$ is $\Set{1,\;0\ (\text{mult.\ }k-k'),\,-\tfrac1{k'-1}\ (\text{mult.\ }k'-1)}$ so it has non-positive second eigenvalue. By construction, for each $x \in [k']$, for all $i, j \in S_x$ the rows $M^{(i)}$ and $M^{(j)}$ of $M$ are equal, and for all $i \in S_x, j \notin S_x$ the rows $M^{(i)}$ and $M^{(j)}$ of $M$ are distinct.

    For all $x \in [k']$, let $a_x^* := \argmax_{a \in S_x} \pi_a$, and let $p_x := \max\Paren{0, \pi_{a_x^{*}} - \sum_{\substack{a \neq a_x^{*} \in S_x}} \pi_a}$. For all $x\in [k']$, using $q\geq 2$ we have $\Abs{S_x}\geq 2$, so $p_x=0$. 

    Therefore, as $1-\sum_{x\in [k']} p_x-\gamma=1-\gamma$, it follows from \cref{thm: finegrainedrowshard} that it is NP-hard under the UGC to find a $(k, 1-\gamma)$-coloring of $\lambda$-one-sided expander graphs that admit an $(M, \epsilon)$-coloring.
\end{proof}

We finish this section with the proof of \cref{thm: lowtopthresholdrankhard}.

\begin{proof}[Proof of \cref{thm: lowtopthresholdrankhard}]
    For the sake of contradiction, suppose an algorithm $\ALG$ solves the problem above efficiently.
    Then, we show that the following derived algorithm solves the problem in \cref{thm:propA1}:

    \begin{algorithm}[H]
        \caption{Low top-threshold rank hardness algorithm}
        \label{alg: lowtopthresholdrankhard}
        \DontPrintSemicolon
        \textbf{Input:} The $n$-vertex graph $G$ with average degree $d_{\avg} = o(n)$ and maximum degree $d_{\max} = o(n)$ from \cref{thm:propA1} with parameters $\e, \frac{3}{2}\gamma$\;
        Insert into the graph $\frac{n+\sqrt{n(n-4d_{\avg})}}{2}$ new vertices and denote their set by $S$\;
        Add edges $H$ between $V(G)$ and $S$ to form a complete bipartite graph\;
        Let $H'$ be a bipartite graph between $V(G)$ and $S$ with degrees $d_G(x)$ for $x\in V(G)$ and $\frac{nd_{\avg}}{\Abs{S}}$ for $y\in S$\;
        Remove $H'$ from $H$ to obtain a graph $G'$\;
        Let $G''$ be a new $\Abs{S}$-regular tripartite graph with $\lambda_2(A_{G''})=o(\Abs{S})$ and sides of size $\Abs{S}, \Abs{S}, \Paren{\frac{1}{2} - \e}n\;$\footnotemark{}\;
        Decide based on whether running $\ALG$ on the final graph $G^* = G' \cup G''$ returns a $\Paren{3, \frac{2}{9}-\gamma}$-coloring of $G^*$\;
    \end{algorithm}
    \footnotetext{Using \cref{lem: blowupexpanding} and $\Abs{S}\geq n/4$, such a graph exists based on the model matrix \(\left(\begin{smallmatrix}
0 & \Abs{S} - (1/4-\e/2)n & (1/4-\e/2)n \\
\Abs{S} - (1/4-\e/2)n & 0 & (1/4-\e/2)n \\
å\Abs{S}/2 & \Abs{S}/2 & 0
\end{smallmatrix}\right)\) with a non-positive second eigenvalue.}

    Let us first show that $G^*$ is $\Abs{S}$-regular. Note that the degrees of $x\in V(G)$ coming from $G$ are offset by the removal of $H'$, so they will exactly have degree $\Abs{S}$ coming from $H$. For a vertex $y\in S$, the degree coming from $H$ is $n$, and the degree removed by $H'$ is $\frac{nd_{\avg}}{\Abs{S}}$, so the degree of $y$ in $G^*$ is also
        
    $$n-\frac{nd_{\avg}}{\Abs{S}}=n-\frac{nd_{\avg}}{\frac{n+\sqrt{n(n-4d_{\avg})}}{2}}=\frac{n+\sqrt{n(n-4d_{\avg})}}{2}=\Abs{S}\,.$$

    Decomposing $A_{G'}=A_G+A_H-A_{H'}$ and using that $\lambda_2(A_H) = 0$ and that $G$ and $H'$ both have maximum degree at most $d_{\max} = o(n)$, we get by  Weyl's inequality that
    \[\lambda_2(A_{G'}) \leq \lambda_1(A_G)+\lambda_2(A_H)-\lambda_n(A_{H'}) = o(n)\,.\]
    Using $d_{\avg}=o(n)$ we have $|S| = \Omega(n)$, so $\lambda_2(A_{G'}) = o(|S|)$. Then, decomposing $A_{G^*}=A_{G'}+A_{G''}$ and using that $\lambda_2(A_{G''})=o(|S|)$ by construction, we finally get by Weyl's inequality that
    \[\lambda_3(A_{G^*}) \leq \lambda_2(A_{G'})+\lambda_2(A_{G''}) =o(|S|)\,.\]

    If $G$ has two disjoint independent sets of size $\Paren{\frac{1}{2}-\e}n$, then $G^*$ admits a \tcl with identical color class sizes on all but $2\e n \leq \e |V(G^*)|$ vertices.
    Hence, using the $|S|$-regularity of both $G'$ and $G''$ and that $\lambda_3(G^*)=o(1)$, $\ALG$ returns a $\Paren{3, \frac{2}{9}-\gamma}$-coloring of $G^*$.

    However, we show that if $G^*$ admits a $\Paren{3, \frac{2}{9}-\gamma}$-coloring of $G^*$, then $G$ has an independent set of size $\frac{3}{2}\gamma n$, and $\ALG$ is a distinguisher.
    Indeed, note that a $\Paren{3, \frac{2}{9}-\gamma}$-coloring of $G^*$ implies an independent set in $G$ of size at least
    $$\frac{1}{3}\Bigg(\Paren{\frac{7}{9}+\gamma}\Abs{V(G^*)}-\Paren{\Abs{S}+|V(G'')|}\Bigg)\,,$$
    which after plugging in the size of the components, and using that $\Abs{S}\leq n$, is lower bounded by $\frac{3}{2} \gamma n$.
\end{proof}

%% file: content/threshold_rank.tex
\section{Proof of top threshold rank lower bound}
\label{sec:rank-lb}

In this section, we prove for the sake of completeness a slight generalization of Lemma 6.1 in \cite{barak-raghavendra-steurer}, such that it extends to all bounded-norm symmetric matrices.
The proof is essentially the same as the original proof.
We note that the original proof has a minor mistake in an application of Cauchy-Schwarz, and to fix it we need to modify the guarantee of the lemma to give a threshold rank lower bound of $(1-1/C)^2r$ instead of the original lower bound of $(1-1/C)r$.

\begin{lemma}[Generalized restatement of Lemma 6.1 in \cite{barak-raghavendra-steurer}]
    Let $A \in \R^{n \times n}$ be symmetric with $\Norm{A} \leq 1$.
    Suppose there exists $M \in \R^{n \times n}$ positive semidefinite such that
    \begin{equation*}
        \iprod{A, M} \geq 1-\e \,, \quad
        \Norm{M}_F^2 \leq \frac{1}{r} \,, \quad
        \Tr(M) = 1 \,.
    \end{equation*}
    Then for all $C > 1$,
    \begin{equation*}
        \rank_{\geq 1- C \cdot \e} (A) \geq \Paren{1-\frac{1}{C}}^2 r \,.
    \end{equation*}
\end{lemma}

\begin{proof}
    Let the spectral decomposition of $A$ be $A = \sum_{i} \lambda_i u_i u_i^{\top}$ and let $k = \rank_{\geq 1- C \e} (A)$.
    We use $\Lambda_{\geq 1- C\e} = \sum_{i=1}^{k} u_i u_i^{\top}$ to denote the subspace spanned by the eigenvectors with corresponding eigenvalues at least $1- C\e$. 
    Then we can rewrite $\iprod{A, M}$ as
    \begin{align}
    \label{eq:BRS11_proof_1}
    \begin{split}
        \iprod{A, M}
        & = \sum_{i=1}^{k} \lambda_i \iprod{u_i u_i^{\top}, M} + \sum_{i=k+1}^{n} \lambda_i \iprod{u_i u_i^{\top}, M} \\
        & \leq \iprod{\Lambda_{\geq 1- C\e}, M} + (1- C\e) \iprod{\Id_n - \Lambda_{\geq 1- C\e}, M} \\
        & =  C \e \iprod{\Lambda_{\geq 1- C\e}, M} + (1- C\e) \iprod{\Id_n, M} \,.
    \end{split}
    \end{align}
    For the first term, it follows by the assumption $\Norm{M}_F^2 \leq \frac{1}{r}$ and Cauchy–Schwarz that
    \begin{equation}
    \label{eq:BRS11_proof_2}
        \iprod{\Lambda_{\geq 1- C\e}, M}
        \leq \Norm{\Lambda_{\geq 1- C\e}}_F \Norm{M}_F
        \leq  \sqrt{k} \cdot \frac{1}{\sqrt{r}}\,.
    \end{equation}
    For the second term, it follows by the assumption $\Tr(M) = 1$ that
    \begin{equation}
    \label{eq:BRS11_proof_3}
        \iprod{\Id_n, M} = \Tr(M) = 1 \,.
    \end{equation}
    Plugging \cref{eq:BRS11_proof_2} and \cref{eq:BRS11_proof_3} into \cref{eq:BRS11_proof_1}, it follows that
    \begin{equation*}
        \iprod{A, M}
        \leq C \e \sqrt{\frac{k}{r}} + 1- C\e \,.
    \end{equation*}
    Since by assumption $\iprod{A, M} \geq 1-\e$, it follows by rearranging the terms that
    \begin{equation*}
        \rank_{\geq 1- C \cdot \e} (A)
        = k
        \geq \Paren{1-\frac{1}{C}}^2 r \,.
    \end{equation*}
\end{proof}

%% file: content/random_planting.tex
\section{Full coloring in the random planted model} \label{app: plantedcoloring}

This part is dedicated to proving \cref{thm: random-planting-full}. We start with the key first step that guarantees the recovery of a \emph{list} of candidate $k$-colorings, such that at least one of them is close to the planted $k$-coloring.

\begin{theorem}[Result for random planting, partial recovery of a list]
    \label{thm: random-planting-partial-list}
    Let $H$ be an $n$-vertex $d$-regular graph with adjacency matrix $A_H$.
    Suppose $\rank_{\geq \lambda}(\frac{1}{d}A_H) \leq t$ for some $0 < \lambda < 1$ and $t \in \mathbb{N}$.
    Let $G$ be the graph obtained by randomly planting a \kcl in $H$, and let $\chi: [n] \to [k]$ be the associated \kcl.
    Then, for $\lambda > 0$ small enough, there exists an algorithm that runs in time $n^{O(1)} \cdot \Paren{O_{k}(1)}^{t}$ and outputs a list of $\Paren{O_{k}(1)}^{t}$ $k$-colorings $\hat{\chi}: [n] \to [k]$ such that with high probability, for at least one of them, it holds up to a permutation of the color classes that $\mathbb{E}_{\bm{x} \sim \operatorname{Unif}([n])} \1[\hat\chi(\bm{x}) \neq \chi(\bm{x})] \leq O_k(1/d^{\Omega(1)})$.
\end{theorem}
\begin{proof}
    The proof is identical to the proof of \cref{thm: random-planting-partial}, but omitting the last rounding step.
\end{proof}

To convert a coloring returned by \cref{thm: random-planting-partial-list} into a full coloring, we will need several helper lemmata. We start with the proofs of the two lemmata showing properties of \textbf{one}-sided expansion.

\begin{proof}[Proof of \cref{lemma: emlonesided}]
    Let $\1_S \in \{0,1\}^n$ be the indicator vector of $S$ and write $\1_S$ in the orthonormal eigenbasis $e_{[n]}$ of $A$:
    $$\1_S = \sum_{i=1}^n \alpha_i e_i\,.$$

    Hence, by orthonormality, we can write

    \begin{align*}
        &\Abs{E(H_S)}\\
        &=\frac{1}{2} \1_S^\top A \1_S\\
        &=\frac{1}{2} \Paren{\sum_{i=1}^n \alpha_i e_i^\top } \Paren{\sum_{i=1}^n \lambda_i(A) e_i e_i^\top } \Paren{\sum_{i=1}^n \alpha_i e_i}\\
        &=\frac{1}{2} \sum_{i=1}^n \lambda_i(A) \alpha_i^2
    \end{align*}

    Using $d$-regularity, we have $\lambda_1(A)=d$.
    Then, as $e_1=\frac{1}{\sqrt{n}}\1$, we get $\alpha_1=\Iprod{e_1, \1_S }=\frac{|S|}{\sqrt{n}}$, so $\lambda_1(A) \alpha_1^2=\frac{d|S|^2}{n}$.
    For the terms with $i>1$, as $\alpha_i^2\geq 0$, we can upper bound $\lambda_i(A) \alpha_i^2$ by $\lambda_2(A) \alpha_i^2\leq cd\alpha_i^2$.
    Plugging these back we get
    $$\Abs{E(H_S)}\leq \frac{1}{2}\Paren{\frac{d|S|^2}{n} + cd\sum_{i=2}^n \alpha_i^2}\,.$$

    Then, using $\sum_{i=2}^n \alpha_i^2\leq \sum_{i=1}^n \alpha_i^2=\Norm{\1_S}_2^2=|S|$ gives the desired bound.
\end{proof}

\begin{proof}[Proof of \cref{lemma: ssveonesided}]
    Let $T:=N(S)\cup S$.
    Then, as all edges touching $S$ are in $E(H_{T})$, we have
    $$e(H_{T})\geq \frac{1}{2}|S|d\,.$$ 

    On the other hand, by \cref{lemma: emlonesided}, writing $|T|=x|S|$ and using $|S|\leq \alpha n$, we get
    $$e(H_{T})\leq \frac{d|T|}{2}\Paren{\frac{|T|}{n}+c}\leq\frac{1}{2}x\Paren{c+\alpha x}|S|d\,.$$

    Comparing the two bounds, and solving the resulting quadratic inequality, we get 
    $$x\geq \frac{\sqrt{c^2 + 4\alpha}-c}{2\alpha}=\frac{2}{c+\sqrt{c^2+4\alpha}}\,.$$

    By using $\sqrt{a+b}\leq \sqrt{a}+\sqrt{b}$ and that $\Abs{N(S)\setminus S}\geq \Abs{N(S)}-|S|$ we get the desired result.
\end{proof}

Next, we need an algorithm that takes a coloring that agrees with most of the planted coloring and turns it into a coloring that leaves some of the vertices uncolored, but agrees with the planted coloring on the colored vertices. We define this notion of coloring formally now.

\begin{definition}[Partial coloring]
    \label[definition]{def: partialcoloring}
    Let $G$ be an $n$-vertex graph that admits a \kcl $\chi$.
    A {\em $(k, \beta)$-partial} coloring w.r.t.\ $\chi$ is a \kcl of at least $(1-\beta)n$ vertices that is identical to $\chi$ on these vertices.
    The remaining $\beta n$ vertices are left uncolored and are referred to as {\em free}. In the random planted model we always implicitly take $\chi$ to be the planted \kcl.
\end{definition}

Let us state now the corresponding algorithm and its guarantees.

\begin{algorithm}[H]
    As long as there is a vertex $x$ that has color $c_1$, and a color $c_2\neq c_1$ such that $x$ has less than $\frac{1}{6k}d$ neighbors with color $c_2$, uncolor $x$.%
    \protect\caption{\label{alg: uncoloring}Uncoloring algorithm}
\end{algorithm}

\begin{lemma}[Adaptation of Lemma B.20 of \cite{MR3536556-David16}]\label[lemma]{lemma: approximatetopartial}
    Let $\alpha\leq\frac{1}{24k}$ be a parameter. Let $H$ be a $d\geq O(1)$-regular $n$-vertex  graph with $\lambda_2(H)\leq \frac{1}{9k}$. 
    Let $G$ be the graph obtained from $H$ by randomly planting a \kcl $\chi$ in it. Given an \kcl $\hat{\chi}$ of $G$ with $\mathbb{E}_{\bm{x} \sim \operatorname{Unif}([n])} \1[\hat{\chi}(\bm{x}) \neq \chi(\bm{x})] \leq \alpha$, \cref{alg: uncoloring} \whp returns an $(k, 4\alpha)$-partial coloring of $G$.
\end{lemma}

To prove \cref{lemma: approximatetopartial}, we will need a definition and two lemmata: one upper bounding the number of atypical vertices with high probability, and a structural lemma showing the existence of a ``nice'' subgraph.

\begin{definition}\label[definition]{def: statisticallybad}
    Let $H$ be a $d$-regular graph. Let $G$ be the graph obtained from $H$ by randomly planting a \kcl $\chi$ in it. We call a vertex $x\in V(G)$ \textit{statistically bad}, if there is a color $c\neq \chi(x)$ such that $x$ has less than $\frac{1}{2k}d$ neighbors $y$ with $\chi(y)=c$. We denote the set of statistically bad vertices by $SB(G)$, writing $SB$ when it is clear from context which graph we are referring to.
\end{definition}

\begin{lemma}[Restatement of Lemma B.14 of \cite{MR3536556-David16}]\label[lemma]{lemma: sbupperbound}
   Let $H$ be a $d\geq O(1)$-regular $n$-vertex graph. Let $G$ be the graph obtained from $H$ by randomly planting a \kcl in it. It holds that $\Abs{SB(G)}\leq n2^{-\Omega(d)}$ \whp.
\end{lemma}

\begin{lemma}[Adaptation of Lemma B.19 of \cite{MR3536556-David16}]\label[lemma]{lemma: highmindegreesubgraph}
    Let $0<\alpha, c<1$ be parameters. Let $H$ be a $d$-regular $n$-vertex graph with $\lambda_2(H)\leq c$. For any $S\subseteq V$ with $\Abs{S}<\alpha n$, there exists $R\subseteq V\setminus S$ such that $\Abs{R}\geq \Paren{1-2\alpha}n$ and the minimum degree in $H_R$ is at least $(1-2\alpha-c)d$.
\end{lemma}

\begin{proof}
    Let $\e:=2\alpha + c$. Consider the following iterative procedure. Repeatedly remove a vertex $v_t$ from $V\setminus \Paren{S \cup \bigcup_{i=1}^{t-1} v_i}$ if $v_t$ has more than $\e d$ neighbors in $S_t := S \cup \bigcup_{i=1}^{t-1} v_i$. Assume towards a contradiction that this process stops after more than $t:=\Abs{S}$ steps. Then, by using the vertices in $S_t\setminus S$, we can lower bound the number of edges in $H_{S_t}$:

    $$e(H_{S_t})\geq \Abs{S_t\setminus S} \e d= \frac{1}{2}\Abs{S_t}\e d\,.$$

    On the other hand, by \cref{lemma: emlonesided},

    $$e(H_{S_t})\leq \frac{\Abs{S_t}d}{2}\Paren{\frac{\Abs{S_t}}{n} + c}\,.$$

    Comparing the two and using $\Abs{S}=\frac{1}{2}\Abs{S_t}$ we get

    $$\Abs{S}\geq\frac{1}{2}\Paren{\e-c}n=\alpha n$$

    which contradicts the upper bound assumption on $\Abs{S}$. Hence, letting $R:=V\setminus S_t$ it must hold that $\Abs{R}\geq n-2\Abs{S}\geq \Paren{1-2\alpha}n$ and the minimum degree in $H_R$ is at least $(1-\e)d=(1-2\alpha-c)d$ by construction.
\end{proof}

Now we are ready to prove \cref{lemma: approximatetopartial}.

\begin{proof}[Proof of \cref{lemma: approximatetopartial}]
    Let $c:=\frac{1}{9k}$. When $d$ is at least a large enough constant, $2^{-\Omega(d)}\leq \alpha$ holds. Let $B:=\Set{x\in V\mid \chi(x)\neq \hat{\chi}(x)}$ denote the set of vertices wrongly colored by $\hat{\chi}$. Note that $\Abs{B}\leq \alpha n$ by assumption, and $\Abs{SB(G)}\leq 2^{-\Omega(d)}\leq \alpha n$ by \cref{lemma: sbupperbound}. Hence, applying \cref{lemma: highmindegreesubgraph} on $H$ to $B\cup SB$ we get a set $R$ with $|R|=(1-4\alpha)n$ such that the minimum degree in $H_R$ is at least $\Paren{1-4\alpha-c}d$. Using $R\cap SB=\emptyset$, for any $x\in R$ and $c\neq \chi(x)$ we get that $x$ has degree at least $\Paren{\frac{1}{2k}-4\alpha-c}d\geq \frac{1}{6k}d$ into $\Set{y\in R\mid \chi(y)=c}$, that is, into the vertices in $R$ colored with $c$. It follows by induction that no vertices in $R$ will be uncolored at any step of \cref{alg: uncoloring}. Hence, the coloring returned has size at least $|R|=(1-4\alpha)n$. It only remains to show that the coloring returned is proper. In fact, we will show that all vertices in $B$ will get uncolored, from which the statement follows.

    For the sake of contradiction, assume there exist vertices $B_2\subseteq B$ that remain colored by the end of the process. Let $x\in B_2$ be any such vertex. As $x$ has never been uncolored, at the end of the process it has at least $\frac{1}{6k}d$ neighbors $y$ with $\hat{\chi}(y)=c\neq \hat{\chi}(x)$. In particular, choosing $c=\chi(x)$, let $Q$ denote those neighbors $y$ of $x$ with $\hat{\chi}(y)=\chi(x)$. As $x$ and $y$ are neighbors in $G$, we have $\chi(x)\neq \chi(y)$, so $\hat{\chi}(y)\neq \chi(y)$, which implies $Q\subseteq B_2$. Summarizing, any vertex in $G_{B_2}$ has degree at least $\frac{1}{3k}d$. Therefore, $$e(G_{B_2})\geq \frac{\frac{1}{6k}d|B_2|}{2}= \frac{d|B_2|}{12k}\,.$$

    On the other hand, by \cref{lemma: emlonesided}, $$e(G_{B_2})\leq e(H_{B_2})\leq \frac{d\Abs{B_2}}{2}\Paren{\frac{\Abs{B_2}}{n}+c}\,.$$
    
    Using $\Abs{B_2}>0$ and the bounds on $\alpha$ and $c$, we can combine and simplify the above two inequalities to get $$\Abs{B_2}\geq \Paren{\frac{1}{6k}-c}n> \alpha n\,.$$

    However, as $\Abs{B_2}\leq \Abs{B}\leq \alpha n$ by assumption, we get a contradiction, so $\Abs{B_2}=0$ finishing the proof.
\end{proof}

Next we need an algorithm that recolors some of the uncolored vertices in a way that guarantees that the remaining uncolored vertices split up into small components.

\begin{algorithm}[H]
    As long as there is an uncolored vertex $x$ that has neighbors in every color except for some color $c$, color $x$ using $c$. 
    \protect\caption{\label{alg: saferecoloring} Safe recoloring algorithm}
\end{algorithm}

\begin{lemma}[Adaptation of Lemma B.21 of \cite{MR3536556-David16}]\label[lemma]{lemma: partialtologncomps}
    Let $H$ be a $d$-regular $n$-vertex  graph with $\lambda_2(H)\leq \frac{1}{16k^2}$. 
    Let $G$ be the graph obtained from $H$ by randomly planting a \kcl $\chi$ in it. Starting from a $\Paren{k, \frac{1}{256k^4}}$-partial coloring $\hat{\chi}$ of $G$, let $U$ be the set of vertices left uncolored by \cref{alg: saferecoloring}. Then, \whp each component in $G_U$ has size at most $\ln{n}$.
\end{lemma}

\begin{proof}
    The proof closely follows that of Lemma B.21 of \cite{MR3536556-David16}, so we only provide the parts that are changed. Let us first prove a helper lemma similar to Claim B.22 of \cite{MR3536556-David16}:
    \begin{lemma}\label[lemma]{lemma: largecoloredneighborhoodprob}
        Let $C$ be any set of vertices of size $t_1$, let $D:=N_{H}\Paren{C}\setminus C$
        and let $t_2:=\Abs{D}$. If $t_2\ge 2k^2t_1$ then $C\subseteq U$ and $D\subseteq V\setminus U$ with probability at most $e^{-\frac{1}{2k}t_2}$.
    \end{lemma}
    \begin{proof}
        Partition the vertices of $D$ into disjoint subsets $\Set{D_{v}}$, where $v\in C$ and $D_{v}\subseteq N_{H}\Paren{v}$. Consider some $D_{v}$ and assume that $v$ is colored by $1$. If both $v\in U$ and $D_{v}\subseteq V\setminus U$ then the set of colors that vertices in $D_{v}$ are using must be contained in exactly one
        of the sets $[k]\setminus {c}$ for some $c\neq 1$ (otherwise \cref{alg: saferecoloring} would color $v$).
        Therefore the probability that both $v\in U$ and $D_{v}\subseteq V\setminus U$
        is at most $(k-1)\Paren{\frac{k-1}{k}}^{\Abs{D_{v}}}$. Hence, the probability that both $C$ is in $U$ and $D\subseteq V\setminus U$ is at most
        \begin{align*}
            \prod_{D_{v}\in\Set{ D_{v}} }(k-1)\Paren{\frac{k-1}{k}}^{\Abs{D_{v}}} & =(k-1)^{\Abs{\Set{ D_{v}} }}\Paren{\frac{k-1}{k}}^{t_2}\\
            & \le(k-1)^{t_1}\Paren{\frac{k-1}{k}}^{t_2}\\
            & =e^{(t_1+t_2)\ln(k-1)-t_2\ln{k}}\\
            & \le e^{-\frac{1}{2k}t_2}
        \end{align*}
        using $t_2\ge 2k^2t_1$ and $\ln{k}-\ln(k-1)\geq \frac1{k}$ on the last line.
    \end{proof}

    Letting $N_{H}\Paren{n,t_1,t_2}$ denote the number of connected components of size $t_1$ with a neighborhood of size exactly $t_2$ and letting $P_{t_1,t_2}:=\max_{C\subset H,\Abs{C}=t_1,\Abs{N\Paren{C}}=t_2}\Pr\Brac{\text{\ensuremath{C}\,\ is in\,\ensuremath{U}}}$, we get by union bound that the probability that there exists a connected component in $G_{U}$ of size $\Omega(\ln n)$ is at most
    
    \begin{align}
    \sum_{t_1=\ln n\,}^{\Abs{U}}\sum_{t_2=1}^{d t_1}N_{H}\Paren{n,t_1,t_2}P_{t_1,t_2} & =\sum_{t_1=\ln n\,}^{\Abs{U}}\sum_{t_2= 7k^2t_1 }^{d t_1}N_{G}\Paren{n,t_1,t_2}e^{-\frac{1}{2k} t_2}\label{eq:alpha1}\\
    & \le\sum_{t_1=\ln n\,}^{\Abs{U}}\sum_{t_2= 7k^2t_1 }^{d t_1}n\binom{t_1+t_2 }{ t_1}e^{-\frac{1}{2k} t_2}\label{eq:alpha2}\\
    & \le\sum_{t_1=\ln n\,}^{\Abs{U}}\sum_{t_2= 7k^2t_1 }^{d t_1}n\Paren{e\frac{t_1+t_2}{t_1}}^{t_1}e^{-\frac{1}{2k} t_2}\label{eq:alpha3}\\
    & \le\sum_{t_1=\ln n\,}^{\Abs{U}}ndt_1\max_{7k^2\le x\le d}e^{\Paren{1+\ln(x+1)-\frac{1}{2k} x}t_1}\nonumber \\
    & \le n^{4}\max_{7k^2\le x\le d}e^{\Paren{1+\ln(x+1)-\frac{1}{2k} x}\ln{n}}\nonumber\\
    & \le n^{-1}\,.\label{eq:alpha4}
    \end{align}

    Equality~\ref{eq:alpha1} follows by \cref{lemma: largecoloredneighborhoodprob} and that by \cref{lemma: ssveonesided}
    if $\lambda_2(H)\le \frac{1}{16k^2}$ then for a set of size $t_1\le\Abs{U}\le\frac{1}{256k^4}n$ it holds that its neighborhood is of size $t_2\ge7k^2t_1$. Inequality~\ref{eq:alpha2} follows by Lemma B.23 of \cite{MR3536556-David16}. Inequality~\ref{eq:alpha3} follows by the binomial inequality ${n \choose k}\le\Paren{e\frac{n}{k}}^{k}$. Inequality~\ref{eq:alpha4} follows by $\ln\Paren{7k^2+1}\leq \frac{7k}{2}-6$ which holds for $k\geq 3$. This completes the proof.
\end{proof}

We are now finally ready to prove \cref{thm: random-planting-full}:

\begin{proof}[Proof of \cref{thm: random-planting-full}]
    We can apply \cref{thm: random-planting-partial-list} to get a list $L$ of $O_{k}(1)$ $k$-colorings $\hat{\chi}: [n] \to [k]$ such that, for at least one of them, it holds up to a permutation of the color classes that $\mathbb{E}_{\bm{x} \sim \operatorname{Unif}([n])} \1[\hat\chi(\bm{x}) \neq \chi(\bm{x})] \leq O_k\Paren{1/d^{\Omega(1)}}$.

    Then, running \cref{alg: uncoloring} on all colorings in $L$ returns a list of colorings $L'$. For $d$ at least a large enough constant compared to $k$, \cref{lemma: approximatetopartial} guarantees that the coloring $\hat{\chi}$ in $L$ with $\mathbb{E}_{\bm{x} \sim \operatorname{Unif}([n])} \1[\hat\chi(\bm{x}) \neq \chi(\bm{x})] \leq O_k\Paren{1/d^{\Omega(1)}}$ turns into a $\Paren{k, O_k\Paren{1/d^{\Omega(1)}}}$-partial coloring in $L'$. 
    
    Then, running \cref{alg: saferecoloring} on all colorings in $L'$ returns a list of colorings $L''$. For $d$ at least a large enough constant compared to $k$, \cref{lemma: partialtologncomps} guarantees that the $\Paren{k, O_k\Paren{1/d^{\Omega(1)}}}$-partial coloring in $L'$ turns into a coloring in $L''$ such that all the remaining still uncolored vertices are in connected components of size $O(\ln{n})$.

    Hence, we can try a brute force search on each component of those colorings in $L''$ that have component sizes $O(\ln{n})$. By the above, we are guaranteed to find a full coloring for at least one of them. This finishes the proof.
\end{proof}